\pdfoutput=1
\documentclass[11pt]{article}
\usepackage{amsmath,amssymb,amsthm}
\usepackage[shortlabels]{enumitem}
\usepackage{graphicx}
\usepackage[round]{natbib}
\usepackage{setspace}
\usepackage{float}

\def\startlocaldefs{}
\def\endlocaldefs{}
\def\paperwithoutbackrefs{}
\usepackage{mathtools}
\usepackage{bbm}
\usepackage{booktabs}
\usepackage{xcolor}
\definecolor{deepblue}{rgb}{0,0,0.35}
\ifdefined\paperwithoutbackrefs
\usepackage[
  colorlinks=true,
  hypertexnames=false,
  linkcolor=deepblue,
  citecolor=deepblue,
  urlcolor=deepblue
]{hyperref}
\else
\usepackage[
  colorlinks=true,
  hypertexnames=false,
  linkcolor=deepblue,
  citecolor=deepblue,
  urlcolor=deepblue,
  pagebackref
]{hyperref}
\fi
\usepackage{tikz}
\usetikzlibrary{positioning}
\usepackage{placeins}
\usepackage[normalem]{ulem}

\startlocaldefs

\def\bs{\boldsymbol}

\DeclareMathOperator{\E}{\mathbb{E}}

\newcommand{\crossref}[1]{\ref{#1}}
\newcommand{\crosseqref}[1]{\eqref{#1}}

\newcommand{\fignote}[2][\fignoteinset]{%
  \par
  {\footnotesize
   \raggedright
   \leftskip=#1\relax
   \rightskip=#1 plus 1fil\relax
   \parindent=0pt\relax
   Notes: #2\par}}

\theoremstyle{plain}

\newtheorem{proposition}{Proposition}
\newtheorem{corollary}{Corollary}

\theoremstyle{definition}

\newtheorem{assumption}{Assumption}
\newtheorem{assumptionNA}{Assumption}

\newtheorem{assumptionCS}{Assumption}

\newtheorem{assumptionTI}{Assumption}

\newtheorem{assumptionTIC}{Assumption}

\newtheorem{assumptionPT}{Assumption}

\newtheorem{assumptionCPT}{Assumption}

\newtheorem{assumptionTILT}{Assumption}

\newtheorem{assumptionCSLT}{Assumption}

\newtheorem{assumptionRSM}{Assumption}

\newtheorem{assumptionM}{Assumption}

\newtheorem{assumptionMk}{Assumption}

\newtheorem{assumptionMLE}{Assumption}

\newtheorem{assumptionID}{Assumption}

\newtheorem{assumptionIDk}{Assumption}

\theoremstyle{definition}
\newtheorem{remark}{Remark}

\newcommand{\papertitle}{Event-Study Designs for Discrete Outcomes with Latent Transition Heterogeneity}

\newcommand{\papertitlenote}{This paper supersedes earlier versions of the manuscript titled ``Event Studies for Discrete Outcomes with Latent Transition Heterogeneity.'' We thank Karun Adusumilli, Stéphane Bonhomme, Brantly Callaway, Clément de Chaisemartin, Xu Cheng, Riccardo Di Francesco, Daniel Haanwinckel, Jinyong Hahn, James Heckman, Peter Hull, Dongwoo Kim, Soonwoo Kwon, Philip Marx, Giovanni Mellace, Jonathan Roth, Yuya Sasaki, Frank Schorfheide, Jeffrey Smith, Petra Todd, Jaume Vives-i-Bastida, Yuanyuan Wan, Jun Zhao, and participants from numerous seminars and conference talks for comments and suggestions.}
\newcommand{\paperabstract}{%
We develop an identification strategy for average treatment effects on the treated (ATT) in panel data with discrete outcomes. For such outcomes, the parallel trends assumption underlying difference-in-differences (DiD) fails in three distinct ways: mean reversion generates divergent trends when groups differ at baseline, counterfactual probabilities can leave the unit interval, and no single trend is well defined for multi-category outcomes. We replace parallel trends with \textit{transition independence}: absent treatment, transition dynamics conditional on pre-treatment outcomes would be identical between treated and control groups. To accommodate selection on persistent unobserved heterogeneity in transition dynamics, we require transition independence to hold only within latent types. Modeling outcomes as a finite mixture of Markov chains, we identify latent-type and aggregate ATTs from short panels. The framework also yields a flow decomposition of the ATT into inflow and outflow channels. In three empirical applications, our ATT estimates differ substantially from conventional DiD.
}
 
\endlocaldefs

\newcommand{\revtext}[1]{#1}
\newenvironment{revblock}{\par}{\par}
\newcommand{\revdel}[1]{}

\begin{document}

\bibliographystyle{aer}

{
\title{\textbf{\papertitle}\footnote{\papertitlenote}}
\author{%
Young Ahn\\
Department of Economics\\
Brown University\\
youngahn@brown.edu
\and
Hiroyuki Kasahara\\
Vancouver School of Economics\\
University of British Columbia\\
hkasahar@mail.ubc.ca
}
\date{\today}
\maketitle

\begin{abstract}
\paperabstract
\end{abstract}


}
 
\section{Introduction}

Many empirical questions in economics involve outcomes that are inherently discrete and categorical: whether a worker is employed, unemployed, or out of the labor force; or which occupation a worker transitions into. Difference-in-differences (DiD) designs are routinely applied in these settings: researchers binarize the categorical labels into
indicators and compare mean differences between treated and control groups under the parallel trends assumption \citep[e.g.,][]{bustos2011,anderson2011,hvide2018university,charoenwong2019does}. As \citet{Roth2023ecma} demonstrate, treatment effect estimates under parallel trends can be sensitive to functional form, a concern that is especially pronounced for discrete outcomes. Yet the properties of DiD when applied to discrete outcomes have received surprisingly little scrutiny.

When outcomes are discrete, the parallel trends assumption can be logically inconsistent with the data-generating process of bounded outcomes, invalidating the identification strategy, with the following three specific failures. First, discrete outcomes evolve through discrete state transitions rather than continuous movements, so when treated and control groups differ in their baseline distributions, mean reversion alone can generate divergent trends even absent any treatment effect. Second, DiD may produce counterfactual means outside the $[0,1]$ range, undermining the interpretability of the resulting estimates. Third, for multi-category outcomes such as occupation or labor-force status, the notion of a single ``trend'' is ill-defined. The parallel trends assumption offers no coherent basis for comparing the joint temporal evolution of categorical distributions between treated and control groups.

This paper provides a resolution. We propose to replace the parallel trends assumption with \textit{transition independence}: the condition that transition
dynamics, conditional on pre-treatment outcome paths, would be identical
between treated and control units in the absence of treatment. By operating directly on the transition structure of discrete outcomes, identification with transition independence avoids the logical failures of parallel trends while accommodating the bounded, categorical nature of the data.

Our paper makes three main contributions:
\begin{enumerate}
    \item We show that the parallel trends assumption is generically violated for discrete, bounded outcomes: it can generate out-of-bounds counterfactual probabilities and confound mean reversion with treatment effects. We propose transition independence as an alternative credible strategy: as with parallel trends, the plausibility of transition independence can be assessed using pre-treatment data.

    \item We incorporate a discrete latent-type structure under which outcomes
follow latent-type-specific Markov chains. This structure addresses
selection on unobserved heterogeneity in transition dynamics and
mitigates the curse of dimensionality from long pre-treatment histories.
We establish identification of \textit{latent-type average
treatment effects on the treated} (LTATTs) and of aggregate ATTs from
short panels.%

    \item We demonstrate the practical importance of our approach in three
empirical applications: the Dodd--Frank Act, Norway's patent reform, and
the Americans with Disabilities Act of 1990 (ADA). In each case our
estimator produces substantively different results from conventional
DiD.
\end{enumerate}
Our approach avoids the three failures identified above by constructing
counterfactuals from transition probabilities rather than mean levels. We
apply control-group transition probabilities to the treated group's
pre-treatment outcome distribution, predicting how treated units would have
evolved absent treatment. We then identify ATTs by comparing observed
treated outcomes with these counterfactuals. Because we build the
counterfactuals from transition probabilities on the outcome's support,
they respect the bounded, discrete nature of the data and account for
differential mean reversion driven by baseline differences.

In practice, two challenges arise. First, persistent unobserved heterogeneity may simultaneously affect treatment assignment and outcome transitions, so treated and control units with the same pre-treatment outcome may still experience differential mean reversion, violating transition independence. Second, the
number of possible outcome histories grows exponentially with the number of
pre-treatment periods, creating limited overlap in finite samples. 

We
address the first challenge by introducing latent types, requiring transition
independence to hold only conditionally on each type, and the second with a
low-order Markov restriction, which reduces the conditioning set to a small
number of recent outcomes. We then identify the aggregate ATT as a
weighted average of LTATTs, with weights equal to the latent-type
probabilities among treated units. We establish identification of both the
LTATTs and the weights from short panels, enabling us to account for unobserved heterogeneity in transition dynamics even when the number of pre-treatment periods is limited. 
The number of latent types can be selected by the Bayesian information criterion (BIC);
exploiting the finite support of our model, we show that BIC selects the true number of types with probability
$1-o(n^{-1/2})$ under regularity
conditions and that confidence intervals for the aggregate ATT
that ignore the selection step remain asymptotically valid
(Remark~\ref{remark:model-selection} and Section~\crossref{sec:bic-postsel} of the Supplemental
Appendix).

A further distinctive contribution of the transition-based framework is the \textit{flow decomposition} of treatment effects: because we model outcome dynamics through
state-to-state transition probabilities, we can decompose the ATT on any
outcome state into inflow and outflow components
(Remark~\ref{remark:decomposition-by-flows}). Standard DiD, which operates
on outcome levels rather than transitions, offers no such decomposition. In the ADA application, the flow decomposition reveals that the negative employment effect operates primarily through increased transitions from employment directly into out-of-labor-force status, rather than through changes in job-search outcomes, a mechanism not apparent from the DiD estimates.

\paragraph*{Empirical illustration.} We illustrate the framework with three applications in Section \ref{sec:empirical-applications}, each demonstrating a distinct limitation of conventional DiD. 

First, DiD can produce out-of-bounds counterfactuals. In our DiD replication of \cite{charoenwong2019does}'s analysis of the Dodd-Frank Act's 2012 reform, under parallel trends, the implied counterfactual complaint rates fall below zero in post-treatment
periods, violating their interpretation as probabilities (\autoref{fig:complaint-counterfactual-intro}). This violation is a logical
consequence of applying linear extrapolation to a bounded outcome. Our transition-based approach avoids
out-of-bounds counterfactuals by construction, and the sign of the
estimated effect reverses: our estimates indicate an increase in service
quality, whereas DiD reports a decrease.

Second, DiD is susceptible to mean-reversion bias when treated and control groups differ in baseline levels. In \citet{hvide2018university}'s study of Norway's 2003 patent law reform, university inventors (treated) had nearly twice the patenting rate of non-university inventors (control) immediately before the reform (\autoref{fig:norway-trend}). Because treated units had more room to decline, mean reversion causes the DiD estimator to overstate the negative impact, reporting a 4.5 percentage-point decline. Our transition-based approach, which directly accounts for state-dependent dynamics, finds no significant change in patenting rates.

\begin{figure}[t]
    \centering
    \caption{Counterfactual Complaint Rates from the Dodd-Frank Act.}\label{fig:complaint-counterfactual-intro}
    \includegraphics[width=0.75\textwidth]{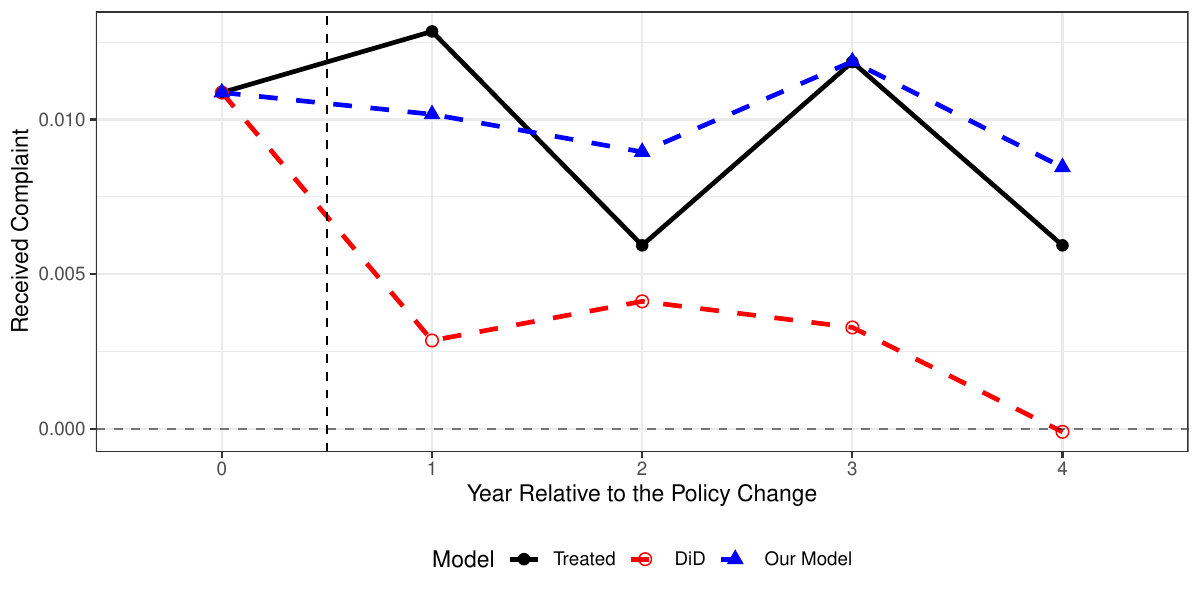}
    \fignote{Average counterfactual complaint rates for treated units had they not been treated, using the data of \cite{charoenwong2019does}. The black line is the observed treated path. The DiD counterfactual adds the control group's level change to the treated group's pre-reform mean; it falls below zero in later post-reform years, an invalid probability. Ours is the transition-based counterfactual with $J=2$, selected by BIC (Remark~\ref{remark:model-selection}).}
\end{figure}

Third, our framework reveals treatment effect mechanisms that DiD cannot detect. Using monthly labor-force status data from the 1990 SIPP panel and following \cite{acemoglu2001consequences} and \cite{lise2023revisiting}, we compare disabled (treated) and non-disabled (control) working-age adults around the ADA's enactment. The flow decomposition shows that the negative employment effect operates primarily through increased transitions from employment directly into out-of-labor-force status, rather than through changes in job-search outcomes (\autoref{fig:ada-decomposition-intro}), an insight into the underlying channel-specific transition mechanism unavailable from standard DiD. Furthermore, while neither conventional DiD nor nonlinear DiD methods \citep{wooldridge2023simple,boussim2026compositional} detect statistically significant employment effects, our transition-based methods reveal significant short-term employment reductions (\autoref{fig:ada-att-simple}). This finding is consistent with prior evidence of unintended short-term consequences of the reform \citep{acemoglu2001consequences}, and further suggests that the effects could be even larger once mean reversion and heterogeneous transition dynamics are taken into account.

\begin{figure}[t]
    \centering
    \caption{Patenting Rates Around the Norwegian Law Reform.}\label{fig:norway-trend}
    \includegraphics[width=0.75\textwidth]{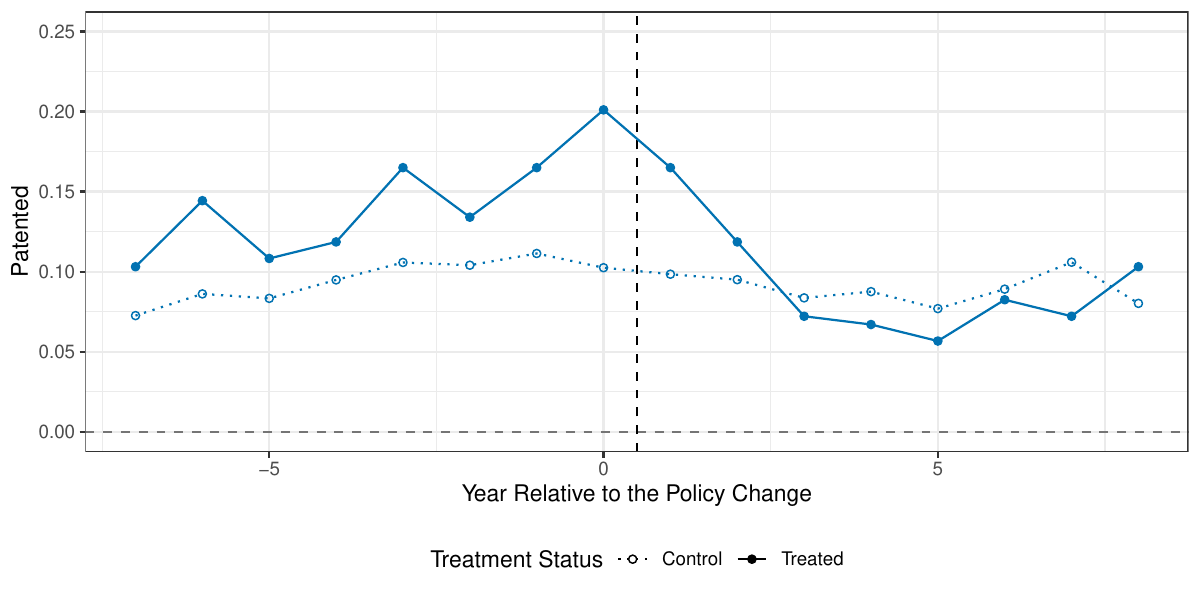}
    \fignote{Annual patenting rates (share of inventors filing at least one patent) for university inventors (treated) and non-university inventors (control) around the 2003 Norwegian reform. Treated rates are nearly twice control rates just before the reform.}
\end{figure}
\begin{figure}[t]
    \centering
    \caption{Decomposing the ADA's Effect on Employment by Transition Channels.}\label{fig:ada-decomposition-intro}
    \includegraphics[width=0.75\textwidth]{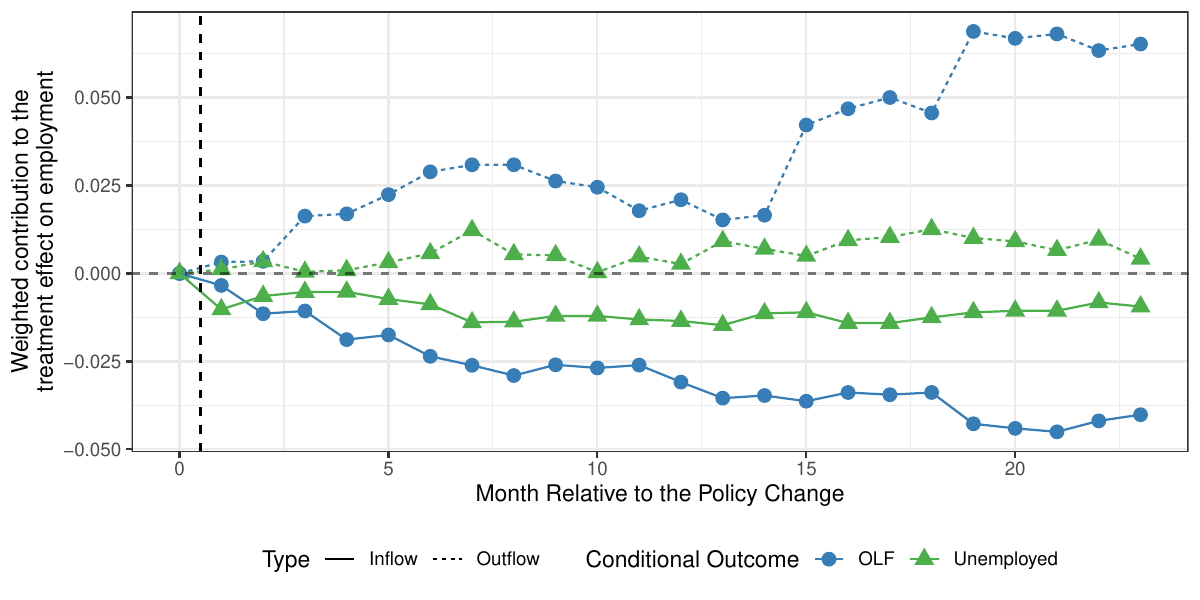}
    \fignote{Flow decomposition of the ADA's effect on employment, using the $J=2$ model selected by BIC (Remark~\ref{remark:model-selection}). Solid lines are weighted inflow effects (transitions into employment from unemployment and from out of the labor force); dashed lines are weighted outflow effects (transitions out of employment); the weights are the treated group's pre-treatment distribution and type shares. The dominant channel is increased movement from employment directly into out-of-labor-force status. See Remark~\ref{remark:decomposition-by-flows} for construction details.}
\end{figure}

\paragraph*{Related literature.} Our framework connects several related literatures: matching on
pre-treatment outcomes, difference-in-differences with pre-treatment
outcomes and nonlinear dynamics, copula-invariance restrictions in DiD,
grouped latent heterogeneity in panels, and sequential exchangeability in
dynamic treatment models.

\begin{figure}[t]
    \centering
    \caption{ATT Estimates of the ADA on Employment: DiD vs Our Method.}\label{fig:ada-att-simple}
    \includegraphics[width=0.75\textwidth]{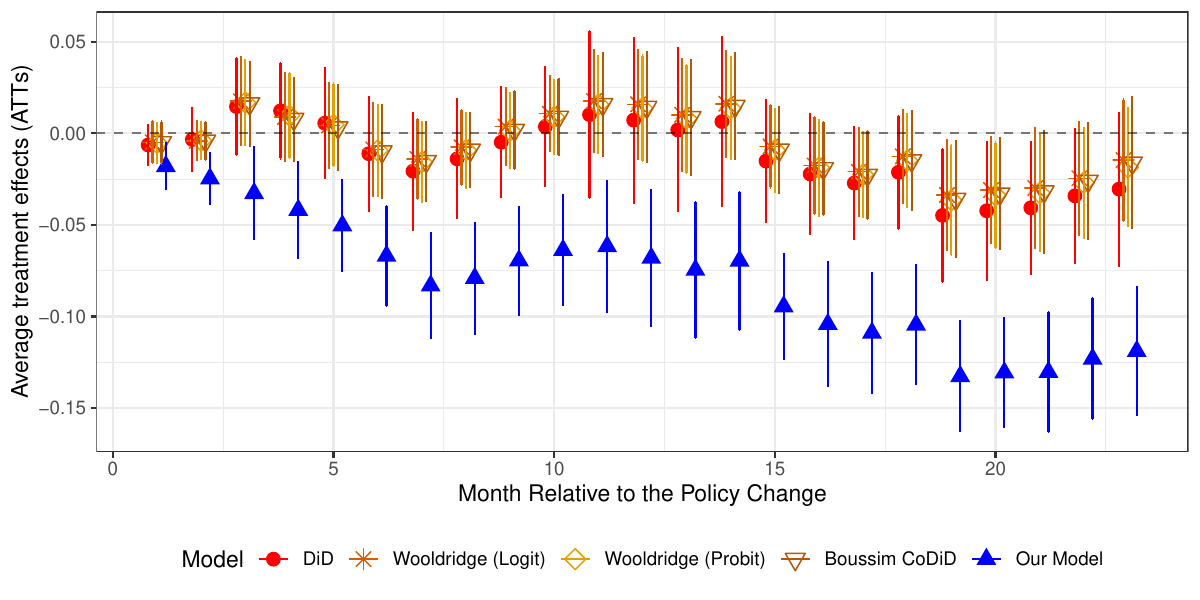}
    \fignote{The figure reports ATT estimates of the ADA on employment. The y-axis reports observed treated employment minus each method's untreated counterfactual. Series are named in the legend. DiD imposes parallel trends. The logit and probit specifications of \citet{wooldridge2023simple} impose parallel trends on the corresponding link transformation of the marginal employment probability. CoDiD \citep{boussim2026compositional} imposes parallel trends on log ratios of category probabilities in the full labor-force-status composition. Our method imposes transition independence with $J=2$ latent types, the specification selected by BIC (Remark~\ref{remark:model-selection}). Vertical bars are 95\% pointwise bootstrap confidence intervals. \autoref{fig:ada-att-all} compares all specifications.}
\end{figure}

The first relevant approach comes from the classic literature on program evaluation using matching methods \citep{card1988measuring,Heckman1997,Heckman1998}.  \cite{card1988measuring} compare outcomes of individuals with identical six-period pre-training employment histories and show that the training effect estimates are sensitive to the histories included: omitting a single pre-treatment history can lead to substantially different estimates. \cite{Heckman1997} use labor force statuses from two periods in their matching procedure: the month in which a unit becomes eligible and the six months prior to eligibility. We complement their approach by proposing a new flexible method to construct counterfactuals by identifying latent groups with similar transition dynamics from past outcomes, while accounting for latent heterogeneity that can affect transition dynamics and selection.  Combining a small number of latent types
with a Markov restriction also resolves the curse of dimensionality
from conditioning on long pre-treatment histories.

The second strand is the DiD literature that incorporates pre-treatment
outcomes and allows for nonlinear outcome dynamics \citep{ROTH2023je}.
\citet[Section~5.4]{angrist2009mostly} and \citet{ding2019bracketing} show
that conditioning on pre-treatment outcomes yields identifying assumptions
distinct from unconditional parallel trends. We complement this literature
by providing a practical solution to the limited-support issues that arise
when outcomes are discrete: a flexible framework that accommodates both
latent heterogeneity and the support restrictions inherent in categorical
data. See Section~\ref{sec:comparison-with-did} for discussion.

Recent work develops nonlinear DiD methods by replacing parallel trends in
outcome means with parallel trends on transformed scales.
\citet{wooldridge2023simple} imposes common changes in a known link
transformation of the conditional mean, whereas
\citet{boussim2026compositional} imposes parallel trends in log ratios across
categories. These methods share our goal of
constructing credible counterfactuals for discrete or compositional
outcomes, but they place identifying restrictions on changes in marginal
distributions rather than on transition dynamics. Consequently,
differences in pre-treatment history distributions may generate the same
mean-reversion problem as in standard DiD. In addition, like standard DiD, nonlinear DiD either models each
category separately or restricts period-by-period marginal compositions, disregarding the joint transition structure. Furthermore,
neither approach accommodates latent heterogeneity in transition
dynamics, which matters especially for discrete outcomes because of their
limited support (see Remark~\ref{remark:logit}).\footnote{\citet{DiFrancesco2025}
develop probability-shift methods for qualitative outcomes. In the health
sciences, \citet{graves2022did} independently propose transition matrices
for DiD with categorical outcomes, but their study does not incorporate latent heterogeneity.}
Our transition-based framework addresses these challenges by capturing
the nonlinear nature of discrete outcomes directly and by modeling latent
heterogeneity in transition laws explicitly.

A third strand, the copula invariance approach to DiD
\citep{callaway2018quantile,callaway2019quantile,ghanem2023evaluating},
restricts the dependence between pre- and post-treatment outcomes in a way
analogous to our restrictions on transition probabilities. This approach,
however, does not accommodate unobserved heterogeneity. With discrete
outcomes, the limited support motivates our introduction of latent types
over pre-treatment outcome paths to capture unobserved heterogeneity in
state dependence and in the propensity to be treated.

A fourth strand addresses latent heterogeneity in panel data.
\citet{Bonhomme2015} and \citet{bonhomme2022discretizing} develop grouped
fixed-effects models that recover latent types by clustering methods; see \citet{Shin2022} for a recent application in the DiD context.
\citet{Arkhangelsky2021} propose the synthetic difference-in-differences
(SDiD) estimator, which accounts for heterogeneous effects by reweighting
control units and time periods. Both approaches, however, require long
panels for consistent estimation, whereas our method identifies
latent-type treatment effects in short panels, which is the
relevant case for most empirical applications.

Finally, transition independence is a version of the sequential
exchangeability assumption of \citet{Robins1986}, which requires that
current treatment be independent of current and future potential outcomes
given the observed covariate--treatment history \citep{RobinsHernan2025}.
With an absorbing treatment and a never-treated control group, this
condition has identifying content only at the adoption date, conditional on the path of
pre-treatment outcomes as covariates; restricted to untreated potential outcomes,
it is our transition independence assumption in the absence of unobserved
heterogeneity (Remark~\ref{remark:se}). Our contribution lies in relaxing it. Sequential exchangeability conditions on observables only
and thus rules out selection on unobserved determinants of outcome
dynamics; \citet{marx2025arxiv} pose identification when it holds only
conditional on unobservables as an open question in their conclusion. We introduce latent heterogeneity in transition probabilities, allow latent types to be correlated with treatment, and establish identification from the finite mixture structure of
observed transitions in short panels.

While transition independence is a natural assumption for discrete outcomes, it may be less plausible when agents are forward-looking and adjust transition behavior in anticipation of treatment, or when aggregate shocks differentially affect treated and control units' transition dynamics. In addition to the BIC selection procedure, we recommend assessing the plausibility of transition independence by testing for pre-treatment differences in transition probabilities between groups, a test analogous to the DiD pre-trends test (Remark \ref{remark:testing-for-transition-independence-Z}), as demonstrated in our empirical applications. %

The remainder of the paper is organized as follows.
Section~\ref{sec:binary-outcome-models} introduces a potential outcome
model with discrete outcomes in the absence of unobserved heterogeneity, and analyzes the identification of ATTs.
Section~\ref{sec:binary-outcome-models-with-latent-type} adds a latent
type structure to allow unobserved heterogeneity in transition dynamics
and establishes identification. Section~\ref{sec:estimation} develops a
two-stage procedure for estimating latent-type and aggregate
ATTs. Section~\ref{sec:empirical-applications} presents three empirical
applications. An R package with replication code is available at
\url{https://github.com/bayesiahn/ak}.

\section{Discrete Outcome Models}\label{sec:binary-outcome-models}

\subsection{Potential Outcome Model with Transition Independence} 
Consider a panel data model over $T=T_0+T_1$ periods, indexed by $t \in \mathcal{T} := \{1, 2, \ldots,T_0, T_0+1,..., T\}$, where $T_0$ and $T_1$ denote the numbers of pre- and post-treatment periods, respectively.  
Observational units are indexed by $i \in \mathcal{N} := \{1, \ldots, n\}$. Each unit may receive a binary treatment at some period or remain untreated throughout.   We assume that all treated units begin treatment simultaneously at period $T_0+1$ and that, once a unit is treated, it remains treated in all remaining periods.\footnote{Extension to staggered treatment timings can be straightforwardly managed by analyzing subgroups with identical treatment onset periods. We provide details in Section~\crossref{sec:staggered} of the Supplemental Appendix.}  %

For each unit $i$ and period $t$, we observe a discrete outcome $Y_{it}$ and a
binary treatment indicator $D_{it} \in \{0,1\}$. The outcome has support
$\mathcal{Y} := \{\bar y^{(1)},\bar y^{(2)},\ldots,\bar y^{(K)}\}$.
Potential outcomes
are indexed by the full treatment path: $Y_{it}(D_{i1},\ldots,D_{iT})$
denotes unit $i$'s outcome at period $t$ under the path
$(D_{i1},\ldots,D_{iT})$. Common adoption at $T_0+1$ and absorbing treatment
leave only two relevant paths. We write
$Y_{it}(1):=Y_{it}(\bs{0}_{T_0},\bs{1}_{T_1})$ for treatment and
$Y_{it}(0):=Y_{it}(\bs{0}_T)$ for no treatment, where $\bs{0}_s$ and
$\bs{1}_s$ denote length-$s$ vectors of zeros and ones. Let
$D_i:=D_{i,T_0+1}$; then $D_{it}=0$ for $t \le T_0$ and $D_{it}=D_i$ for
$t \ge T_0+1$.

For notational brevity, we drop the unit subscript $i$ when no confusion arises; e.g., we write $Y_{it}(0)$ and $D_{i}$ as $Y_t(0)$  and  $D$, respectively.  

Given this definition of $Y_{t}(d)$,  $d\in\{0,1\}$,
define the corresponding vector of binary potential outcomes:
\begin{equation}\label{eq:binary}
\mathbf{X}_{t}(d):=\big(X_t^{(1)}(d),\ldots,X_t^{(K)}(d)\big)^{\top}
:=\big(\mathbf{1}(Y_{t}(d)=\bar y^{(1)}),\ldots,\mathbf{1}(Y_{t}(d)=\bar y^{(K)})\big)^{\top}\in\mathcal{X},
\end{equation}
where $\mathbf{1}(\cdot)$ denotes the indicator function and $\mathcal{X}:=\{(x^{(1)},...,x^{(K)})\in\{0,1\}^K: \sum_{k=1}^K x^{(k)} = 1\}$.
Define $\bs X_t$ analogously to $\bs X_t(d)$, replacing $Y_t(d)$ with the observed outcome $Y_t$.
Then, the vector of observed binarized outcomes $\bs X_t$ relates to the potential outcomes as
 \begin{equation}\label{eq:observed}
 \bs X_{t} = 
 \left\{
 \begin{array}{cc}
 \bs X_t(0) & \text{if $1\leq t \leq T_0$}\\
  D \bs X_{t}(1) + (1-D) \bs X_{t}(0)& \text{if $ t\geq T_0+1$}.
 \end{array}
 \right.
 \end{equation}

Our primary object of interest is the vector of {average treatment effects on the treated (ATTs)}, defined as
\begin{equation}\label{eq:ATT}
\boldsymbol{\mu}^{\text{ATT}}_t := \E\left[\bs{X}_t(1) - \bs{X}_t(0) \mid D = 1\right], \quad t \geq T_0 + 1,
\end{equation}
where the $k$th element of $\boldsymbol{\mu}^{\text{ATT}}_t$ represents the change in the probability of belonging to category $k$ induced by treatment.

We adopt two standard assumptions from the event-study literature: the no anticipatory effects and common support (overlap) assumptions.  
Let $\bs X_1^{T_0}:=\{\bs X_s\}_{s=1}^{T_0}$ denote the collection of pre-treatment outcomes.

\begin{assumptionNA}[No anticipatory effects] \label{A-anticipation}  For all  $t\in\{1,...,T_0\}$ and $i\in \mathcal{N}$, $\bs X_{it} (1) = \bs X_{it}(0)$.
\end{assumptionNA}
\begin{assumptionCS}[Common Support]  \label{A-overlap} For any $\bs x_1^{T_0} \in \mathcal{X}^{T_0}$ with $\Pr(\bs X_1^{T_0}=\bs x_1^{T_0})>0$, 
there exists a positive constant $\epsilon>0$ such that 
    $\epsilon\leq \Pr(D=1|\bs X_1^{T_0}=\bs x_1^{T_0})<1-\epsilon$. \end{assumptionCS}

We impose the  transition independence assumption, which equates the post-$T_0$ transition behavior of untreated potential outcomes across treated and control units, conditional on the entire pre-treatment path. Formally, we require that the transition probabilities of untreated potential outcomes be independent of treatment status:
\begin{equation}\label{eq:transition}
\begin{aligned}
    &\Pr \left(\bs X_{t}(0) = \bs x_{t} \mid \bs X_{T_0} (0) = \bs x_{T_0}, \hdots, \bs X_{1}(0) = \bs x_1, D = 1\right) \\
    &=\Pr  \left(\bs X_{t}(0) = \bs x_{t} \mid \bs X_{T_0} (0) = \bs x_{T_0}, \hdots, \bs X_{1}(0) = \bs x_1, D = 0\right)\quad\text{for $t\geq T_0+1$}
\end{aligned}\tag{TI}
\end{equation}
for all possible pre-treatment outcome paths $\{\bs x_1,  \hdots, \bs x_{T_0}\}$.

\begin{assumptionTI}[Transition independence]\label{A-transition} For $t =T_0+1,...,T$,  Equation (\ref{eq:transition}) holds for all  $\bs x_t\in\mathcal{X}$ and $\bs x_1^{T_0} \in   \mathcal{X}^{T_0}$.
\end{assumptionTI}  

Transition independence implies that the outcome dynamics of control units serve as valid counterfactuals for treated units sharing the same pre-treatment history. By matching treated and control units on their observed outcome paths, the post-treatment transitions of the control group identify what the treated group would have experienced absent treatment.

Assumption~\ref{A-transition} is a conditional-unconfoundedness
condition: (\ref{eq:transition}) is equivalent to
\[
\bs X_{t}(0)\mathrel{\perp\!\!\!\perp}D\mid \bs X_1^{T_0}(0)
\]
for each post-treatment period $t$, so treatment is unconfounded with respect to
the untreated potential outcome given the pre-treatment outcome path. Conditions of this type are well established
\citep[e.g.,][]{Heckman1997,Heckman1998,angrist2009mostly,ding2019bracketing}.
Identification follows the
selection-on-observables argument of the $g$-computation algorithm of
\citet{Robins1986} (Remark \ref{remark:se}); in the DiD literature,
such conditions restrict how treatment selection depends on untreated
potential outcomes \citep{ghanem2025selection}. 

The ATTs in (\ref{eq:ATT}) are then identified under Assumptions
\ref{A-transition}, \ref{A-anticipation}, and \ref{A-overlap}.
\begin{proposition}\label{prop:att-identification-vanilla}
  Suppose that Assumptions \ref{A-transition}, \ref{A-anticipation}, and \ref{A-overlap} hold. Then, for $t=T_0+1,...,T$, ATTs are identified by
  \begin{align}\label{eq:att-identification-vanilla}
\bs \mu^{\text{ATT}}_t = \mathbb{E} \left[\mathbf{X}_t - \mathbb{E} \left[ \mathbf{X}_t \mid \bs X_1^{T_0}, D = 0\right] \bigg| D = 1 \right].
  \end{align}
  \end{proposition}
  
  Proposition \ref{prop:att-identification-vanilla} demonstrates that the ATTs are identified by the difference between the mean of observed post-treatment outcomes from the treated units and that of their \textit{counterfactual} expected untreated potential outcomes  conditional on the entire sequence of pre-treatment outcomes. This result is intuitive: the treatment effect is attributable to the difference between the observed outcomes and the counterfactual outcomes that would have been observed if the treated units had not received the treatment. The counterfactual expected untreated potential outcomes are constructed by extrapolating the transition dynamics from the control group to the treated group with the same history.\footnote{Extensions of Proposition \ref{prop:att-identification-vanilla} to transition
independence conditional on discrete covariates, and to staggered treatment
adoption, where the assumption is imposed cohort by cohort and cohort ATTs
are aggregated as in \citet{CALLAWAY2021}, are developed in
Sections~\crossref{sec:supp-covariates} and~\crossref{sec:staggered} of the
Supplemental Appendix.}

  Proposition \ref{prop:att-identification-vanilla} requires conditioning on the entire pre-treatment outcome path. Sections~\ref{sec:binary-outcome-models-with-latent-type} 
and~\ref{sec:estimation}  turn
this formula into an estimable object  when the path space
is large and transition dynamics are subject to unobserved heterogeneity.

\begin{remark}[Relation to sequential exchangeability]\label{remark:se}
Transition independence is closely related to \textit{sequential
exchangeability}:
at each period, treatment must be independent
of current and future potential outcomes, jointly, given the observed
covariate--treatment history \citep{Robins1986, RobinsHernan2025}; see also
\citet[Assumption 2]{marx2025arxiv}, where the covariate history
consists only of observed lagged outcomes.  In our event-study setting, treatment
varies only at the adoption date ($D_s = 0$ for $s \leq T_0$, $D_s = D$
thereafter), where, on the support of Assumption~\ref{A-overlap}, the
condition is stated as
\begin{equation}\label{eq:se}
\big(\bs X_{T_0+1}(d),\ldots,\bs X_{T}(d)\big)\;\perp\!\!\!\perp\; D
\;\Big|\; \bs X_1^{T_0}\quad\text{for $d\in\{0,1\}$,}
\tag{SE}
\end{equation}
where $\bs X_1^{T_0} = \bs X_1^{T_0}(0)$ by (\ref{eq:observed}). Restricted to untreated potential outcomes ($d=0$), (\ref{eq:se}) implies
$\bs X_t(0) \perp\!\!\!\perp D \mid \bs X_1^{T_0}(0)$ for each $t \geq
T_0+1$, which is (\ref{eq:transition}). Transition independence restricts neither treated potential outcomes
nor the joint dependence of the untreated path, and the converse
fails.\footnote{With $T_0=1$, $T=2$, and binary $Y_t$: let
$\Pr(Y_2(0)=1 \mid Y_1(0), D)=0.5$ throughout, so (\ref{eq:transition})
holds; let $\Pr(D=1 \mid Y_1=1)=0.8$, $\Pr(D=1 \mid Y_1=0)=0.2$,
$\Pr(Y_2(1)=1 \mid Y_1=1, D=1)=0.9$, and $\Pr(Y_2(1)=1 \mid \cdot)=0.5$ in
every other positive-probability cell. Then
$Y_2(1) \not\perp\!\!\!\perp D \mid Y_1$, so (\ref{eq:se}) fails for the treated potential outcome.}  
\end{remark}

\begin{remark}[Decomposition by flows]\label{remark:decomposition-by-flows}  When transition independence holds with one lag of conditioning, the ATT admits a flow decomposition by inflows to and outflows from each outcome state, which is useful for investigating the underlying mechanism through which the treatment dynamically affects outcomes.

Specifically, under Assumptions \ref{A-transition}, \ref{A-anticipation}, and \ref{A-overlap}, we can express the ATT on the $k$th outcome as 
\begin{equation}\label{eq:flow-decomposition}
\begin{aligned}
&\E\!\big[X_{t}^{(k)}(1)-X_{t}^{(k)}(0)\mid D=1\big] \\
&=
\sum_{y \neq \bar{y}^{(k)} }
\underbrace{\begin{multlined}[t]
\Big\{\Pr(Y_t=\bar{y}^{(k)}\mid Y_{T_0}=y, D=1)\\
\quad-\Pr(Y_{t}=\bar{y}^{(k)}\mid Y_{T_0}=y, D=0)\Big\}
\Pr(Y_{T_0}=y\mid D=1)
\end{multlined}}_{\text{inflow effect from } y}\\
&\quad
-
\sum_{y \neq \bar{y}^{(k)}}
\underbrace{\begin{multlined}[t]
\Big\{\Pr(Y_{t}=y\mid Y_{T_0}=\bar{y}^{(k)}, D=1)\\
\quad-\Pr(Y_{t}=y\mid Y_{T_0}=\bar{y}^{(k)}, D=0)\Big\}
\Pr(Y_{T_0}=\bar{y}^{(k)}\mid D=1)
\end{multlined}}_{\text{outflow effect to } y}.
\end{aligned}
\end{equation}
\begin{revblock}
The flow decomposition \eqref{eq:flow-decomposition} implies that changes in the probability of remaining in the focal state $\bar y^{(k)}$ (e.g., employed-to-employed) can be understood as the net contribution of transition flows to and from all alternative states. For illustration, consider the ADA example with three outcome states: Employment ($E$), Unemployment ($U$), and Out of Labor Force ($O$).  The introduction of the ADA may induce changes in the transition dynamics over periods across these three outcome states. Under transition independence with the Markov assumption, such impacts on transition probabilities can be illustrated as in Figure~\ref{fig:markov-comparison}, where the ADA induces a larger net outflow from $E$ to $U$ or $O$ by increasing the outflow from  $E$ to $U$ or $O$ and decreasing the inflow from $U$ or $O$ to $E$.

The flow-based decomposition \eqref{eq:flow-decomposition} reveals
cumulative mobility patterns that conventional ATTs on outcome levels
can mask. For example, an increase in employment may be driven primarily by reduced separations (lower outflow from $E$), by enhanced hiring from $U$ or $O$ (higher inflow into $E$), or by a combination of both. Furthermore, the decomposition provides insight into the underlying mechanism. Many policies, including the ADA, operate through specific channels such as job retention, accommodation costs, or hiring practices. Estimating inflow and outflow components separately can thus help assess which hypothesized channels are empirically salient.\footnote{The decomposition extends to each latent type and can be aggregated across types when we introduce latent heterogeneity as in Section \ref{sec:binary-outcome-models-with-latent-type}. See Corollary~\crossref{cor:flow-decomposition-by-type} in the Supplemental Appendix for details.}
\end{revblock}

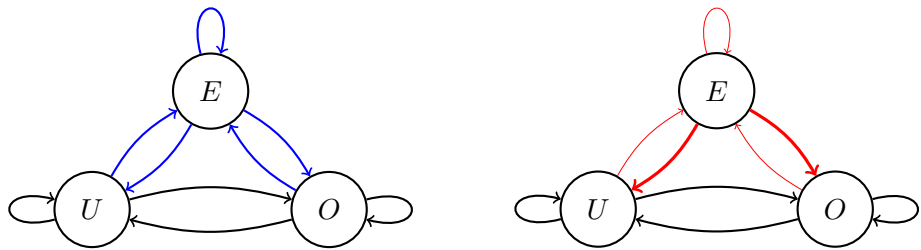
\begin{figure}[t]
\centering
\caption{ADA's Effect on Employment Transitions.}\label{fig:markov-comparison}
\begin{minipage}{0.40\textwidth}
\centering
\begin{tikzpicture}[node distance=1.2cm, auto, thick,
                    state/.style={circle, draw, minimum size=1cm, align=center}]

  \node[state] (E) {$E$};
  \node[state] (U) [below left=of E] {$U$};
  \node[state] (O) [below right=of E] {$O$};

  \draw[->, bend left=15, blue] (E) to (U);
  \draw[->, bend left=15, blue] (E) to (O);
  \draw[->, loop above, blue] (E) to (E);

  \draw[->, bend left=15, blue] (U) to (E);
  \draw[->, bend left=15, blue] (O) to (E);

  \draw[->, bend left=15] (U) to (O);
  \draw[->, bend left=15] (O) to (U);

  \draw[->, loop left] (U) to (U);
  \draw[->, loop right] (O) to (O);

\end{tikzpicture}

\bigskip
{\footnotesize \quad $Y_{it}({ 0})$: Untreated Potential Outcome}
\end{minipage}
\begin{minipage}{0.40\textwidth}
\centering
\begin{tikzpicture}[node distance=1.2cm, auto, thick,
                    state/.style={circle, draw, minimum size=1cm, align=center}]

  \node[state] (E) {$E$};
  \node[state] (U) [below left=of E] {$U$};
  \node[state] (O) [below right=of E] {$O$};

  \draw[->, bend left=15, very thick, red] (E) to (U);
  \draw[->, bend left=15, very thick, red] (E) to (O);

  \draw[->, bend left=15, very thin, red] (U) to (E);
  \draw[->, bend left=15, very thin, red] (O) to (E);

  \draw[->, bend left=15] (U) to (O);
  \draw[->, bend left=15] (O) to (U);

  \draw[->, loop above, very thin, red] (E) to (E);
  \draw[->, loop left] (U) to (U);
  \draw[->, loop right] (O) to (O);

\end{tikzpicture}

\bigskip
{\footnotesize \quad ${Y_{it}({ 1})}$: {Treated} Potential Outcome}
\end{minipage}
\vskip 2.00ex
\fignote{How the ADA affects employment through the inflow and outflow effects in \eqref{eq:flow-decomposition}, among employment (E), unemployment (U), and out-of-labor-force (O). Colored arrows are the transitions that drive the change in employment under treatment. Arrow thickness reflects the relative strength of each flow and is proportional to the weighted inflow and outflow effects estimated for the ADA in Section~\ref{sec:ada}.}
\end{figure}

\end{remark}

\subsection{Comparison with the DiD Estimator}\label{sec:comparison-with-did}
Many empirical studies apply the DiD estimator when the outcome of interest is
binary \citep{bustos2011,anderson2011,hvide2018university,charoenwong2019does}.
Even when the outcome is discrete with multiple categories, we may apply the
DiD estimator to each of the binary outcomes in (\ref{eq:binary}), where the
DiD estimator identifies the population quantity
\[
\bs \mu^{\text{DiD}}_t = \E[\bs X_{t} - \bs X_{T_0} |D=1] - \E[\bs X_{t} -\bs X_{T_0} |D=0]\quad \text{for $t=T_0+1,...,T$}.
\]
The key assumption for the DiD estimator is the following mean change
independence assumption, also known as parallel trends.
\begin{assumptionPT}[Parallel trends] \label{A-trend}
   For $t=T_0+1,...,T$,
    \begin{equation}\mathbb{E}\left[\bs X_{t} (0)-\bs X_{T_0}(0) \mid D=0\right]=\mathbb{E}\left[\bs X_{t} (0)-\bs X_{T_0}(0)\mid D=1\right]. \label{eq:trend} \tag{PT}
    \end{equation}
\end{assumptionPT}
Under Assumptions \ref{A-anticipation}, \ref{A-overlap}, and \ref{A-trend},
the DiD estimand identifies the ATT, $\bs\mu^{\text{ATT}}_t=\bs\mu^{\text{DiD}}_t$;
see \citet{ROTH2023je}.

However, if the initial distributions of outcomes in the pre-treatment periods differ between treated and control units, the parallel trends assumption (\ref{A-trend}) becomes implausible for binary outcomes $\{\bs X_t\}_{t=1}^T$ due to mean reversion effects.

To illustrate mean reversion, consider binary employment status $X_t\in \{0,1\}$ across two periods, $t=1,2$. Initially, 50\% of treated units and 25\% of control units are employed ($\Pr (X_{1}(0) = 1 \mid D= 1) = 0.5$ and $\Pr (X_{1}(0) = 1 \mid D = 0) = 0.25$). By period 2, 87.5\% of treated units are employed. We assume that employed individuals remain employed.

As illustrated in \autoref{fig:parallel-trends-vs-transitions}, suppose that two-thirds of unemployed control units become employed ($\Pr(X_{2}(0) = 1 \mid X_{1}(0) = 0, D_i = 0) = 2/3$), corresponding to a 50 percentage-point increase in their employment rate. Under the parallel trends assumption, this requires the treated group's employment rate to also rise by 50 percentage points in the absence of treatment. Such an implication is implausible: it would require all unemployed treated units to become employed, yielding 100\% employment ($\Pr(X_{2}(0) = 1 \mid X_{1}(0) = 0, D = 1) = 1$). This extreme prediction arises because parallel trends ignores mean reversion: groups with higher initial rates have limited room to improve further. 

\begin{figure}[t]
    \centering
    \caption{A Numerical Example: Parallel Trends and Transition Independence.}\label{fig:parallel-trends-vs-transitions}
    \includegraphics[width=0.75\textwidth]{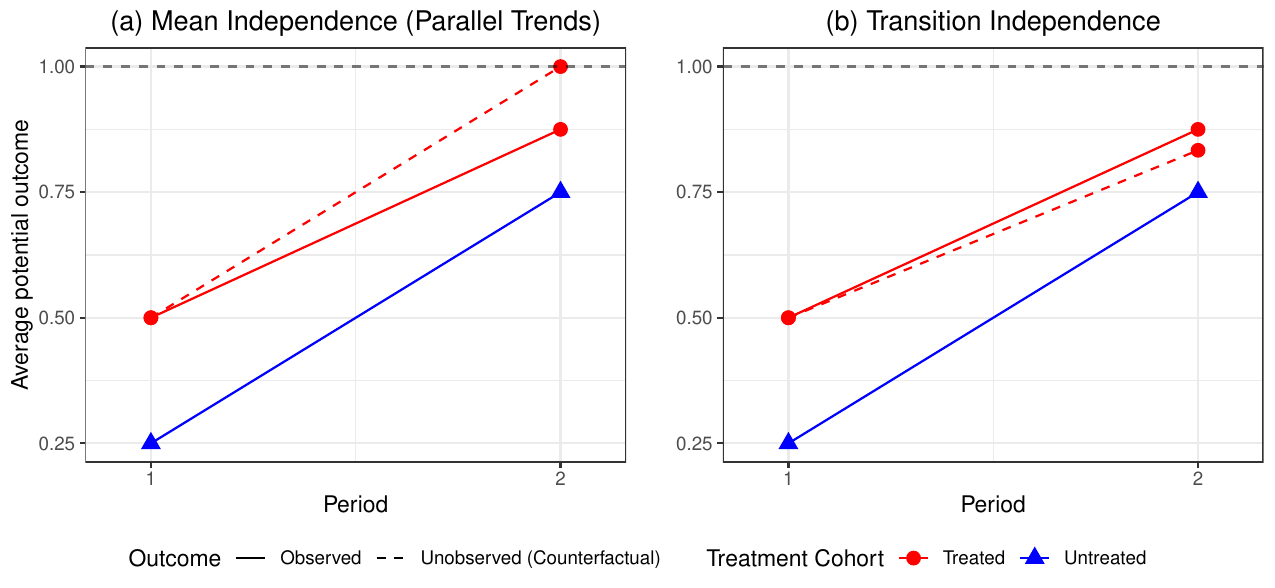}
    \fignote{In the left panel parallel trends holds but transition independence fails; in the right panel transition independence holds but parallel trends fails. The data generating process is given in Section~\ref{sec:binary-outcome-models}.}
\end{figure}

Transition independence instead constructs counterfactuals directly from transition probabilities, thereby avoiding the implausible implications induced by parallel trends:  Assumption \ref{A-transition} requires that the transition probability function of the potential outcome $\bs X_t(0)$, rather than its mean change over time, be independent of treatment status $D$ as $\Pr(X_{2}(0) = 1 \mid X_{1}(0) = 0, D = 0) = \Pr(X_{2}(0) = 1 \mid X_{1}(0) = 0, D = 1) = 2/3$. Notably, these assumptions yield opposite ATT estimates in this example: negative under parallel trends but positive under transition independence.  %

Prior work notes the fragility of parallel trends for discrete and
bounded outcomes
\citep{blundell2004,wooldridge2023simple,marx2023parallel}; we provide the exact bias characterization in the following proposition.

\begin{proposition}\label{proposition:dbt-vs-did}
    Suppose that Assumptions \ref{A-transition}, \ref{A-anticipation}, and \ref{A-overlap} hold. Then, the bias of the DiD estimator is given by
\begin{align}\label{eq:did-bias-decomposition}
 \bs \mu^{\text{DiD}}_t-\bs \mu^{\text{ATT}}_t=    & \sum_{\bs x_1^{T_0}  \in \mathcal{X}^{T_0}}  \mathbb{E} [ \bs X_{t}-\bs x_{T_0} |\bs X_1^{T_0} = \bs x_1^{T_0},
 D=0]\nonumber \\
 &\quad\qquad \times  
   \left\{ \Pr\left(\bs X_1^{T_0} = \bs x_1^{T_0} \mid D = 1\right) - \Pr\left(\bs X_1^{T_0} = \bs x_1^{T_0} \mid D = 0\right) \right\}
\end{align}
  for $t=T_0+1,...,T$.
\end{proposition}

Proposition \ref{proposition:dbt-vs-did} highlights the  two sources of bias for the DiD estimator under the transition independence assumption.  Specifically, the bias arises from differences in pre-treatment outcome distributions between treated and control units,  $\Pr(\bs X_1^{T_0} = \bs x_1^{T_0} \mid D = 1) - \Pr(\bs X_1^{T_0} = \bs x_1^{T_0} \mid D = 0)$, multiplied by trends in untreated potential outcomes, $\mathbb{E} [ \bs X_{t}-\bs x_{T_0} |\bs X_1^{T_0} = \bs x_1^{T_0},
 D=0]$. See \cite{ding2019bracketing} for a similar decomposition in continuous outcomes and \cite{angrist2009mostly} for linear models in the context of conditional parallel trends. This suggests that, when transition independence holds, the DiD estimator may be substantially biased when pre-treatment outcome distributions differ between treated and control units, especially when control units experience large outcome changes over time. Since each term in the bias decomposition \eqref{eq:did-bias-decomposition} is identifiable from the data, we can investigate the potential  bias and its source in the DiD estimator under the transition independence assumption.

\begin{remark}[Nonlinear DiD and transition independence]\label{remark:logit}
A natural alternative for discrete outcomes is the nonlinear DiD of \citet{wooldridge2023simple}, which replaces linear parallel trends with parallel trends on a nonlinear transformation of the conditional {mean} \citep[see also][]{blundell2004,athey2006identification,blundell2009,Puhani2012}. For each category $k$, let $G$ be a known, strictly increasing conditional-mean function, so that $G^{-1}$ is the corresponding link. The index parallel trends restriction requires the change in $G^{-1}\!\big(\E[X^{(k)}_t(0)\mid D=d]\big)$ between $T_0$ and $t$ to be common across $d\in\{0,1\}$ \citep[eqs.~(2.6)--(2.8)]{wooldridge2023simple}.

Replacing the mean change in (\ref{eq:trend}) with the change in $G^{-1}(\E[X_t^{(k)}(0)\mid D])$, the nonlinear DiD estimator identifies, for each category $k$ and each $t=T_0+1,...,T$, the quantity
\begin{equation}\label{eq:index-did-estimand}
\E\!\left[X^{(k)}_t \mid D=1\right]-G\!\left(G^{-1}\!\left(\E\!\left[X^{(k)}_{T_0} \mid D=1\right]\right)+\Delta^{(k)}_t\right),
\end{equation}
where $\Delta^{(k)}_t:=G^{-1}\!\left(\E\!\left[X^{(k)}_{t} \mid D=0\right]\right)-G^{-1}\!\left(\E\!\left[X^{(k)}_{T_0} \mid D=0\right]\right)$ is the change in the control group between $T_0$ and $t$ on the index scale.

Like canonical parallel trends, the counterfactual untreated mean for nonlinear DiD  is constructed from the pre-treatment mean and the change in the control group, but now on the transformed scale. Under transition independence, this construction yields bias analogous
to \eqref{eq:did-bias-decomposition} in
Proposition~\ref{proposition:dbt-vs-did}. 
\begin{proposition}\label{proposition:wooldridge-index}
Suppose that Assumptions \ref{A-transition}, \ref{A-anticipation}, and \ref{A-overlap} hold, and let $G$ be a known, strictly increasing conditional-mean function. Then, for $t=T_0+1,...,T$ and each category $k$, the bias of the nonlinear DiD estimand (\ref{eq:index-did-estimand}) relative to the $k$th element $\mu^{\text{ATT},(k)}_t$ of $\bs\mu^{\text{ATT}}_t$ is given by
\begin{align}\label{eq:index-did-bias}
& \sum_{\bs x_1^{T_0}  \in \mathcal{X}^{T_0}}  \mathbb{E} [ X^{(k)}_{t}- \lambda^{(k)}_t x^{(k)}_{T_0} |\bs X_1^{T_0} = \bs x_1^{T_0},
 D=0]\nonumber \\
 &\quad\qquad \times
   \left\{ \Pr\left(\bs X_1^{T_0} = \bs x_1^{T_0} \mid D = 1\right) - \Pr\left(\bs X_1^{T_0} = \bs x_1^{T_0} \mid D = 0\right) \right\},
\end{align}
where
\begin{equation}\label{eq:index-did-weight}
\lambda^{(k)}_t:=\frac{G\!\left(G^{-1}\!\left(\E\!\left[X^{(k)}_{T_0} \mid D=1\right]\right)+\Delta^{(k)}_t\right)-G\!\left(G^{-1}\!\left(\E\!\left[X^{(k)}_{T_0} \mid D=0\right]\right)+\Delta^{(k)}_t\right)}{\E\!\left[X^{(k)}_{T_0} \mid D=1\right]-\E\!\left[X^{(k)}_{T_0} \mid D=0\right]}
\end{equation}
if $\E[X^{(k)}_{T_0}\mid D=1]\neq\E[X^{(k)}_{T_0}\mid D=0]$, and $\lambda^{(k)}_t:=1$ otherwise.
\end{proposition}

Proposition~\ref{proposition:wooldridge-index} shows that a logit, probit, or any other
single-index link does not eliminate the source of DiD bias characterized in Proposition~\ref{proposition:dbt-vs-did}. Transforming the scale therefore does not address mean reversion: the nonlinear DiD estimand may remain biased if the pre-treatment outcome distributions differ between treated and control units.\footnote{Imposing the index restriction conditional on pre-treatment outcomes carries the same curse of dimensionality  as the conditional DiD estimator. One could impose the index-scale analogue of Assumption~\ref{A-conditional-trend}  \citep{wooldridge2023simple}, requiring the change in $G^{-1}(\E[X^{(k)}_t(0)\mid \bs X_1^{T_0}=\bs x_1^{T_0},D])$ to be common across the two groups within each pre-treatment history. That restriction has to hold history by history, so the conditioning set again grows exponentially in the number of pre-treatment periods, and the link adds a functional-form restriction on top of that conditioning rather than reducing what must be conditioned on.}

Another related approach is \citet{boussim2026compositional}'s Compositional Difference-in-Differences (CoDiD), which imposes parallel trends on the log ratios of the marginal category probabilities. Specifically, CoDiD requires the change in $\log(\E[X^{(k)}_t(0)\mid D = d]/\E[X^{(K)}_t(0)\mid D = d])$ between $T_0$ and $t$ to be common across $d \in \{0, 1\}$ for each category $k=1, \ldots, K-1$. This restriction also concerns period-by-period marginal compositions rather than the transition law, and therefore is not generally implied by transition independence. For binary outcomes, the log ratio reduces to the scalar logit, and Proposition~\ref{proposition:wooldridge-index} applies as stated.

\end{remark}

\begin{remark}[Equivalence with conditional parallel trends]\label{remark:equivalence}\label{remark:conditional-did-categorical}
A related identification approach conditions the parallel trends assumption
on pre-treatment outcomes.
 \begin{assumptionCPT}[Parallel trends, conditional on pre-treatment outcomes] \label{A-conditional-trend}
    For all $\bs x_1^{T_0} \in \mathcal{X}^{T_0}$ and post-treatment period $t =T_0+1,...,T$,
    $\mathbb{E}\bigl[\bs X_{t} (0)-\bs X_{T_0}(0) \bigm| \bs X_1^{T_0} = \bs x_1^{T_0},\allowbreak D=0\bigr]
     =\mathbb{E}\bigl[\bs X_{t} (0)-\bs X_{T_0}(0) \bigm| \bs X_1^{T_0} = \bs x_1^{T_0},\allowbreak D=1\bigr]$.
\end{assumptionCPT}
Under the no-anticipation assumption (Assumption \ref{A-anticipation}),
Assumption \ref{A-conditional-trend} is equivalent to Assumption
\ref{A-transition}, and the equivalence is immediate: $\bs X_{T_0}(0)$ is
fixed by the conditioning path, so it cancels from both sides of the
conditional parallel trends restriction, which then requires the conditional
mean of $\bs X_t(0)$ given the path to be common across the two groups;
because $\bs X_t(0)$ is a vector of indicators, that conditional mean is the
entire probability mass function, which is the transition independence
restriction (\ref{eq:transition}).   A formal statement is in
Section~\crossref{sec:supp-equivalence} of the Supplemental Appendix.
Without latent types, transition independence therefore yields an ATT estimate identical to that of a conditional DiD that conditions on the full pre-treatment
outcome history.\footnote{One practical caveat is that
with more than two categories,  conditioning on the history of a single 
binarized outcome is generally insufficient for identification: the analogous restriction may hold for each binarized outcome separately while Assumption CPT, which conditions on the full categorical history, fails. } %
\end{remark}

Two challenges arise in applying the conditional parallel trends
approach. First, when persistent unobserved heterogeneity affects both
treatment assignment and untreated outcome transitions, conditioning on
pre-treatment outcomes does not eliminate differential mean reversion:
treated and control units with the same pre-treatment values can revert
toward permanently different untreated means
\citep[Section~4.1]{dechaisemartin2026didbook}. Second, the number of
possible histories grows exponentially in the number of pre-treatment
periods: in typical panels, many histories contain few observations and
overlap between treated and control units weakens, which prevents researchers from directly implementing Assumption~\ref{A-conditional-trend}.

The
next section addresses both challenges: discrete latent heterogeneity
restores transition independence conditional on latent type, and,
within each type, a Markov restriction reduces the conditioning history
to the most recent outcomes.
    
\section{Discrete Outcome Models with Latent Heterogeneity}\label{sec:binary-outcome-models-with-latent-type}

We model unobserved heterogeneity with a discrete
latent type: each unit $i$ has type $Z_i \in \mathcal{J} :=
\{1, \ldots, J\}$, with $J$ known, and $\pi^j := \Pr(Z_i = j)$ denotes
the population share of type~$j$. We impose analogs of Assumptions~\ref{A-transition} and~\ref{A-overlap}
across all periods, restricted to untreated potential outcomes and
conditional on latent types.

\begin{assumptionTILT}[Transition independence conditional on latent
types]\label{A-transition-Z}
For all $t = 2, \ldots, T$, $\bs x_t \in \mathcal{X}$, \revtext{$\bs x_1^{t-1} \in \mathcal{X}^{t-1}$,}
 and $j \in \mathcal{J}$,
$\Pr\bigl(\bs X_{t}(0) = \bs x_t \mid \bs X_1^{t-1}(0) = \bs x_1^{t-1},\,
D=0,\, Z=j\bigr)
=
\Pr\bigl(\bs X_{t}(0) = \bs x_t \mid \bs X_1^{t-1}(0) = \bs x_1^{t-1},\,
D=1,\, Z=j\bigr).$
\end{assumptionTILT}

\begin{assumptionCSLT}[Common support conditional on latent types]  \label{A-overlap-Z} 
There exists $\epsilon>0$ such that, for all $j = 1, \ldots, J$ and all
$\bs x_1^{t}\in\mathcal{X}^{t}$ with $\Pr(\bs X_1^{t}(0) = \bs x_1^{t}\mid Z = j)>0$, 
$\epsilon\leq \Pr(D=1\mid \bs X_1^{t}(0) = \bs x_1^{t}, Z = j)\leq 1-\epsilon$. 
\end{assumptionCSLT}

Because Assumption~\ref{A-transition-Z} restores transition
independence within each latent type, the natural target is the
type-specific ATT. We define the \emph{latent-type average treatment
effect on the treated} (LTATT) as
\begin{equation}\label{LTATT}
\bs \mu^{\text{ATT},j}_t := \E\left[\bs X_t(1) - \bs X_t(0) \mid
D = 1,\, Z=j\right]
\quad\text{for } j \in \mathcal{J},\ t = T_0+1,\ldots,T.
\end{equation}
The ATT aggregates the LTATTs with weights equal to the type shares
among the treated:
\begin{equation}\label{ATT}
\bs \mu^{\text{ATT}}_t := \E\left[\bs X_t(1) - \bs X_t(0) \mid
D = 1\right]
= \sum_{j=1}^J \Pr(Z=j \mid D=1)\, \bs \mu^{\text{ATT},j}_t
\quad\text{for } t = T_0+1,\ldots,T.
\end{equation}

Let $\bs W_i := (\bs X_{i1}^\top, \ldots, \bs X_{iT}^\top, D_i)^\top
\in \mathcal{W} := \mathcal{X}^T \times \{0,1\}$ collect unit $i$'s
outcome vectors over the $T$ periods and its treatment indicator. We
assume the data are generated from the finite mixture with probability
mass function (PMF)
\begin{equation}\label{eq:mixture}
p_{\bs W}(\bs w) = \sum_{j=1}^J \pi^j\, p_{\bs W}^j(\bs w),
\end{equation}
where $p_{\bs W}^j(\bs w) := \Pr(\bs W = \bs w \mid Z=j)$ is the
conditional PMF of $\bs W$ given $Z=j$. The true number of components
$J$ is the smallest integer for which $p_{\bs W}$ admits the
representation \eqref{eq:mixture}.

\begin{assumptionRSM}[Random sampling from a finite mixture]
\label{A-sample}
(a) $\{\bs W_i\}_{i=1}^n$ are i.i.d.\ draws from \eqref{eq:mixture},
where observed and potential outcomes are related by
\eqref{eq:observed} and Assumptions \ref{A-anticipation},
\ref{A-transition-Z}, and \ref{A-overlap-Z} hold. (b) The number of
components $J$ is known. (c) $\pi^1 < \pi^2 < \cdots < \pi^J$.
\end{assumptionRSM}

Assumption~\ref{A-sample}(a) combines the maintained restrictions (no
anticipation, transition independence conditional on latent types, and
common support) with the specification that the data are drawn from a
$J$-component mixture, each component corresponding to a latent type
with its own transition dynamics and treatment probability, so latent
type may be correlated with treatment. Assumption~\ref{A-sample}(b)
takes $J$ as known; Remark~\ref{remark:model-selection} discusses its
selection in practice. Assumption~\ref{A-sample}(c) resolves the
labeling indeterminacy: finite mixtures are identified only up to
permutation of their components, and ordering the weights pins down
the labels.

 For identification, we assume that, conditional on latent type and
treatment status, the potential outcome processes follow first-order
Markov chains. The Markov restriction also reduces the conditioning
history to the most recent outcome, addressing the dimensionality
problem noted above. Write
$\bs X_1^{t-1}(d) := \{\bs X_{s}(d)\}_{s=1}^{t-1}$.
\begin{assumptionM}[the first-order Markov]\label{A-Markov} 
 For all $j=1,2,...,J$, conditional on $Z=j$ and $D=d$, $\{\bs X_{t}(d): t= 1,...,T\}$ follows a (non-stationary) first-order Markov process, i.e., for $t=2,..., T$ and all $d\in \{0,1\}$,
 $\Pr\bigl(\bs X_{t}(d)\mid \{\bs X_s(d)\}_{s=1}^{t-1},\allowbreak D=d,\allowbreak Z=j\bigr)
 =
 \Pr\bigl(\bs X_{t}(d)\mid \bs X_{t-1}(d),\allowbreak D=d,\allowbreak Z=j\bigr)$.
\end{assumptionM}

Collect the unknown probability mass functions for 
$Z$ and the type-specific transition probabilities of the potential outcomes $\bs X_t(d)$, $d\in\{0,1\}$, into the parameter vector
$\bs{\psi}=(\bs\pi,\bs\varphi^1,...,\bs\varphi^J)\in \Theta_{\bs\psi}$.
Here $\bs\pi=(\pi^1,...,\pi^J)^\top$ satisfies $\epsilon\leq \pi^j\leq 1-\epsilon$ for some $\epsilon>0$ and $\sum_{j=1}^J \pi^j=1$, and each component is the collection
$\bs\varphi^j=\{p^j_{\bs X_1(0),D}(\cdot,\cdot),\allowbreak
\{p^j_{\bs X_t(0)|\bs X_{t-1}(0)}(\cdot|\cdot)\}_{t=2}^{T},\allowbreak
\{p^j_{\bs X_t(1)|\bs X_{t-1}(1)}(\cdot|\cdot)\}_{t=T_0+1}^{T}\}$.
The component probabilities are
$p^j_{\bs X_1(0),D}(\bs x_1,d):=\Pr\bigl(\bs X_{1}(0)=\bs x_1,\allowbreak D=d \mid Z = j\bigr)$ and
$p^j_{\bs X_t(d)|\bs X_{t-1}(d)}(\bs x_t|\bs x_{t-1}):=\Pr\bigl(\bs X_{t}(d)=\bs x_t\mid \bs X_{t-1}(d)=\bs x_{t-1},\allowbreak Z=j\bigr)$
for $d\in\{0,1\}$ and $t\in\{2,...,T\}$.
 
Under Assumption \ref{A-Markov}, and noting from Assumption \ref{A-anticipation} that $\bs X_t(0)=\bs X_t(1)$ for $t=1,\ldots,T_0$, by explicitly writing its dependence on the parameters $\bs \psi$ and $\bs \varphi^j$,
the PMF of $\bs W=(\bs X_{1},\ldots,\bs X_{T},D)$ in \eqref{eq:mixture} can be written as
\begin{equation}\label{mixture-2}
p_{\bs W}(\bs w;\bs \psi) := \sum_{j=1}^J \pi^j p_{\bs W}^j(\bs w;\bs \varphi^j)
\end{equation}
where
\begin{align*}
&p_{\bs W}^j(\bs w;\bs \varphi) \nonumber\\
&:=\left\{
\begin{array}{lc}
   p^j_{\bs X_1(0),D}(\bs x_1,0) \prod_{t=2}^{T} p^j_{\bs X_t(0)|\bs X_{t-1}(0)}(\bs x_t|\bs x_{t-1})&\text{if $D=0$}\\
  p^j_{\bs X_1(0),D}(\bs x_1,1) \prod_{t=2}^{T_0} p^j_{\bs X_t(0)|\bs X_{t-1}(0)}(\bs x_t|\bs x_{t-1})\prod_{t=T_0+1}^{T} p^j_{\bs X_t(1)|\bs X_{t-1}(1)}(\bs x_t|\bs x_{t-1})&\text{if $D=1$.}
\end{array}
\right.
\end{align*} 

Let $\bs \psi^*$ denote the true value of   $\bs \psi$, so that the true distribution of 
$\bs W$ is given by $p_{\bs W}(\bs w;\bs \psi^*)$.

By extending the arguments in \citet{anderson54pcma}, \citet{Madansky60}, \citet{hu08}, \citet{Kasahara2009}, \citet{carroll10jns}, \citet{Hu12}, and \citet{ishimaru2025},  the following proposition establishes the identification of the mixture models for the potential outcome processes $\{\boldsymbol{X}_t(0)\}_{t=1}^T$ and $\{\boldsymbol{X}_t(1)\}_{t=T_0+1}^T$ in (\ref{mixture-2}), even when the panel is relatively short.  

\begin{proposition}\label{P-2}
Suppose that Assumptions~\ref{A-sample}, \ref{A-Markov}, and \crossref{A-P-2} hold.
If $T_0 \ge k + 1$, $T \ge 2(k + 1)$, and $J \le |\mathcal{X}|^{k}$ for $k = 1, 2, 3, \ldots$, 
then $\boldsymbol{\psi}$ is uniquely identified from 
$\{p_{\boldsymbol{W}}(\boldsymbol{w}; \boldsymbol{\psi}) : \boldsymbol{w} \in \boldsymbol{\mathcal{W}}\}$.
\end{proposition}

Assumption~\crossref{A-P-2}, presented in the proof of Proposition~\ref{P-2},
imposes a set of regularity conditions for identification. 
The regularity conditions ensure that certain matrices of transition probabilities are non-singular, guaranteeing that variations in past outcomes induce sufficiently heterogeneous changes in current outcome probabilities across latent types for identification.

Applying Proposition~\ref{P-2} with $k = 1$, we can identify $\boldsymbol{\psi}$ when 
$J \le |\mathcal{X}|$, $T_0 = 2$, and $T = 4$ under these regularity conditions. This implies, for example, that when the discrete outcome takes three possible values (e.g., $E$, $U$, and $O$ in the ADA example), we may identify three latent types ($J=3$) from four-period panel data with two pre-treatment periods ($T=4$, $T_0=2$) under the first-order Markov assumption. When the panel is longer, we may identify more latent types.

Under Assumption \ref{A-Markov},  
Assumption \ref{A-transition-Z}  implies that 
$\Pr\bigl(\bs X_{t}(0)=\bs x_t\mid \bs X_{T_0}(0)=\bs x_{T_0},\allowbreak D=0,\allowbreak Z=j\bigr)
=
\Pr\bigl(\bs X_{t}(0)=\bs x_t\mid \bs X_{T_0}(0)=\bs x_{T_0},\allowbreak D=1,\allowbreak Z=j\bigr)$.
Then, given the identification of   $\bs\psi $ in Proposition \ref{P-2},
 we may identify  the LTATTs $\bs\mu_{t}^{ATT,j}$ in \eqref{LTATT} for each $j$ as the following proposition states. For $\bs x\in\mathcal{X}$, $d\in\{0,1\}$, and $j\in\mathcal{J}$, write $p^{j}_{T_0}(\bs x):=\Pr(\bs X_{T_0}=\bs x\mid D=1,Z=j)$ and $m^{j,d}_{t}(\bs x):=\E[\bs X_{t}\mid \bs X_{T_0}=\bs x,D=d,Z=j]$. 
 
\begin{proposition}\label{P-1}    Suppose that Assumptions  \ref{A-sample}, \ref{A-Markov}, \ref{A-transition-Z}, and \crossref{A-P-2} hold.  Then,  for each post-treatment period $t \geq T_0+1$ and  for all  $j\in \mathcal{J}$, we may uniquely identify $ \bs  \mu_{t}^{ATT,j}$   from $\{p_{\bs W}(\bs w;\bs\psi): \bs w\in\bs{\mathcal{W}}\}$ as
    \begin{equation}\label{eq:LTATT-identification}
\bs\mu_{t}^{ATT,j}=\sum_{\bs x\in\mathcal{X}}p^{j}_{T_0}(\bs x)\bigl(m^{j,1}_{t}(\bs x)-m^{j,0}_{t}(\bs x)\bigr).
\end{equation}
\end{proposition}

As a corollary, we may identify the ATT $\bs\mu_t^{ATT}$ defined in (\ref{ATT}) for each post-treatment period $t$.
\begin{corollary} Under Assumptions \ref{A-transition-Z}, \ref{A-sample}, \ref{A-Markov}, and \crossref{A-P-2}, for each post-treatment period $t \geq T_0+1$, we may uniquely identify the ATT $\bs\mu_t^{ATT}$ from $\{p_{\bs W}(\bs w;\bs\psi): \bs w\in\bs{\mathcal{W}}\}$.

\end{corollary}

The first-order Markov assumption in Assumption~\ref{A-Markov} can be straightforwardly relaxed to a higher (but finite) order Markov process, provided that $T_0$ and $T$ are sufficiently large, as the following proposition shows.
\begin{assumptionMk}[the $\ell$th-order Markov]\label{A-Markov-2} 
 For all $j=1,2,...,J$, conditional on $Z=j$ and $D=d$, $\{\bs X_{t}(d): t= 1,...,T\}$ follows a (non-stationary) $\ell$th-order Markov process, i.e., for $t=\ell+1,..., T$ and all $d\in \{0,1\}$,
$ \Pr(\bs X_{t}(d)\mid \{\bs X_s(d)\}_{s=1}^{t-1}, D=d,Z=j)=
 \Pr(\bs X_{t}(d)\mid \{\bs X_s(d)\}_{s=t-\ell}^{t-1},D=d,Z=j).$
\end{assumptionMk}

Under Assumptions \ref{A-Markov-2} and \ref{A-anticipation}, with abuse of notation, we redefine the parameter vector as
$\bs{\psi}=(\bs\pi,\bs\varphi^1,\ldots,\bs\varphi^J)\in \Theta_{\bs\psi}$.
Each component $\bs\varphi^j$ contains the initial distribution
$p^j_{\{\bs X_s(0)\}_{s=1}^\ell,D}(\cdot,\cdot)$.
Its transition-probability collections are
$\{p^j_{\bs X_t(0)|\{\bs X_s(0)\}_{s=t-\ell}^{t-1}}(\cdot|\cdot)\}_{t=\ell+1}^{T}$
and
$\{p^j_{\bs X_t(1)|\{\bs X_s(1)\}_{s=t-\ell}^{t-1}}(\cdot|\cdot)\}_{t=T_0+1}^{T}$,
where $p^j_{\{\bs X_s(0)\}_{s=1}^\ell,D}(\cdot,\cdot)$ and $p^j_{\bs X_t(d)|\{\bs X_s(d)\}_{s=t-\ell}^{t-1}}(\cdot|\cdot)$ for $d\in\{0,1\}$ are defined similarly to $p^j_{X_t(0),D}(\cdot,\cdot)$ and $p^j_{\bs X_t(d)|\bs X_{t-1}(d)}(\cdot|\cdot)$ but using $\{\bs X_s(0)\}_{s=1}^\ell$ and $\{\bs X_s(d)\}_{s=t-\ell}^{t-1}$ in place of $X_t(0)$ and $\bs X_{t-1}(d)$, respectively.

\begin{proposition}\label{P-2-k}
Suppose that Assumptions~\ref{A-sample}, \ref{A-Markov-2}, and \crossref{A-P-2-k} hold.
\revtext{If $T_0 \ge (k+1)\ell$, $T \ge 2(k+1)\ell$, and $J \le |\mathcal{X}|^{k\ell}$ for $k,\ell=1,2,\ldots$,}
then: (a) $\boldsymbol{\psi}$ is uniquely identified from 
$\{p_{\boldsymbol{W}}(\boldsymbol{w}; \boldsymbol{\psi}) : \boldsymbol{w} \in \boldsymbol{\mathcal{W}}\}$; (b) for each post-treatment period $t \geq T_0+1$ and  for all  $j\in \mathcal{J}$, we may uniquely identify $ \bs  \mu_{t}^{ATT,j}$   from $\{p_{\bs W}(\bs w;\bs\psi): \bs w\in\bs{\mathcal{W}}\}$ as 
$\bs  \mu_{t}^{ATT,j}
=\sum_{\bs x_{T_0-\ell+1}^{T_0}\in \mathcal{X}^\ell}\allowbreak
\Pr\bigl(\bs X_{T_0-\ell+1}^{T_0}=\bs x_{T_0-\ell+1}^{T_0}\mid D=1,\allowbreak Z=j\bigr)
\times \bigl\{\mathbb{E}\bigl[\bs X_{t} \mid \bs X_{T_0-\ell+1}^{T_0}=\bs x_{T_0-\ell+1}^{T_0},\allowbreak D = 1,\allowbreak Z=j\bigr]
- \mathbb{E}\bigl[\bs X_{t} \mid \bs X_{T_0-\ell+1}^{T_0}=\bs x_{T_0-\ell+1}^{T_0},\allowbreak D = 0,\allowbreak Z=j\bigr]\bigr\}$.

\end{proposition} 

\revtext{Under the second-order Markov assumption, for example, Proposition~\ref{P-2-k} with $\ell=2$ and $k=1$ implies that the LTATTs are identified when $J\le|\mathcal{X}|^{2}$, $T_0=4$, and $T=8$, and with $k=2$ when $J\le|\mathcal{X}|^{4}$, $T_0=6$, and $T=12$, provided that the stated regularity conditions are satisfied. Propositions~\ref{P-2} and~\ref{P-1} are the case $\ell=1$. For $k\ge 2$, Assumption~\crossref{A-P-2-k} requires in addition that every initial $\ell$-period history have positive probability in the control arm for every type, and its rank conditions are stated on $k\ell$-period histories; the panel length $T\ge 2(k+1)\ell$ grows with both indices.}

\begin{remark}[Choosing the number of latent types]\label{remark:model-selection}
Our identification results treat the number of latent types $J$ as
known; in practice it must be chosen. We recommend selecting $J$ by
the Bayesian information criterion (BIC) computed from the likelihood
in \eqref{mixture-2}. {BIC does not asymptotically underestimate
the number of mixture components \citep{leroux92as} and, under
regularity conditions, selects it consistently \citep{keribin00sankhya}.
Section~\crossref{sec:bic-postsel} of the Supplemental Appendix verifies
the penalty condition and establishes consistency directly for our
finite-support model, with mis-selection probability $o(n^{-1/2})$.}
 {Proposition~\crossref{prop:bic-postsel}(iii)} implies that
confidence intervals that ignore the selection step are asymptotically
valid. As a finite-sample check, each resample in the weighted bootstrap
recomputes the BIC over the candidate set and selects $J$; the fraction
of resamples that re-select the original-sample choice $\widehat{J}_{\mathrm{BIC}}$
measures the stability of the selection. In our applications, the BIC
re-selects the original-sample value of $J$ in $88.8\%$ of bootstrap
resamples for the Dodd--Frank Act (Section~\ref{sec:dodd-frank}), $98.6\%$
for Norway's patent reform (Section~\ref{sec:norwegian-patent-reform}), and
$100.0\%$ for the Americans with Disabilities Act of 1990 (ADA;
Section~\ref{sec:ada}), indicating that the selected value of $J$ is stable
across resamples.
\end{remark}

\begin{remark}[Diagnostic test for transition independence in pre-treatment periods]
\label{remark:testing-for-transition-independence}
\label{remark:placebo-test}
\label{remark:testing-for-transition-independence-Z}
In addition to the BIC model selection procedure in Remark~\ref{remark:model-selection}, a diagnostic test for transition independence in pre-treatment periods can be conducted,  analogous to the standard pre-trend tests in DiD designs \citep{Roth2022aer}.
Under
Assumptions~\ref{A-transition-Z} and~\ref{A-Markov}, the null becomes
type specific and conditions only on the most recent outcome:
\begin{align}
H_0\colon\,
&\Pr\left(\bs X_{t} = \bs x_{t} \mid \bs X_{t-1} = \bs x_{t-1},\,
  D=0,\, Z=j\right)\nonumber\\
&=\Pr\left(\bs X_{t} = \bs x_{t} \mid \bs X_{t-1} = \bs x_{t-1},\,
  D=1,\, Z=j\right)
\label{eq:conditional-mean-recursive-markovian}
\end{align}
for all $(\bs x_{t-1}, \bs x_{t}) \in \mathcal{X}^{2}$,
$t = 2,\ldots,T_{0}$, and $j \in \mathcal{J}$. Although $Z$ is
unobserved, both sides are identified and estimable from posterior
type probabilities
(Section~\crossref{sec:empirical-transition-probability-construction}).
The treated--control differences in these probabilities, plotted for
each $(\bs x_{t-1}, j)$ across $t$, give a diagnostic analogous to
pre-trends plots
\citep{freyaldenhoven2019pre,kahn2020promise,Rambachan2023}; uniform
confidence bands, bootstrapped jointly across periods, histories, and
types in an extension of \citet{CALLAWAY2021}, provide a formal diagnostic test for the violation of transition independence.\footnote{Reporting ATT estimates only when the diagnostic test for (\ref{eq:conditional-mean-recursive-markovian}) does not reject raises the usual pre-testing concern.
\citet{dechaisemartin2026pretest} show that, for convex symmetric pre-test statistics, naive confidence intervals are asymptotically conservative conditional on non-rejection under the null of valid specification, regardless of the asymptotic dependence between the pre-test statistic and the estimator, and that conditional coverage can exceed unconditional coverage under local violations of the null.}
Alternatively, the placebo version of
\eqref{eq:att-identification-vanilla} that sets the outcome period to
$T_0$ and shifts the conditioning set back one period equals zero
under Assumption~\ref{A-anticipation} and transition independence at
$T_0$.
\end{remark}

\section{Estimation}\label{sec:estimation}

We propose the following two-stage estimation procedure. For brevity, we assume a first-order Markov process, but the extension to higher-order Markov processes is straightforward.%

In the first stage, we estimate $\bs\psi$ by the maximum likelihood estimator

\begin{equation}\label{mle}
\hat{\bs\psi} = \arg\max_{\bs\psi\in \Theta_{\bs\psi}} \sum_{i=1}^n \log p_{\bs W}(\bs W_{i}; \bs \psi),
\end{equation}
where $p_{\bs W}(\bs w; \bs \psi)$ is the likelihood function defined in (\ref{mixture-2}), and $\{\bs W_i\}_{i=1}^n$ is a random sample of $n$ i.i.d. observations as described in Assumption \ref{A-sample}.

In the second stage, define the conditional type probabilities of $Z=j$ given $\bs W=\bs w$ as a function of $\bs\psi$ as
$\tau^j(\bs w; {\bs\psi}) := \pi^{j} p_{\bs W}^j(\bs w ;{ \bs \varphi}^{j} ) \big/ (\sum_{k=1}^J \pi^{k} p_{\bs W }^k(\bs w; {\bs \varphi}^{k}))$.
From the MLE $\hat{\bs\psi}$, we obtain the \textit{estimated type probabilities} of $Z=j$ for each observation as
\begin{equation}\label{posterior-estimate}
\hat\tau_i^j:= {\tau}^j(\bs W_i;\hat{\bs\psi})= \hat\pi^{j} p_{\bs W}^j(\bs W_i ;\hat{\bs \varphi}^{j})\big/\sum\nolimits_{k=1}^J \hat\pi^{k} p_{\bs W }^k(\bs W_i; \hat{\bs \varphi}^{k}).
\end{equation}
Note that the posterior type probability $\hat\tau_i^j$ conditions on the full observation vector $\bs W_i$, which includes both pre- and post-treatment outcomes. This is valid because $\bs W_i$ is an observed (realized) quantity for each unit: the posterior simply uses all available data to classify units into latent types, and no counterfactual quantities enter the computation.

Then, we consistently estimate the LTATTs~\eqref{LTATT} by the sample analogue estimator of (\ref{eq:LTATT-identification}) using $\{ \hat{\tau}^j_i\}_{i=1}^n$ as weights: for $t\geq T_0+1$,
\begin{equation}\label{eq:LTATT-estimator}
\hat{\bs\mu}_{t}^{ATT,j}=\sum_{\bs x\in\mathcal{X}}\hat p^{j}_{T_0}(\bs x)\bigl(\hat m^{j,1}_{t}(\bs x)-\hat m^{j,0}_{t}(\bs x)\bigr),
\end{equation}
where
\begin{equation}\label{eq:P-X0}
\hat p^{j}_{T_0}(\bs x)=\sum\nolimits_{i=1}^n \mathbf{1}\{\bs X_{iT_0} = \bs x, D_i=1\}\hat\tau_i^j\big/\sum\nolimits_{i=1}^n \mathbf{1}\{D_i=1\}\hat\tau_i^j
\end{equation}
and the estimated conditional expectation is
\begin{equation}\label{eq:E-Xt}
\hat m^{j,d}_{t}(\bs x)=\sum\nolimits_{i=1}^n \bs X_{it} \mathbf{1}\{\bs X_{iT_0}= \bs x,D_i=d\}\hat\tau_i^j\big/\sum\nolimits_{i=1}^n \mathbf{1}\{ \bs X_{iT_0}= \bs x,D_i=d\}\hat\tau_i^j.
\end{equation}

We also propose the following consistent estimator for the ATT as a weighted average of the estimated LTATTs~\eqref{eq:LTATT-estimator} across latent types:
\begin{align}\label{eq:ATT-estimator}
\hat{\bs\mu}_t^{ATT}=\sum_{j=1}^J \widehat{\Pr}(Z=j\mid D=1)\, \hat{\bs\mu}_t^{ATT,j},
\end{align} 
where $\widehat{\Pr}(Z=j\mid D=1)={\sum_{i=1}^n \mathbf{1}\{D_i=1\}\hat\tau_i^j}/{\sum_{i=1}^n \mathbf{1}\{D_i=1\}}$.

Let
$\bs I(\bs\psi^*)$ be  the Fisher information matrix for the MLE in (\ref{mle}).  We assume the following regularity conditions for the MLE $\hat{\bs\psi}$.
\begin{assumptionMLE}\label{A-MLE}  (a) $\Theta_{\bs\psi}$ is compact. All PMF entries in $\bs\varphi^j$ and mixture weights $\pi^j$ are bounded away from $0$ and $1$ by a common constant $\epsilon>0$.
(b) $\bs\psi^*$ lies in the interior of $\Theta_{\bs\psi}$. (c) the Fisher information matrix 
$\bs I(\bs\psi^*)$ is nonsingular. 
\end{assumptionMLE}

Let 
${\boldsymbol\theta}_t := \big(\mathrm{vec}({\bs\mu}_{t}^{ATT,1})^\top,\ldots,
\mathrm{vec}({\bs\mu}_{t}^{ATT,J})^\top,\ \mathrm{vec}({\bs\mu}_{t}^{ATT})^\top \big)^\top$ 
and 
${\boldsymbol\theta} := \big( {\boldsymbol\theta}_{T_0+1}^\top, \ldots, {\boldsymbol\theta}_{T}^\top \big)^\top$.  
Denote the true value of ${\boldsymbol\theta}$ by $\bs\theta^*$, and its two-step estimator, defined above, by $\hat{\bs\theta}$.
 
\begin{proposition}[Consistency and asymptotic normality of LTATT and ATT estimators]
\label{P-estimation}
Suppose that Assumptions \ref{A-sample}, \ref{A-Markov},  \crossref{A-P-2}, and \ref{A-MLE} hold.
Then: (a)  $\hat{\boldsymbol\theta} \xrightarrow{p} \boldsymbol\theta^*$. (b) $\sqrt{n}\,(\hat{\boldsymbol\theta}-\boldsymbol\theta^*)
    \Rightarrow \mathcal N(0,\bs V)$,
    where $\bs V = \E[\bs\Psi(\bs W;\bs\psi^*)\bs\Psi(\bs W;\bs\psi^*)^\top]$ is the variance of the combined influence function $\bs\Psi$, accounting for both the second-stage sampling variability and the first-stage estimation error propagated through the score. 
    Explicit expressions are given in the proof.
\end{proposition}

To obtain a consistent estimator $\hat{\boldsymbol{V}}$, we employ a nonparametric weighted bootstrap that yields consistent standard errors for the estimated LTATTs and ATTs while accounting for the dependence structure induced by the two-stage estimation procedure.
Further details on the EM algorithm and the weighted bootstrap are provided in Section~\crossref{sec:estimation-procedure-details} of the Supplemental Appendix.

\section{Applications}\label{sec:empirical-applications}

We illustrate the proposed methodology through three empirical applications, each highlighting a distinct limitation of conventional DiD with discrete outcomes. In \cite{charoenwong2019does}'s study of the Dodd-Frank Act, DiD counterfactuals fall below zero, producing out-of-bounds predictions that our transition-based approach avoids by construction. In \cite{hvide2018university}'s analysis of a Norwegian patent reform, large baseline differences in patenting rates between treated and control groups induce mean-reversion bias, leading DiD to overstate the reform's negative effect. In the ADA application, DiD fails to detect statistically significant employment effects because pre-treatment level differences mask the underlying transition dynamics; our method reveals significant negative effects operating through specific labor-force exit channels.

\subsection{Application to \cite{charoenwong2019does}}\label{sec:dodd-frank}

\cite{charoenwong2019does} investigate how regulatory jurisdiction affects the quality of investment advisor regulation by exploiting a unique policy shift under the Dodd-Frank Act implemented in 2012, which transferred oversight of midsize registered investment advisers (RIAs) from the SEC to state regulators. This 2012 reform left fiduciary standards unchanged but allowed the SEC to focus on other areas. \cite{charoenwong2019does} compare midsize RIAs transitioned to state oversight (treated) with the remaining RIAs under SEC oversight (control) and document an increase in client complaints among midsize advisers under state oversight, indicating a decrease in service quality following the reform.

We revisit \cite{charoenwong2019does} by comparing a DiD estimator and our proposed estimators. The main outcome of interest is the annual rate of client complaints, i.e., whether an RIA received a customer complaint, an indicator of declining service quality.\footnote{\revtext{We note that the regression specification in \cite{charoenwong2019does} contains RIA, year, and state--year fixed
effects together with a third-order polynomial in the number of investment-adviser
representatives employed by the RIA (Table 2, panel A).} This feature makes their reported estimates numerically different from DiD estimates reported here. For consistency with the other empirical examples, we implement a canonical DiD estimator with RIA-level fixed effects and time fixed effects only. The qualitative conclusions remain unchanged: the difference-in-differences estimates still suggest an increase in complaint rates following the reform.}

\begin{figure}[t]
    \centering
    \caption{ATT Estimates of the Dodd-Frank Act on Complaint Rates.}\label{fig:complaint-att-all}
    \includegraphics[width=0.75\textwidth]{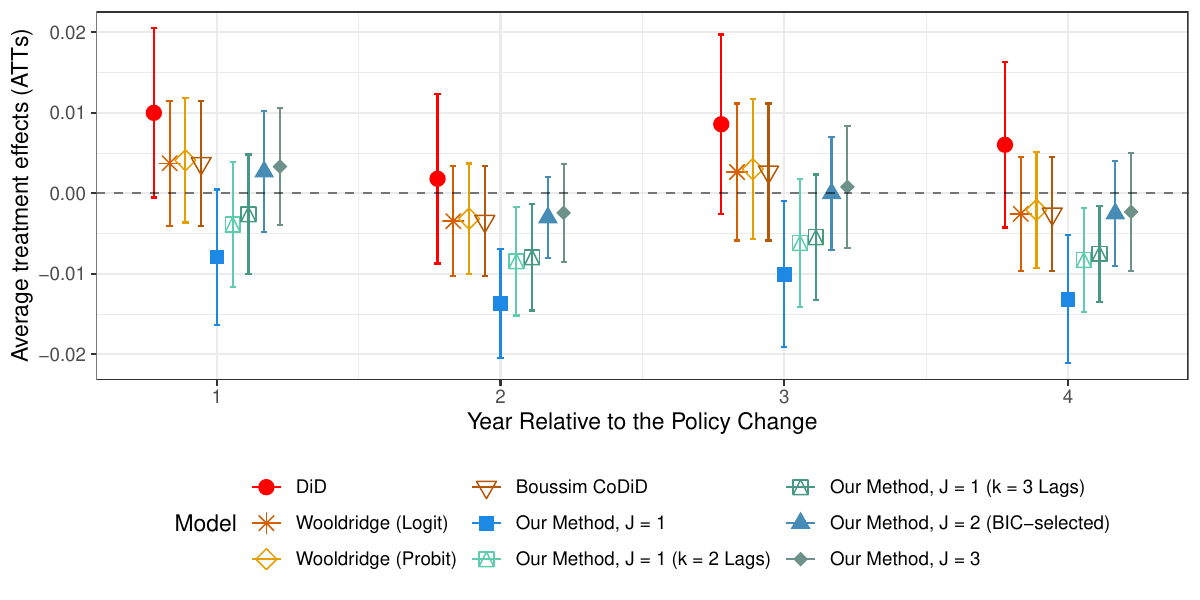}
    \fignote{ATT estimates of the 2012 Dodd-Frank reform on complaint rates. Estimators and markers are as in \autoref{fig:ada-att-simple}; transition-based specifications span $J \in \{1, 2, 3\}$ (BIC selects $J = 2$) and $k \in \{2, 3\}$. Vertical bars indicate 95\% pointwise bootstrap confidence intervals.}
\end{figure}

The counterfactual untreated average outcomes implied by the parallel trends assumption fall below zero as illustrated in \autoref{fig:complaint-counterfactual-intro}, indicating an out-of-bounds issue common in linear probability models. In contrast, our transition-based approach naturally restricts counterfactual untreated potential outcomes within the feasible range of probabilities by construction. This difference in the counterfactuals gives the ATT estimates opposite signs across methods: we find a slight decrease in complaint rates following the reform using our transition-based approach. \autoref{fig:complaint-att-all} summarizes the ATT estimates from both approaches across alternative model specifications, all indicating either improvements or null impacts on service quality, in contrast to the findings implied by DiD.

\subsection{Application to \cite{hvide2018university}}\label{sec:norwegian-patent-reform}

\cite{hvide2018university} examine the impact of the 2003 Norwegian reform that transferred one-third of patent rights from researchers, who previously held full ownership, to universities. \cite{hvide2018university} use a DiD design for the period 1995-2010 comparing inventors who were university researchers employed from 2000-2002 (treated) and non-university inventors (control). Their main TWFE specification (Equation 1 in \citealp{hvide2018university}) uses a dummy variable for annual patent applications as the outcome. They report a statistically significant estimate of -0.045 (Table 9 in \citealp{hvide2018university}), corresponding to a 4.5 percentage point decline in patenting probability for university researchers. We replicate their analysis using the same data to implement both the standard DiD design and our proposed method.

\begin{figure}[t]
    \centering
    \caption{ATT Estimates of the Norwegian Law Reform on Annual Patenting Rates.}\label{fig:patenting-att-all}
    \includegraphics[width=0.75\textwidth]{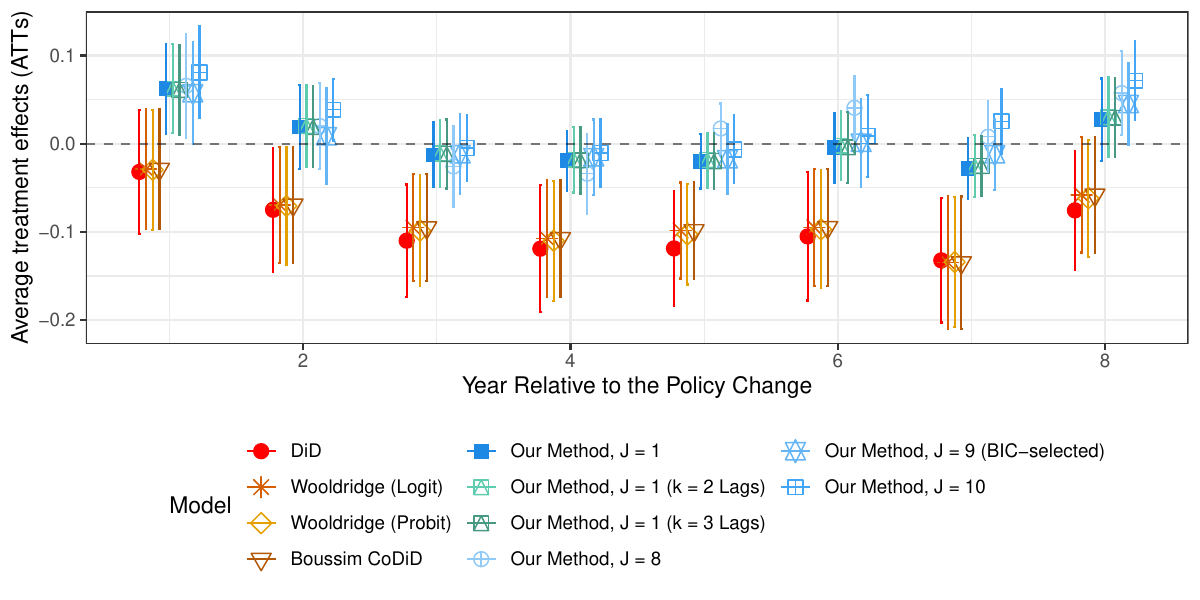}
    \fignote{ATT estimates of the Norwegian patent reform on annual patenting rates \citep{hvide2018university}. Estimators and markers are as in \autoref{fig:ada-att-simple}; transition-based specifications span $J \in \{1, 8, 9, 10\}$ (BIC selects $J = 9$) and $k \in \{2, 3\}$. Vertical bars indicate 95\% pointwise bootstrap confidence intervals.}
\end{figure}

\autoref{fig:patenting-att-all} compares ATT estimates across methods
and post-treatment periods. Both DiD and nonlinear DiD estimates reproduce the sizable
negative effect, whereas our transition-based estimates are close to
zero and, if anything, slightly positive.  For robustness, we condition on additional lagged outcomes ($k = 2$ and $k = 3$) and introduce latent types spanning the BIC selection ($J = 8$, $J = 9$, and $J = 10$) to account for potential violations of the transition independence assumption. Across these specifications the aggregate ATT (the average impact on patenting rates over post-reform periods) is positive throughout, including the model with BIC-selected $J = 9$. The estimates are far from---and, if anything, opposite in sign to---the sizable negative effect implied by DiD, reinforcing that the reform did not reduce patenting.

The large discrepancy between the DiD and our transition-based estimates is likely attributable to substantial pre-treatment differences in patenting rates between university inventors (treated) and non-university inventors (control), as illustrated in \autoref{fig:norway-trend}. In fact, in the year immediately preceding the reform, university inventors patented at nearly twice the rate of their non-university counterparts. As discussed in \autoref{sec:comparison-with-did}, this baseline disparity most likely biases DiD estimates: due to higher baseline rates, mean reversion would lead treated units to experience a steeper decline than controls in the absence of the reform, thereby violating the parallel trends assumption.

\subsection{Employment Effects of the Americans with Disabilities Act of 1990}\label{sec:ada}

\begin{figure}[!t]
    \centering
    \caption{Employment Rates Around the ADA of 1990.}\label{fig:ada-employment-trends}
    \includegraphics[width=0.75\textwidth]{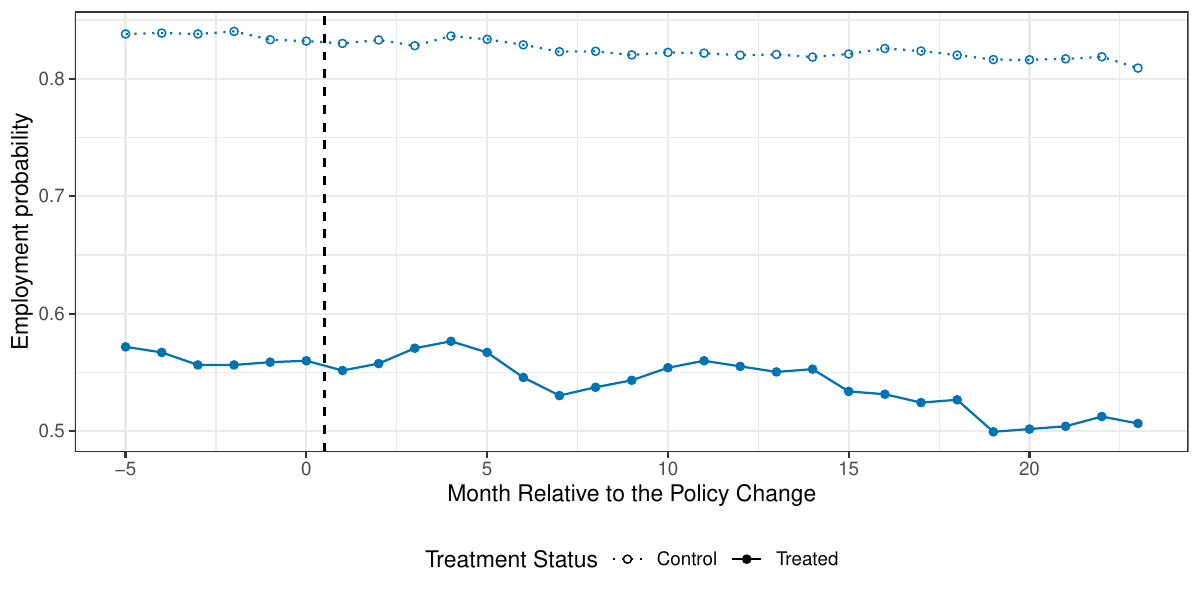}
    \fignote{Average monthly employment rates for individuals with disabilities (treated) and without disabilities (control) over January 1990 to May 1992. Period zero is June 1990, the last month before the ADA's July 1990 enactment, marked by the dashed line.}
\end{figure}

\begin{figure}[!t]
    \centering
    \caption{Difference in Labor Force Transition Probabilities Before the ADA.}\label{fig:ada-employment-transitions}
    \includegraphics[width=0.75\textwidth]{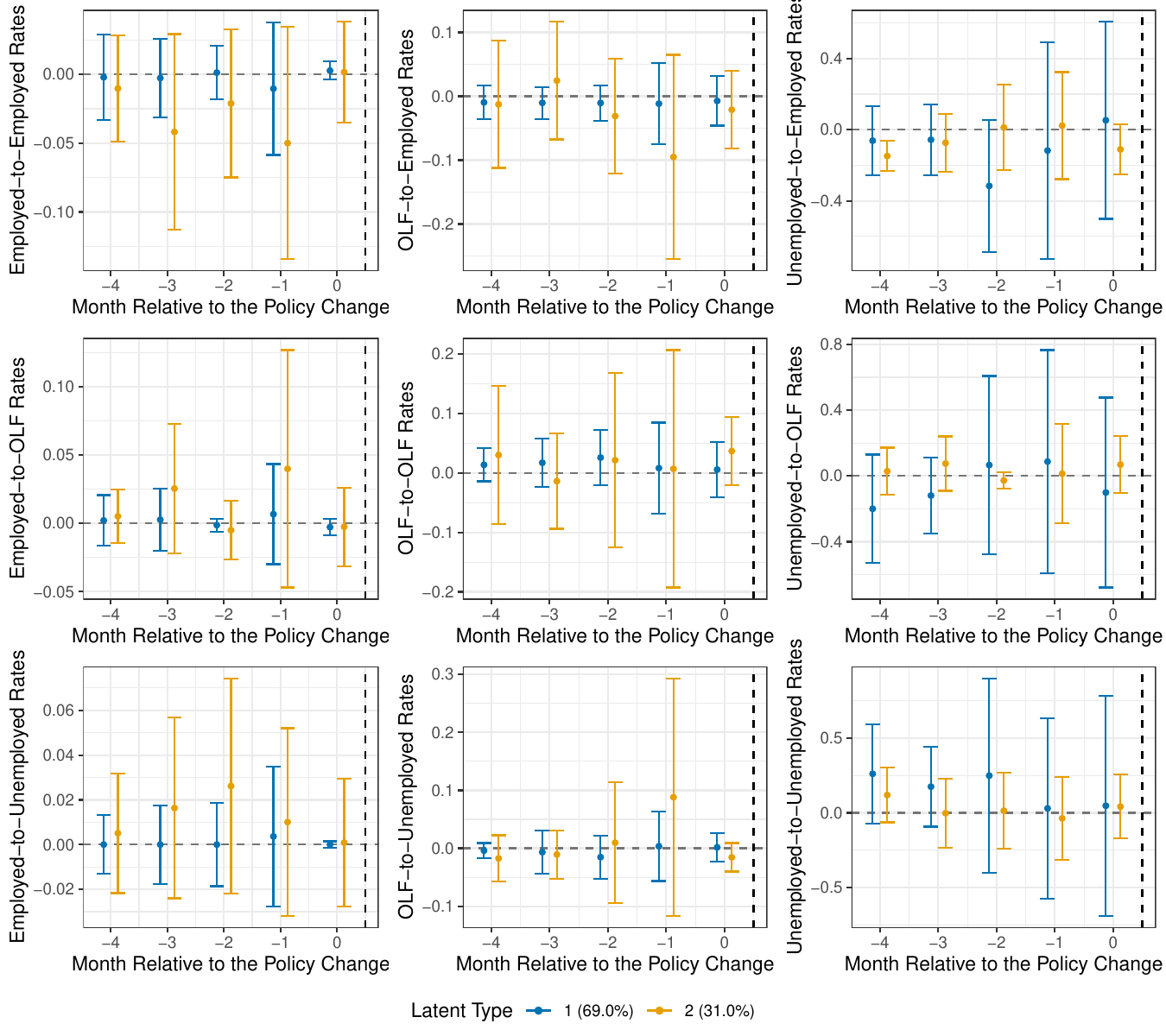}
    \fignote{Treated-minus-control differences in monthly labor force transition probabilities before the ADA, one panel per transition among employed, unemployed, and out-of-labor-force (OLF) states. The figure reports estimates from the BIC-selected model with $J = 2$ latent types. Period zero is the last pre-treatment period. Vertical bars are 95\% bootstrap uniform confidence intervals across periods within each transition pair.}
\end{figure}

\begin{figure}[!t]
    \centering
    \caption{ATT Estimates of the ADA on Employment.}\label{fig:ada-att-all}
    \includegraphics[width=0.75\textwidth]{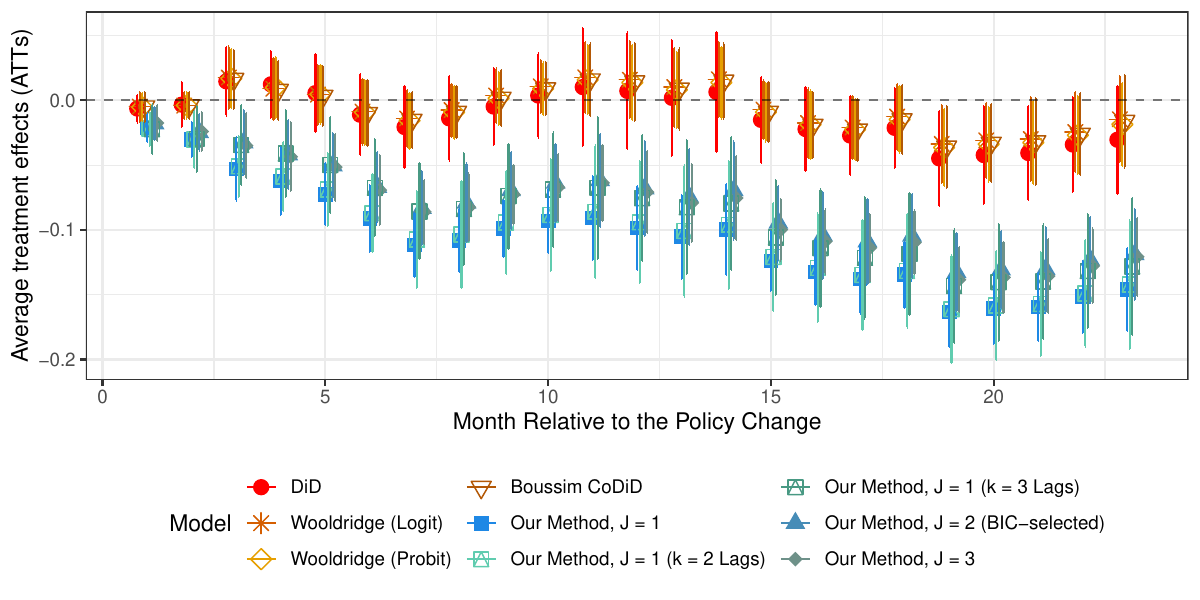}
    \fignote{ATT estimates of the ADA on employment. Estimators and markers are as in \autoref{fig:ada-att-simple}; transition-based specifications span $J \in \{1, 2, 3\}$ (BIC selects $J = 2$) and $k \in \{2, 3\}$. Vertical bars indicate 95\% pointwise bootstrap confidence intervals.}
\end{figure}

In this example, we examine the effects of the Americans with Disabilities Act (ADA) signed into law in 1990 on employment comparing  working-age individuals with work-related disability (treated) and individuals without work-related disability (control). The ADA was enacted to prohibit discrimination against individuals with disabilities in employment, public services, and accommodations. Title I of the ADA specifically mandated that employers with 15 or more employees could not discriminate against qualified individuals with disabilities in hiring, advancement, or discharge decisions, and required that reasonable accommodations be provided in the workplace.

We follow previous event-study analyses of the ADA's employment effects \citep{acemoglu2001consequences,lise2023revisiting}.  The outcome is
monthly labor force status, which takes one of three mutually
exclusive values: employed, unemployed, and out of the labor force
(OLF).  We use data from Rotation Group 1 of the 1990 Survey of Income and Program Participation (SIPP) panel \citep{SIPP1990}, which provides monthly observations from January 1990 through May 1992. Taking the ADA signing in July 1990 as the treatment date yields six pre-treatment and up to 22 post-treatment monthly observations per individual, thereby capturing the full labor force transition dynamics before and after the ADA. Since the ADA was implemented nationwide without explicit state-level variation, we compare individuals with and without disabilities, as in difference-in-differences strategies used in prior research \citep{acemoglu2001consequences,lise2023revisiting}. Following \citet{lise2023revisiting}, we
restrict the sample to adults aged 21--58 and classify individuals
reporting a work-limiting disability as disabled. We refer to the disabled
as the treated group and to the non-disabled as the control group.

We first note that the treated and control groups had different baseline employment rates as well as diverging trends even before the introduction of the ADA, as illustrated in \autoref{fig:ada-employment-trends}. For example, in the pre-treatment period, the employment rate among individuals with disabilities was 27.5 percentage points lower than that of individuals without disabilities in our sample. This pre-existing disparity raises concerns about the validity of the parallel trends assumption due to potential mean reversion characterized in Proposition~\ref{proposition:dbt-vs-did}. Beyond employment status (employed versus the others), similar level differences are present across all possible labor force states, as shown in Figure~\crossref{fig:ada-labor-force-status-trends} of the Supplemental Appendix.

In contrast, treated and control groups exhibit similar pre-treatment labor force transition probabilities under the BIC-selected model, supporting the plausibility of transition independence (\autoref{fig:ada-employment-transitions}). In addition to the BIC selection procedure, we provide a diagnostic test to assess transition independence before the ADA, by computing treated--control differences in transition probabilities for all nine possible transition pairs and reporting bootstrap uniform confidence intervals as suggested in Remark \ref{remark:testing-for-transition-independence-Z}. Across all specifications ($J=1,2,3$), the estimated pre-treatment differences in transition probabilities are close to zero and are negligible. Specifications with alternative values of $J$ yield the same pattern (Figures~\crossref{fig:ada-employment-transitions-J2} and \crossref{fig:ada-employment-transitions-J3} in the Supplemental Appendix).

Our transition-based estimates indicate a statistically significant negative impact on employment across all specifications ($J=1$, $J=2$, and $J=3$), whereas the standard and nonlinear DiD models, likely biased due to differences in pre-treatment employment levels, do not yield statistically significant negative estimates. The ATT estimates across different methods and specifications are provided in \autoref{fig:ada-att-all}. The aggregate ATT is significantly negative for the BIC-selected $J=2$, as well as for models with different $J$, with bootstrap confidence bands excluding zero in each case. Overall, the results indicate that the ADA could have resulted in an unintended short-run negative impact on employment for individuals with disabilities.

\begin{revblock}
A distinctive advantage of the transition-based framework is its ability to decompose treatment effects into specific transition channels via the flow decomposition in \eqref{eq:flow-decomposition}. \autoref{fig:ada-decomposition-intro} applies this decomposition using the $J=2$ model to identify how the employment declines arise over periods, namely, which cumulative inflows and
outflows account for the employment decline.  The dominant channel is
increased cumulative transitions from employment into OLF among
individuals with disabilities after the ADA; reduced inflows from OLF
into employment play a secondary role, and flows through unemployment
contribute little. The ADA's adverse employment effect thus operates
through the participation margin, a mechanism invisible to standard
DiD, which estimates only the net change in outcome levels.
\end{revblock}

\FloatBarrier
 
\appendix
\section{Proofs}

\subsection{Proof of Proposition \ref{proposition:dbt-vs-did}}
By definition,
$\bs\mu^{\text{DiD}}_t=\E[\bs X_t-\bs X_{T_0}\mid D=1]-\E[\bs X_t-\bs X_{T_0}\mid D=0]$\allowbreak{} and
$\bs\mu^{\text{ATT}}_t=\E[\bs X_t(1)-\bs X_t(0)\mid D=1]$.
By Assumption~\ref{A-anticipation},
$\bs X_t=\bs X_t(0)$, $t\geq T_0+1$, for $D=0$ and $\bs X_{T_0}=\bs X_{T_0}(0)$ for all units, so
\begin{equation}\label{eq:diff}
\bs\mu^{\text{DiD}}_t-\bs\mu^{\text{ATT}}_t
= \E[\bs X_t(0)-\bs X_{T_0}(0)\mid D=1]
 -\E[\bs X_t(0)-\bs X_{T_0}(0)\mid D=0]. %
\end{equation} 
Applying the Law of Iterated Expectations to the first term,
$\E[\bs X_t(0)\mid D=1]=\E\!\left[\E[\bs X_t(0)\mid \bs X_1^{T_0},D=1]\mid D=1\right]$.
By  Assumption~\ref{A-transition},
\begin{equation}\label{eq:diff-2}
\E[\bs X_t(0)\mid \bs X_1^{T_0},D=1]
=\E[\bs X_t(0)\mid \bs X_1^{T_0},D=0]
=\E[\bs X_t\mid \bs X_1^{T_0},D=0].
\end{equation}
Using (\ref{eq:diff-2}) and the same argument for $D=0$, equation (\ref{eq:diff}) becomes
$\bs\mu^{\text{DiD}}_t-\bs\mu^{\text{ATT}}_t
= \E[\E[\bs X_t\mid \bs X_1^{T_0},D=0]-\bs X_{T_0}\mid D=1]\allowbreak
 -\E[\E[\bs X_t\mid \bs X_1^{T_0},D=0]-\bs X_{T_0}\mid D=0]$.
Under Assumption~\ref{A-overlap}, expand the above equation over the support $\mathcal{X}^{T_0}$:
\begin{multline*}
\bs\mu^{\text{DiD}}_t-\bs\mu^{\text{ATT}}_t
= \sum_{\bs x_1^{T_0}\in \mathcal{X}^{T_0}}
\underbrace{\E[\bs X_t-\bs x_{T_0}\mid \bs X_1^{T_0}=\bs x_1^{T_0},D=0]}_{\text{control mean trend}}\\
\times\left\{
\Pr\left(\bs X_1^{T_0}=\bs x_1^{T_0}\mid D=1\right)
-\Pr\left(\bs X_1^{T_0}=\bs x_1^{T_0}\mid D=0\right)
\right\}.
\end{multline*}
This proves the stated result. $\qedsymbol$

\subsection{Proof of Proposition \crossref{P-2}}
For brevity, we prove only the ``boundary'' case $T_0=k+1$, $T=2(k+1)$,
and $J=|\mathcal{X}|^{k}$; the proof extends to smaller $J$ and to
larger $T_0$ and $T$.
Partition $(\bs X_1,...,\bs X_T)$ as $(\bs X_1,...,\bs X_T) = (\xi_1,\xi_2,\xi_3,\xi_4)$, where 
\begin{align}
&\xi_1:=\bs X_1^k,\  \xi_2:= \bs X_{k+1},\  \xi_3=\bs X_{k+2}, \text{  and  }  \xi_4:=\bs X_{k+3}^{T}.\label{eq:xi}
\end{align} 
Note that, when $k=1$, we have $\xi_t=\bs X_t$ for $t=1,...,4$. For $d\in\{0,1\}$, define the corresponding potential outcomes $\xi_1(d):=\bs X_1^k (d)$, $\xi_2(d):= \bs X_{k+1}(d)$, $\xi_3(d)=\bs X_{k+2}(d)$, and $\xi_4(d):=\bs X_{k+3}^{T}(d)$. 
Because $\bs X_t=\bs X_t(0)$ for $t=1,2,...,T$ for all control units with $D=0$, we have $\xi_j = \xi_j(0)$ for $j=1,2,3,4$ for all units with $D=0$. Then, the probability mass function of $(\xi_2,\xi_3,\xi_4)$ given $(\xi_1,D)=(\xi_1,0)$ is written as
  \begin{align}\label{mixture-3}
&p_{\xi_2,\xi_3,\xi_4|\xi_1,D}(\xi_2,\xi_3,\xi_4|\xi_1,0)%
 = \sum_{j=1}^J\lambda^j_1(\xi_2 |\xi_1)\lambda^j_2(\xi_3|\xi_2)\lambda^j_{3}(\xi_4|\xi_3),
 \end{align} 
where   $\tau^j_{\xi_1(0),D}(\xi_1,d):=  \frac{ \pi^j p^j_{\xi_1,D}(\xi_1,d) }{ \sum_{k=1}^J \pi^k p^k_{\xi_1,D}(\xi_1,d) }, \quad \lambda^j_1(\xi_2 |\xi_1):=\tau_{\xi_1(0),D}^j(\xi_1,0) p^j_{\xi_2(0)|\xi_1(0)}(\xi_2|\xi_1)$, $ \lambda^j_2(\xi_3|\xi_2):=p^j_{\xi_3(0)|\xi_2(0)}(\xi_3|\xi_2)$, and $\lambda^j_{3}(\xi_4|\xi_3):=p^j_{\xi_4(0)|\xi_3(0)}(\xi_4|\xi_3)$.
 Evaluating (\ref{mixture-3})  at $(\xi_1,\xi_2,\xi_3,\xi_4)=( a ,\xi_2,\xi_3, b )$ and  $p_{\xi_2,\xi_3|\xi_1,D}(\xi_2,\xi_3|\xi_1,0)= $ \\ $\sum_{j=1}^J  \tau_{\xi_1(0),D}^j(\xi_1,0) p^j_{\xi_2|\xi_1}(\xi_2|\xi_1)  p^j_{\xi_3(0)|\xi_2(0)}(\xi_3|\xi_2)$ at $(\xi_1,\xi_2,\xi_3)=( a ,\xi_2,\xi_3)$ gives
\begin{align}
p_{\xi_2,\xi_3,\xi_4|\xi_1,D}(\xi_2,\xi_3,b|a,0)&= \sum_{j=1}^J \lambda^j_1(\xi_2|a) \lambda^j_2(\xi_3|\xi_2)\lambda^j_3(b|\xi_3) \quad \text{and}\label{p-w}\\
p_{\xi_2,\xi_3|\xi_1,D}(\xi_2,\xi_3|a,0) & =\sum_{j=1}^J  \lambda^j_1(\xi_2|a) \lambda^j_2(\xi_3|\xi_2). \label{p-w-2}
\end{align}
Denote $q_{\xi_2,\xi_3}( a , b ) :=p_{\xi_2,\xi_3,\xi_4|\xi_1,D}(\xi_2,\xi_3,b|a,0)$ and $\bar q_{\xi_2, \xi_3}( a) :=p_{\xi_2,\xi_3|\xi_1,D}(\xi_2,\xi_3|a,0)$.

 Evaluating (\ref{p-w}) at $ a = a _1,...,a_J$ and $ b =  b _1,..., b _{J-1}$ gives $J(J-1)$ equations while evaluating (\ref{p-w-2}) at $ a = a _1,...,a_J$ gives $J$ equations.  
We collect  these $J(J-1)+J=J^2$ equations as
\begin{equation}
\bs Q_{\xi_2,\xi_3} =  {\bs L}_{\xi_3}  {\bs D}_{\xi_3|\xi_2} \bar {\bs L}_{\xi_2}^\top , \label{decomp}
\end{equation}
where 
\begin{equation}
\begin{aligned}
&\bs Q_{\xi_2,\xi_3}:= 
\begin{pmatrix}
\bar q_{\xi_2,\xi_3}( a _1) &\bar q_{\xi_2,\xi_3}( a _2) &\cdots &\bar q_{\xi_2,\xi_3}( a _J) \\
q_{\xi_2,\xi_3}( a _1, b _1) &q_{\xi_2,\xi_3}( a _2, b _1) &\cdots &q_{\xi_2,\xi_3}( a _J, b _1) \\
\vdots & \vdots & \ddots & \vdots \\
q_{\xi_2,\xi_3}( a _1, b _{J-1}) &q_{\xi_2,\xi_3}( a _2, b _{J-1}) &\cdots &q_{\xi_2,\xi_3}( a _J, b _{J-1}) 
\end{pmatrix},
\end{aligned} \label{L}
\end{equation} 
 \begin{equation}
\begin{aligned}  
&  
\bar{\bs L}_{\xi_2}:=
\begin{pmatrix}
\lambda^1_{1}(\xi_2|a_1)&\cdots&\lambda^J_{1}(\xi_2|a_1)\\
\vdots & \ddots & \dots  \\
\lambda^1_{1}(\xi_2|a_J)&\cdots&\lambda^J_{1}(\xi_2|a_J)
\end{pmatrix},\quad
 {\bs L}_{\xi_3}:=
\begin{pmatrix}
1& \cdots &1\\
\lambda^1_{3}(b_1|\xi_3) & \cdots &\lambda^J_{3}(b_1|\xi_3)\\
\vdots & \ddots &  \vdots \\
\lambda^1_{3}(b_{J-1}|\xi_3) & \cdots &\lambda^J_{3}(b_{J-1}|\xi_3)
\end{pmatrix},
\end{aligned} \label{LL}
\end{equation}  
and ${\bs D}_{\xi_3|\xi_2} := \text{diag}\left(\lambda^1_2(\xi_3|\xi_2),...,\lambda^J_2(\xi_3|\xi_2)\right)$.  Define 
$ \tau_{\xi_1(0),\xi_2(0),D}^j(\xi_1,\xi_2,1):=  $ \\ $ \tau_{\xi_1(0),D}^j(\xi_1,1) p^j_{\xi_2(0)|\xi_1(0)}(\xi_2|\xi_1).$
We make the following assumption.

\begin{assumptionID}\label{A-P-2} (a) There exists a value $\xi_3^*$ that satisfies the following condition: 
 For every $\xi_3\in \mathcal{X}$, we can find $( \check\xi_2,\bar \xi_2,\bar \xi_3)\in\mathcal{X}^3$, $(a_1,...,a_{J})\in \mathcal{X}^{kJ}$ and $(b_1,...,b_{J-1})\in \mathcal{X}^{(T-(k+2))(J-1)}$ such that (i)   $\bar{\bs L}_{\check\xi_2}$, $\bar{\bs L}_{\bar \xi_2}$, ${\bs L}_{ \xi_3}$, ${\bs L}_{ \xi_3^*}$, and ${\bs L}_{\bar \xi_3}$ are non-singular; (ii) all the diagonal elements of $\bs D_{ \xi_3, \xi_3} $ defined in  (\ref{D-z}) with $ \xi_3^*= \xi_3$ take  distinct values.
(b)  For every $(\xi_2, \xi_3, \xi_4)\in \mathcal{X}^{T-k}$,  $p^j_{\xi_3(0)|\xi_2(0)}(\xi_3|\xi_2)> 0$ and $p^j_{\xi_4(0)|\xi_3(0)}(\xi_4|\xi_3)> 0$ for $j=1,...,J$.  \revtext{(c) For every $x_1\in\mathcal{X}$: (i) if $k\ge 2$,
$p^j_{\bs X_1(0),D}(x_1,0)>0$ for $j=1,\ldots,J$; and (ii) there exist
$\{c_l\}_{l=1}^J$ with $c_l=(c_{l,2},\ldots,c_{l,k+1})\in\mathcal{X}^{k}$
such that, writing $c_{l,1}:=x_1$, the $J\times J$ matrix
$$
\begin{pmatrix}
 p^1_{\bs X_2^{k+1}(0)|\bs X_1(0)}(c_1|x_1)& p^2_{\bs X_2^{k+1}(0)|\bs X_1(0)}(c_1|x_1) &\cdots& p^J_{\bs X_2^{k+1}(0)|\bs X_1(0)}(c_1|x_1)\\
 \vdots&\vdots&\ddots &\vdots\\
 p^1_{\bs X_2^{k+1}(0)|\bs X_1(0)}(c_J|x_1) &p^2_{\bs X_2^{k+1}(0)|\bs X_1(0)}(c_J|x_1)&\cdots&p^J_{\bs X_2^{k+1}(0)|\bs X_1(0)}(c_J|x_1)
\end{pmatrix}$$
is non-singular, where
$p^j_{\bs X_2^{k+1}(0)|\bs X_1(0)}(c_l|x_1):=\prod_{t=2}^{k+1}p^j_{\bs X_t(0)|\bs X_{t-1}(0)}(c_{l,t}|c_{l,t-1})$
is the type-$j$ probability of the $k$-period path $c_l$ given
$\bs X_1(0)=x_1$. When $k=1$, $c_l\in\mathcal{X}$ and the matrix has
elements $p^j_{\xi_2(0)|\xi_1(0)}(c_l|\xi_1)$.} (d) For each value of $\xi_2\in \mathcal{X}$, there exists $\{e_j\}_{j=1}^J\in\mathcal{X}^{kJ}$ such that
$$
\begin{pmatrix}
 \tau_{\xi_1(0),\xi_2(0),D}^1(e_1,\xi_2,1)& \tau_{\xi_1(0),\xi_2(0),D}^2(e_1,\xi_2,1)& \cdots &  \tau_{\xi_1(0),\xi_2(0),D}^J(e_1,\xi_2,1)\\
 \vdots & \vdots &  \ddots & \vdots\\
 \tau_{\xi_1(0),\xi_2(0),D}^1(e_J,\xi_2,1)& \tau_{\xi_1(0),\xi_2(0),D}^2(e_J,\xi_2,1)& \cdots &  \tau_{\xi_1(0),\xi_2(0),D}^J(e_J,\xi_2,1)
\end{pmatrix}
$$
is non-singular.
\end{assumptionID}
 
We fix the value of  $\xi_3^*$  and, for each $ \xi_3$, we choose $( \check\xi_2,\bar \xi_2,\bar \xi_3)$, $(a_1,...,a_{J})$, and $(b_1,...,b_{J-1})$ as stated in Assumption \ref{A-P-2}(a). 
Evaluating (\ref{decomp}) at four different points, $({\xi}_3, \check\xi_2)$, $(\bar{\xi}_3,\check\xi_2)$, $(\bar {\xi}_3, \bar{\xi}_2)$, and $(  {\xi}_3^*,\bar {\xi}_2)$ gives
\[
\begin{aligned}
\bs Q_{ \check\xi_2,{\xi}_3} &= {\bs L}_{{\xi}_3} {\bs D}_{{\xi}_3|\check\xi_2} \bar {\bs L}_{\check\xi_2}^\top,\quad 
\bs Q_{\check\xi_2,\bar{\xi}_3}  = {\bs L}_{ \bar{\xi}_3} {\bs D}_{\bar {\xi}_3|  \check\xi_2} \bar {\bs L}_{\check\xi_2}^\top, \\
\bs Q_{ \bar{\xi}_2,\bar {\xi}_3} &= {\bs L}_{\bar {\xi}_3} {\bs D}_{\bar {\xi}_3|\bar{\xi}_2} \bar {\bs L}_{\bar{\xi}_2}^\top,\quad
 \bs Q_{\bar {\xi}_2,{\xi}_3^*}  = {\bs L}_{{\xi}_3^*}{\bs D}_{ {\xi}_3^*|\bar {\xi}_2} \bar {\bs L}_{\bar {\xi}_2}^\top.
\end{aligned}
\]  
Then, following the identification argument in  \cite{carroll10jns}, under Assumption \ref{A-P-2}(a)(b), we have%
\begin{align}
 A _{ \xi_3, \xi_3^*}  & :=  \bs Q_{\check\xi_2,{\xi}_3} (\bs Q_{\check\xi_2,\bar{\xi}_3})^{-1} \bs Q_{\bar{\xi}_2,\bar {\xi}_3} ( \bs Q_{\bar {\xi}_2,{\xi}_3^*} )^{-1}   = {\bs L}_{{\xi}_3}  {\bs D}_{ {\xi}_3,  {\xi}_3^*} {\bs L}_{{\xi}_3^*} ^{-1}, \label{A_z}\\
\text{ where }& {\bs D}_{ {\xi}_3,  {\xi}_3^*} : = \bs D_{{\xi}_3| \check\xi_2} (\bs D_{\bar{\xi}_3|\check\xi_2})^{-1} \bs D_{\bar {\xi}_3| \bar{\xi}_2} ( \bs D_{{\xi}_3^*|\bar {\xi}_2} )^{-1}. \label{D-z}
\end{align}

We  first identify ${\bs L}_{{\xi}_3}$ for all $ \xi_3\in  \mathcal{X}$ up to an unknown permutation matrix. Evaluating (\ref{A_z}) and (\ref{D-z}) at $ \xi_3^*= \xi_3$, we have  
 $\bs A_{ \xi_3, \xi_3}  {\bs L}_{{\xi}_3}= {\bs L}_{{\xi}_3} {\bs D}_{{\xi}_3, {\xi}_3}$. Because $\bs A_{\xi_3,\xi_3}$ has $J$ distinct eigenvalues under
Assumption \ref{A-P-2}(a), the eigenvalues of $\bs A_{\xi_3,\xi_3}$
determine the diagonal elements of $\bs D_{\xi_3,\xi_3}$ up to their
ordering, while the right eigenvectors of $\bs A_{\xi_3,\xi_3}$ determine
the columns of $\bs L_{\xi_3}$ up to multiplicative constants and the
ordering of its columns. Namely, collecting the right eigenvectors of
$\bs A_{\xi_3,\xi_3}$ into a matrix in descending order of their
eigenvalues, we obtain a matrix $\bs B$ that can be written as
$\bs B = \bs L_{\xi_3}\Delta_{\xi_3}\bs C$, where $\Delta_{\xi_3}$ is an
unknown permutation matrix and $\bs C$ is some diagonal matrix with
non-zero diagonal elements. Because
$\bs A_{\xi_3,\xi_3}\bs L_{\xi_3} = \bs L_{\xi_3}\bs D_{\xi_3,\xi_3}$ and
diagonal matrices commute, $\bs B$ satisfies
$\bs A_{\xi_3,\xi_3}\bs B = \bs L_{\xi_3}\bs D_{\xi_3,\xi_3}\Delta_{\xi_3}\bs C
= \bs L_{\xi_3}\Delta_{\xi_3}\bs C\,\tilde{\bs D}_{\xi_3,\xi_3}
= \bs B\,\tilde{\bs D}_{\xi_3,\xi_3}$, where
$\tilde{\bs D}_{\xi_3,\xi_3} := \Delta_{\xi_3}^{-1}\bs D_{\xi_3,\xi_3}\Delta_{\xi_3}$
is the diagonal matrix whose diagonal elements are the eigenvalues of
$\bs A_{\xi_3,\xi_3}$ in descending order.

We can determine the diagonal matrix $\bs C\tilde{\bs D}_{\xi_3,\xi_3}$ from
the first row of $\bs A_{\xi_3,\xi_3}\bs B = \bs B\tilde{\bs D}_{\xi_3,\xi_3}
= \bs L_{\xi_3}\Delta_{\xi_3}\bs C\tilde{\bs D}_{\xi_3,\xi_3}$ because the
first row of $\bs L_{\xi_3}\Delta_{\xi_3}$ is a vector of ones: the first
row of $\bs L_{\xi_3}$ is a vector of ones, and post-multiplication by a
permutation matrix only re-orders its entries. Then,
$\bs L_{\xi_3}\Delta_{\xi_3}$ is determined from $\bs A_{\xi_3,\xi_3}\bs B$
and $\bs C\tilde{\bs D}_{\xi_3,\xi_3}$ as
$\bs L_{\xi_3}\Delta_{\xi_3} = \bs A_{\xi_3,\xi_3}\bs B(\bs C\tilde{\bs D}_{\xi_3,\xi_3})^{-1}$
in view of
$\bs A_{\xi_3,\xi_3}\bs B = \bs L_{\xi_3}\Delta_{\xi_3}\bs C\tilde{\bs D}_{\xi_3,\xi_3}$,
where $\bs C\tilde{\bs D}_{\xi_3,\xi_3}$ is invertible because the diagonal
elements of $\bs D_{\xi_3,\xi_3}$ are non-zero under Assumption
\ref{A-P-2}(b). Repeating the above argument for all values of
$\xi_3\in\mathcal{X}$, the eigenvalue decomposition algorithm identifies
the matrices

 Next, we identify permutation matrices that re-arrange ${\bs L}_{{\xi}_3} \Delta_{{\xi}_3} $ into a common order of latent types across different values of $ \xi_3$, following the identification argument in \cite{HigginsJochmans21}. Pre- and post-multiplying (\ref{A_z}) by $\tilde {\bs L}_{{\xi}_3}^{-1}$ and $\tilde {\bs L}_{ \xi_3^*}$, respectively, we obtain 
$\tilde {\bs D}_{ \xi_3, {\xi}_3^*} := \tilde {\bs L}_{{\xi}_3}^{-1} \bs A_{ \xi_3, {\xi}_3^*} \tilde {\bs L}_{ \xi_3^*} = \Delta_{{\xi}_3}^{-1} {\bs D}_{ \xi_3, {\xi}_3^*} \Delta_{ \xi_3^*}= \left( \Delta_{{\xi}_3}^{-1} \Delta_{ \xi_3^*} \right) \left( \Delta_{ \xi_3^*}^{-1} {\bs D}_{ \xi_3, {\xi}_3^*} \Delta_{ \xi_3^*} \right)$,
where the last equality inserts the identity matrix $\Delta_{ \xi_3^*} \Delta_{ \xi_3^*}^{-1}$. Define the relative permutation matrix $\bs P_{\xi_3} := \Delta_{{\xi}_3}^{-1} \Delta_{ \xi_3^*}$ and the diagonal matrix $\bs D^*_{\xi_3,\xi_3^*} := \Delta_{ \xi_3^*}^{-1} {\bs D}_{ \xi_3, {\xi}_3^*} \Delta_{ \xi_3^*}$. We observe that $\tilde {\bs D}_{ \xi_3, {\xi}_3^*} = \bs P_{\xi_3} \bs D^*_{\xi_3,\xi_3^*}$.
Because $\bs P_{\xi_3}$ is a permutation matrix, $\tilde {\bs D}_{ \xi_3, {\xi}_3^*}$ is a column-permuted diagonal matrix. Specifically, the $j$-th column of $\tilde {\bs D}_{ \xi_3, {\xi}_3^*}$ contains exactly one non-zero element, which corresponds to a diagonal element of $\bs D^*_{\xi_3,\xi_3^*}$. Therefore, each diagonal element of $\bs D^*_{\xi_3,\xi_3^*}$ is identified by the sum of the elements in the corresponding column of $\tilde {\bs D}_{ \xi_3, {\xi}_3^*}$. Consequently, the matrix $\bs D^*_{\xi_3,\xi_3^*} = \Delta_{ \xi_3^*}^{-1} {\bs D}_{ \xi_3, {\xi}_3^*} \Delta_{ \xi_3^*}$ is identified.

With $\tilde {\bs D}_{ \xi_3, {\xi}_3^*}$ and $\bs D^*_{\xi_3,\xi_3^*}$ identified, we can identify the relative permutation matrix as $\bs P_{\xi_3} = \tilde {\bs D}_{ \xi_3, {\xi}_3^*} (\bs D^*_{\xi_3,\xi_3^*})^{-1}$. Substituting the definitions back into this expression, we have identified $\Delta_{{\xi}_3}^{-1} \Delta_{ \xi_3^*} = \tilde {\bs D}_{ \xi_3, {\xi}_3^*} ( \Delta_{ \xi_3^*}^{-1} {\bs D}_{ \xi_3, {\xi}_3^*} \Delta_{ \xi_3^*} )^{-1}$.
Finally, we identify ${\bs L}_{{\xi}_3}$ up to a common permutation matrix $\Delta_{ \xi_3^*}$ (which does not depend on $\xi_3$ under Assumption \ref{A-P-2}(a)) by aligning the columns of $\tilde {\bs L}_{{\xi}_3}$ using the identified relative permutation as
\begin{equation}\label{L-star}
{\bs L}_{{\xi}_3}^* := {\bs L}_{{\xi}_3} \Delta_{ \xi_3^*} = \tilde {\bs L}_{{\xi}_3} (\Delta_{{\xi}_3}^{-1} \Delta_{ \xi_3^*}) = \tilde {\bs L}_{{\xi}_3} \tilde {\bs D}_{ \xi_3, {\xi}_3^*} \left( \Delta_{ \xi_3^*}^{-1} {\bs D}_{ \xi_3, {\xi}_3^*} \Delta_{ \xi_3^*} \right)^{-1}.
\end{equation}
 
In the next step, we identify $\{\tau_{\xi_1(0),D}^j(\cdot,0), p^j_{\xi_2(0)|\xi_1(0)}(\cdot|\cdot),  p^j_{\xi_3(0)|\xi_2(0)}(\cdot|\cdot)\}_{j=1}^J$ up to a permutation matrix $\Delta_{{\xi}_3^*}$. Evaluating $
p_{\xi_2,\xi_3,\xi_4|\xi_1,D}(\xi_2,\xi_3,\xi_4|\xi_1,0)=  \sum_{j=1}^J \tau_{\xi_1(0),D}^j(\xi_1,0)$ \\ $ p^j_{\xi_2(0)|\xi_1(0)}(\xi_2|\xi_1) p^j_{\xi_3(0)|\xi_2(0)}(\xi_3|\xi_2) p^j_{\xi_4(0)|\xi_3(0)}(\xi_4|\xi_3)$ at  $\xi_4=b_1,b_2,...,b_{J-1}$ and collecting them into a vector together with $
p_{\xi_2,\xi_3|\xi_1,D}({\xi}_2,{\xi}_3|{\xi}_1,0)=\sum_{j=1}^J \tau_{\xi_1(0),D}^j(\xi_1,0)$ \\ $ p^j_{\xi_2(0)|\xi_1(0)}(\xi_2|\xi_1) p^j_{\xi_3(0)|\xi_2(0)}(\xi_3|\xi_2)$  gives
\begin{equation}\label{r}
\bs q_{{\xi}_2,{\xi}_3|\xi_1}  = {\bs L}_{{\xi}_3}  \bs \ell_{{\xi}_2,\xi_3|{\xi}_1} = {\bs L}_{{\xi}_3}^* \Delta_{ \xi_3^*}^{-1}  \bs \ell_{{\xi}_2,\xi_3|{\xi}_1}
\end{equation}
where
\begin{align*}
\bs q_{{\xi}_2,{\xi}_3|\xi_1} &:=
\begin{pmatrix}
p_{\xi_2,\xi_3|\xi_1,D}({\xi}_2,{\xi}_3|{\xi}_1,0)\\
p_{\xi_2,\xi_3,\xi_4|\xi_1,D}(\xi_2,\xi_3,b_1|\xi_1,0)\\ 
\vdots\\
p_{\xi_2,\xi_3,\xi_4|\xi_1,D}(\xi_2,\xi_3,b_{J-1}|\xi_1,0)
\end{pmatrix}\ \text{ and }\\
\bs \ell_{{\xi}_2,{\xi}_3|\xi_1}&  :=\begin{pmatrix}
\ell^{1}({\xi}_2,\xi_3|{\xi}_1) \\
\vdots\\
\ell^{J}({\xi}_2,\xi_3|{\xi}_1)
\end{pmatrix} :=
\begin{pmatrix}
 \tau_{\xi_1(0),D}^1(\xi_1,0) p^1_{\xi_2(0)|\xi_1(0)}(\xi_2|\xi_1) p^1_{\xi_3(0)|\xi_2(0)}(\xi_3|\xi_2)\\ 
\vdots\\
 \tau_{\xi_1(0),D}^J(\xi_1,0) p^J_{\xi_2(0)|\xi_1(0)}(\xi_2|\xi_1) p^J_{\xi_3(0)|\xi_2(0)}(\xi_3|\xi_2)
\end{pmatrix}.
\end{align*}
From (\ref{L-star}) and (\ref{r}), we identify $\bs \ell_{{\xi}_1,{\xi}_2,\xi_3}$  up to $\Delta_{ \xi_3^*}$  for all $(\xi_1,\xi_2,\xi_3)\in |\mathcal{X}|^{k+2}$ as
  \begin{equation}\label{d-3}
\Delta_{ \xi_3^*}^{-1} \bs \ell_{{\xi}_2,\xi_3|{\xi}_1} =\begin{pmatrix}
\ell^{\delta(1)}({\xi}_2,\xi_3|{\xi}_1),
\ldots,
\ell^{\delta(J)}({\xi}_2,\xi_3|{\xi}_1)
\end{pmatrix}^\top
=\left( {\bs L}_{{\xi}_3}^* \right)^{-1}\bs q_{{\xi}_2,{\xi}_3|\xi_1},
  \end{equation}
  where  
 $ \delta: \{1,2,...,J\} \rightarrow \{1,2,...,J\}$ 
  is a permutation implied by $ \Delta_{ \xi_3^*}$.  Then, $\tau_{\xi_1(0),D}^{\delta(j)}(\xi_1,0)$ is identified from $\ell^{\delta(j)}({\xi}_2,\xi_3|{\xi}_1)$ in (\ref{d-3}) as $\tau_{\xi_1(0),D}^{\delta(j)}(\xi_1,0)
=\sum_{(\xi_2,\xi_3)\in\mathcal{X}^2} \ell^{\delta(j)}({\xi}_2,\xi_3|{\xi}_1)
= \sum_{(\xi_2,\xi_3)\in\mathcal{X}^2}
  \tau_{\xi_1(0),D}^{\delta(j)}(\xi_1,0)\, p^{\delta(j)}_{\xi_2(0)|\xi_1(0)}(\xi_2|\xi_1)\, p^{\delta(j)}_{\xi_3(0)|\xi_2(0)}(\xi_3|\xi_2)$.
  Given that $\tau_{\xi_1(0),D}^{\delta(j)}(\xi_1,0)$ is identified,  $p^{\delta(j)}_{\xi_2(0)|\xi_1(0)}(\xi_2|\xi_1)$ and $p^{\delta(j)}_{\xi_3(0)|\xi_2(0)}(\xi_3|\xi_2)$ are also identified as
   $p^{\delta(j)}_{\xi_2(0)|\xi_1(0)}(\xi_2|\xi_1)=\sum_{\xi_3\in\mathcal{X}}\ell^{\delta(j)}({\xi}_2,\xi_3|{\xi}_1)/ \tau_{\xi_1(0),D}^{\delta(j)}(\xi_1,0)$ and
   $p^{\delta(j)}_{\xi_3(0)|\xi_2(0)}(\xi_3|\xi_2)=$\\ $ \ell^{\delta(j)}({\xi}_2,\xi_3|{\xi}_1)/( \tau_{\xi_1(0),D}^{\delta(j)}(\xi_1,0)p^{\delta(j)}_{\xi_2(0)|\xi_1(0)}(\xi_2|\xi_1)).$
  This proves the identification of \\ $\{ \tau_{\xi_1(0),D}^{\delta(j)}(\cdot,0), p^{\delta(j)}_{\xi_2(0)|\xi_1(0)}(\cdot|\cdot), p^{\delta(j)}_{\xi_3(0)|\xi_2(0)}(\cdot|\cdot)\}_{j=1}^J$  up to $\Delta_{ \xi_3^*}$.
  
To identify $\{p^j_{\xi_4(0)|\xi_3(0)}(\cdot|\cdot)\}_{j=1}^J$ up to $\Delta_{ \xi_3^*}$, define
\begin{align*}
\bar{\bs q}_{\check{\xi}_2,{\xi}_3,\xi_4} &:=
\begin{pmatrix}
p_{\xi_2,\xi_3,\xi_4|\xi_1,D}(\check\xi_2,\xi_3,\xi_4|a_1,0)\\
\vdots\\
p_{\xi_2,\xi_3,\xi_4|\xi_1,D}(\check\xi_2,\xi_3,\xi_4|a_J,0)
\end{pmatrix}\ \text{ and }\
\bar{\bs \ell}_{{\xi}_4|\xi_3} :=\begin{pmatrix}
p^1_{\xi_4(0)|\xi_3(0)}(\xi_4|\xi_3) \\
\vdots\\
p^J_{\xi_4(0)|\xi_3(0)}(\xi_4|\xi_3)
\end{pmatrix}.
\end{align*}
For each $(\xi_3,\xi_4)$, we have the linear system: $\bar{\bs q}_{\check\xi_2,{\xi}_3,\xi_4} = \bar {\bs L}_{\check\xi_2} {\bs D}_{{\xi}_3|\check\xi_2} \bar{\bs \ell}_{{\xi}_4|\xi_3} = $ \\ $(\bar {\bs L}_{\check\xi_2}\Delta_{ \xi_3^*}) ( \Delta_{ \xi_3^*}^{-1}{\bs D}_{{\xi}_3|\check\xi_2} \Delta_{ \xi_3^*} ) (\Delta_{ \xi_3^*}^{-1} \bar{\bs \ell}_{{\xi}_4|\xi_3})$. %
The matrices $\bar {\bs L}_{\check\xi_2}\Delta_{ \xi_3^*}$  and $\Delta_{ \xi_3^*}^{-1}{\bs D}_{{\xi}_3|\check\xi_2} \Delta_{ \xi_3^*}$ are now identified from 
$\{ \tau_{\xi_1(0),D}^{\delta(j)}(\cdot,0),\, p^{\delta(j)}_{\xi_2(0)|\xi_1(0)}(\cdot|\cdot),\, p^{\delta(j)}_{\xi_3(0)|\xi_2(0)}(\cdot|\cdot)\}_{j=1}^J$ 
up to $\Delta_{ \xi_3^*}$. Both are invertible by assumption.
We can therefore solve for the permuted component vector as
$\Delta_{ \xi_3^*}^{-1} \bar{\bs \ell}_{{\xi}_4|\xi_3} =
\begin{pmatrix}
p^{\delta(1)}_{\xi_4(0)|\xi_3(0)}(\xi_4|\xi_3), 
\ldots,
p^{\delta(J)}_{\xi_4(0)|\xi_3(0)}(\xi_4|\xi_3)
\end{pmatrix}^\top
= ( \Delta_{ \xi_3^*}^{-1}{\bs D}_{{\xi}_3|\check\xi_2} \Delta_{ \xi_3^*} )^{-1} (\bar {\bs L}_{\check\xi_2}\Delta_{ \xi_3^*})^{-1}$ $\bar{\bs q}_{\check\xi_2,{\xi}_3,\xi_4}$. 
Therefore, $\{p^{\delta(j)}_{\xi_4(0)|\xi_3(0)}(\cdot|\cdot)\}_{j=1}^J$ is identified. %

\revtext{We proceed to prove the identification of $\{\tau_{\xi_1(0),D}^j(\cdot,1), p^j_{\xi_3(1)|\xi_2(1)}(\cdot|\cdot), p^j_{\xi_4(1)|\xi_3(1)}(\cdot |\cdot)\}_{j=1}^J$ up to the permutation $\delta(\cdot)$ implied by $\Delta_{ \xi_3^*}$. First, to identify $\{\tau_{\xi_1(0),D}^{\delta(j)}(\cdot,1)\}_{j=1}^J$, we use the $k$-period block $\bs X_2^{k+1}:=(\bs X_2,\ldots,\bs X_{k+1})$. Because $T_0=k+1$, this block is observed without treatment for both $D=0$ and $D=1$ by Assumption~\crossref{A-anticipation}, and by Assumptions~\crossref{A-transition-Z} and~\crossref{A-Markov}, we have $p^j_{\bs X_2^{k+1}(0)|\bs X_1(0)}(x_2^{k+1}|x_1)=\prod_{t=2}^{k+1}p^j_{\bs X_t(0)|\bs X_{t-1}(0)}(x_t|x_{t-1})$.
Let $\tau^j_{\bs X_1(0),D}(x_1,d):=$ $\pi^j p^j_{\bs X_1(0),D}(x_1,d)$ $/\sum_{i=1}^J\pi^i p^i_{\bs X_1(0),D}(x_1,d)$ denote the type share given the first-period outcome and treatment status. Then, for $d\in\{0,1\}$,
\begin{equation}\label{eq:treated-mixture}
p_{\bs X_2^{k+1}|\bs X_1,D}(x_2^{k+1}|x_1,d)=\sum_{j=1}^J\tau^{j}_{\bs X_1(0),D}(x_1,d)\,p^{j}_{\bs X_2^{k+1}(0)|\bs X_1(0)}(x_2^{k+1}|x_1).
\end{equation}

We identify  $p^j_{\bs X_2^{k+1}(0)|\bs X_1(0)}(x_2^{k+1}|x_1)$ up to $\delta(\cdot)$, from objects recovered in the preceding steps as follows. For $\xi_1=(x_1,x_2^{k})\in\mathcal{X}^{k}$, $\tau_{\xi_1(0),D}^{\delta(j)}(\xi_1,0)\,p_{\xi_1,D}(\xi_1,0)=\pi^{\delta(j)}p^{\delta(j)}_{\xi_1(0),D}(\xi_1,0)$, where $p_{\xi_1,D}(\xi_1,0)$ is the observed PMF of $(\xi_1,D)$, and, by Assumption~\crossref{A-Markov}, $p^{\delta(j)}_{\xi_1(0),D}(\xi_1,0)=p^{\delta(j)}_{\bs X_1(0),D}(x_1,0)$ $\prod_{t=2}^{k}p^{\delta(j)}_{\bs X_t(0)|\bs X_{t-1}(0)}(x_t|x_{t-1})$. Hence 
\begin{equation}\label{eq:path-id}
p^{\delta(j)}_{\bs X_2^{k+1}(0)|\bs X_1(0)}(x_2^{k+1}|x_1)
=\frac{\tau_{\xi_1(0),D}^{\delta(j)}(\xi_1,0)\,p_{\xi_1,D}(\xi_1,0)\,p^{\delta(j)}_{\xi_2(0)|\xi_1(0)}(x_{k+1}|\xi_1)}{\sum_{\tilde x_2^{k}\in\mathcal{X}^{k-1}}\tau_{\xi_1(0),D}^{\delta(j)}((x_1,\tilde x_2^{k}),0)\,p_{\xi_1,D}((x_1,\tilde x_2^{k}),0)},
\end{equation}
where the denominator equals $\pi^{\delta(j)}p^{\delta(j)}_{\bs X_1(0),D}(x_1,0)$ and is positive by Assumption~\ref{A-P-2}(c)(i), and the factors in the numerator are identified in \eqref{d-3} and the step following it. %

For each $x_1\in\mathcal{X}$ with $p_{\bs X_1,D}(x_1,1)>0$, choose $\{c_l\}_{l=1}^J$ as in Assumption~\ref{A-P-2}(c)(ii) and evaluate \eqref{eq:treated-mixture} at $x_2^{k+1}=c_1,\ldots,c_J$ and $d=0,1$. This yields the matrix system
\begin{align}
&\begin{pmatrix}
p_{\bs X_2^{k+1}|\bs X_1,D}(c_1|x_1,0)& \cdots &p_{\bs X_2^{k+1}|\bs X_1,D}(c_J|x_1,0)\\
p_{\bs X_2^{k+1}|\bs X_1,D}(c_1|x_1,1)& \cdots &p_{\bs X_2^{k+1}|\bs X_1,D}(c_J|x_1,1)
\end{pmatrix}=\nonumber\\
&
\begin{pmatrix}
\tau_{\bs X_1(0),D}^{\delta(1)}(x_1,0) & \cdots& \tau_{\bs X_1(0),D}^{\delta(J)}(x_1,0)\\
\tau_{\bs X_1(0),D}^{\delta(1)}(x_1,1) &\cdots& \tau_{\bs X_1(0),D}^{\delta(J)}(x_1,1)
\end{pmatrix}
\begin{pmatrix}
p^{\delta(1)}_{\bs X_2^{k+1}(0)|\bs X_1(0)}(c_1|x_1)&\cdots & p^{\delta(1)}_{\bs X_2^{k+1}(0)|\bs X_1(0)}(c_J|x_1)\\
\vdots & \ddots & \vdots\\
p^{\delta(J)}_{\bs X_2^{k+1}(0)|\bs X_1(0)}(c_1|x_1)&\cdots & p^{\delta(J)}_{\bs X_2^{k+1}(0)|\bs X_1(0)}(c_J|x_1)
\end{pmatrix}. \label{mat-sys}
\end{align}
The $2\times J$ matrix on the left is observed while 
the $J\times J$ matrix on the right has been identified in (\ref{eq:path-id}) and is invertible by Assumption~\ref{A-P-2}(c)(ii); its columns are indexed by $k$-period paths, of which there are $|\mathcal{X}|^{k}\ge J$. Thus, the $2\times J$ matrix of permuted first-period type shares is identified by inverting \eqref{mat-sys}, and its second row is $\{\tau_{\bs X_1(0),D}^{\delta(j)}(x_1,1)\}_{j=1}^J$. Finally, by Bayes' rule and Assumption~\crossref{A-Markov}, for $\xi_1=(x_1,x_2^{k})$ with $p_{\xi_1,D}(\xi_1,1)>0$,
$\tau_{\xi_1(0),D}^{\delta(j)}(\xi_1,1)=\frac{\tau_{\bs X_1(0),D}^{\delta(j)}(x_1,1)\,p^{\delta(j)}_{\bs X_2^{k}(0)|\bs X_1(0)}(x_2^{k}|x_1)}{\sum_{i=1}^J\tau_{\bs X_1(0),D}^{\delta(i)}(x_1,1)\,p^{\delta(i)}_{\bs X_2^{k}(0)|\bs X_1(0)}(x_2^{k}|x_1)}$,
where $p^{\delta(j)}_{\bs X_2^{k}(0)|\bs X_1(0)}(x_2^{k}|x_1):=\sum_{x_{k+1}\in\mathcal{X}}p^{\delta(j)}_{\bs X_2^{k+1}(0)|\bs X_1(0)}(x_2^{k+1}|x_1)$ has been identified and equals one when $k=1$, and the denominator equals $p_{\bs X_2^{k}|\bs X_1,D}(x_2^{k}|x_1,1)>0$. This identifies $\{\tau_{\xi_1(0),D}^{\delta(j)}(\cdot,1)\}_{j=1}^J$.}

Next, for the identification of $\{p^{\delta(j)}_{\xi_3(1)|\xi_2(1)}(\cdot|\cdot), p^{\delta(j)}_{\xi_4(1)|\xi_3(1)}(\cdot |\cdot)\}_{j=1}^J$, define 
$\tau_{\xi_1,\xi_2,D}^{\delta(j)}(\xi_1,\xi_2,1):= \tau_{\xi_1(0),D}^{\delta(j)}(\xi_1,1) p^{\delta(j)}_{\xi_2(0)|\xi_1(0)}(\xi_2|\xi_1)$ and 
$p_{\xi_3,\xi_4|\xi_2} ^{\delta(j)}(\xi_3,\xi_4|\xi_2):=p^{\delta(j)}_{\xi_3(1)|\xi_2(1)}(\xi_3|\xi_2)p^{\delta(j)}_{\xi_4(1)|\xi_3(1)}(\xi_4|\xi_3)$. 
Note that $\tau_{\xi_1,\xi_2,D}^{\delta(j)}$ is identified from the preceding steps. The conditional PMF for $D=1$ can be expressed as:
$
p_{\xi_2,\xi_3,\xi_4|\xi_1,D}(\xi_2,\xi_3,\xi_4|\xi_1,1)= \sum_{j=1}^J \tau_{\xi_1,\xi_2,D}^{\delta(j)}(\xi_1,\xi_2,1) p_{\xi_3,\xi_4|\xi_2} ^{\delta(j)}(\xi_3,\xi_4|\xi_2).
$
For each $\xi_2 \in \mathcal{X}$, choose $\{e_j\}_{j=1}^J$ as in Assumption~\ref{A-P-2}(d). Evaluating the PMF at $\xi_1 = e_1, \ldots, e_J$ yields:
\begin{align*}
\begin{pmatrix}
p_{\xi_2,\xi_3,\xi_4|\xi_1,D}(\xi_2,\xi_3,\xi_4|e_1,1)\\
\vdots\\
p_{\xi_2,\xi_3,\xi_4|\xi_1,D}(\xi_2,\xi_3,\xi_4|e_J,1)
\end{pmatrix}
&=\begin{pmatrix}
\tau_{\xi_1,\xi_2,D}^{\delta(1)}(e_1,\xi_2,1)& \cdots & \tau_{\xi_1,\xi_2,D}^{\delta(J)}(e_1,\xi_2,1)\\
\vdots & \ddots & \vdots\\
\tau_{\xi_1,\xi_2,D}^{\delta(1)}(e_J,\xi_2,1)& \cdots & \tau_{\xi_1,\xi_2,D}^{\delta(J)}(e_J,\xi_2,1)
\end{pmatrix}\!\!
\begin{pmatrix}
p_{\xi_3,\xi_4|\xi_2} ^{\delta(1)}(\xi_3,\xi_4|\xi_2)\\
\vdots\\
p_{\xi_3,\xi_4|\xi_2} ^{\delta(J)}(\xi_3,\xi_4|\xi_2)
\end{pmatrix}.
\end{align*}
The $J \times J$ matrix of $\tau_{\xi_1,\xi_2,D}^{\delta(j)}$ components is identified and is invertible by Assumption~\ref{A-P-2}(d). We can therefore solve for the vector of permuted product probabilities to identify $\begin{pmatrix}
p_{\xi_3,\xi_4|\xi_2} ^{\delta(1)}(\xi_3,\xi_4|\xi_2),
\ldots,
p_{\xi_3,\xi_4|\xi_2} ^{\delta(J)}(\xi_3,\xi_4|\xi_2)
\end{pmatrix}^\top$. 
The individual components\\ $\{p^{\delta(j)}_{ \xi_3(1)|\xi_2(1)}(\cdot |\cdot ),p^{\delta(j)}_{\xi_4(1)|\xi_3(1)}(\cdot |\cdot )\}_{j=1}^J$ are then identified by marginalizing $p_{\xi_3,\xi_4|\xi_2} ^{\delta(j)}$ (e.g., summing over $\xi_4$ to find $p^{\delta(j)}_{\xi_3(1)|\xi_2(1)}$).

Finally, we identify the base probabilities $\{\pi^{\delta(j)},p^{\delta(j)}_{\xi_1(0),D}(\cdot,\cdot)\}_{j=1}^J$ using the identified $\{\tau_{\xi_1(0),D}^{\delta(j)}(\cdot,\cdot)\}_{j=1}^J$ and the observable PMF $p_{\xi_1,D}(\xi_1,d)$: 
$\pi^{\delta(j)} = \sum_{(\xi_1,d)} \tau_{\xi_1(0),D}^{\delta(j)}(\xi_1,d) p_{\xi_1,D}(\xi_1,d)$ and
$p^{\delta(j)}_{\xi_1(0),D}(\xi_1,d) = \frac{\tau^{\delta(j)}_{\xi_1(0),D}(\xi_1,d)p_{\xi_1,D}(\xi_1,d)}{\pi^{\delta(j)}}$.
Under Assumption \crossref{A-sample}(c), re-ordering the components such that $\pi^1<\pi^2<\cdots<\pi^J$ fixes the permutation $\delta(\cdot)$ and uniquely identifies all components, which completes the proof.
 $\qedsymbol$

\clearpage
\begin{center}
\vspace*{1cm}
{\Large\textbf{Online Supplemental Appendix}}\\[6pt]
\end{center}

\setcounter{section}{0}
\renewcommand{\thesection}{S.\arabic{section}}
\def\theHsection{supp.\arabic{section}}
\def\theHsubsection{supp.\arabic{section}.\arabic{subsection}}
\setcounter{equation}{0}
\renewcommand{\theequation}{S.\arabic{equation}}
\def\theHequation{supp.\arabic{equation}}
\setcounter{figure}{0}
\renewcommand{\thefigure}{S.\arabic{figure}}
\def\theHfigure{supp.\arabic{figure}}
\setcounter{table}{0}
\renewcommand{\thetable}{S.\arabic{table}}
\def\theHtable{supp.\arabic{table}}
\setcounter{theorem}{0}
\renewcommand{\thetheorem}{S.\arabic{theorem}}
\def\theHtheorem{supp.thm.\arabic{theorem}}
\setcounter{lemma}{0}
\renewcommand{\thelemma}{S.\arabic{lemma}}
\def\theHlemma{supp.lem.\arabic{lemma}}
\setcounter{proposition}{0}
\renewcommand{\theproposition}{S.\arabic{proposition}}
\def\theHproposition{supp.prop.\arabic{proposition}}
\setcounter{corollary}{0}
\renewcommand{\thecorollary}{S.\arabic{corollary}}
\def\theHcorollary{supp.cor.\arabic{corollary}}
\setcounter{definition}{0}
\renewcommand{\thedefinition}{S.\arabic{definition}}
\def\theHdefinition{supp.def.\arabic{definition}}
\setcounter{example}{0}
\renewcommand{\theexample}{S.\arabic{example}}
\def\theHexample{supp.exa.\arabic{example}}
\setcounter{remark}{0}
\renewcommand{\theremark}{S.\arabic{remark}}
\def\theHremark{supp.rem.\arabic{remark}}
\setcounter{assumption}{0}
\renewcommand{\theassumption}{S.\arabic{assumption}}
\def\theHassumption{supp.assu.\arabic{assumption}}
\setcounter{footnote}{0}%
\def\theHfootnote{supp.\arabic{footnote}}
\section{Proofs (Continued)}\label{sec:proofs-continued}

  \subsection{Proof of Proposition \crossref{prop:att-identification-vanilla}}
The result follows from the identity $\E\left[\mathbf{X}_t(1) \mid D = 1\right]=\E\left[\mathbf{X}_t \mid D = 1\right]$ and 
\begin{align*}
\E\left[\mathbf{X}_t(0) \mid D = 1\right]&=\E\left[\E\left[\mathbf{X}_t(0)\mid \bs X_1^{T_0}(0), D = 1\right]  \mid D = 1\right]\\
&=\E\left[\E\left[\mathbf{X}_t(0)\mid \bs X_1^{T_0}(0), D = 0\right]  \mid D = 1\right]\\
&=\E\left[\E\left[\mathbf{X}_t\mid\bs X_1^{T_0}, D = 0\right]  \mid D = 1\right],
\end{align*}
where the first equality follows from the law of iterated expectations, where the outer expectation over $\bs X_1^{T_0}(0)$ is taken with respect to $\Pr(\bs X_1^{T_0}(0) \mid D=1)$;
the second from Assumption \crossref{A-transition} together with the no-anticipation assumption (Assumption~\crossref{A-anticipation}) and the common support assumption (Assumption~\crossref{A-overlap}); and the third from the fact that $\bs X_t(0)=\bs X_t$ for all $t$ when $D=0$, which also uses Assumption~\crossref{A-anticipation} to replace $\bs X_1^{T_0}(0)$ with the observed $\bs X_1^{T_0}$ in the conditioning set. Substituting these expressions into the definition of $\bs\mu_t^{\text{ATT}}$ in \crosseqref{eq:ATT} yields equation \crosseqref{eq:att-identification-vanilla}. $\qedsymbol$

\subsection{Proof of Proposition \ref{proposition:wooldridge-index}}
Let $\mu^{\text{DiD},(k)}_t$ and $\mu^{\text{ATT},(k)}_t$ denote the $k$th elements of $\bs\mu^{\text{DiD}}_t$ and $\bs\mu^{\text{ATT}}_t$. Using the identity $\E[X^{(k)}_t(1)\mid D=1]=\E[X^{(k)}_t\mid D=1]$, we have $\mu^{\text{ATT},(k)}_t=\E[X^{(k)}_t\mid D=1]-\E[X^{(k)}_t(0)\mid D=1]$, so subtracting $\mu^{\text{ATT},(k)}_t$ from the estimand (\ref{eq:index-did-estimand}) cancels $\E[X^{(k)}_t\mid D=1]$ and expresses the bias as
\begin{equation}\label{eq:index-did-bias-step1}
\E\!\left[X^{(k)}_t(0)\mid D=1\right]-G\!\left(G^{-1}\!\left(\E\!\left[X^{(k)}_{T_0}\mid D=1\right]\right)+\Delta^{(k)}_t\right).
\end{equation}
By Assumption~\ref{A-anticipation}, $\bs X_{T_0}=\bs X_{T_0}(0)$ for all units and $\bs X_t=\bs X_t(0)$ for control units, so the $k$th element of (\ref{eq:diff}) reads 
$\mu^{\text{DiD},(k)}_t-\mu^{\text{ATT},(k)}_t=\E\!\left[X^{(k)}_t(0)\mid D=1\right]-\E\!\left[X^{(k)}_{T_0}\mid D=1\right]-\E\!\left[X^{(k)}_t\mid D=0\right]+\E\!\left[X^{(k)}_{T_0}\mid D=0\right]$.
Proposition~\ref{proposition:dbt-vs-did} equates the same difference with the $k$th element of (\ref{eq:did-bias-decomposition}); combining the two expressions and rearranging,
\begin{align}\label{eq:index-did-bias-step2}
&\E\!\left[X^{(k)}_t(0)\mid D=1\right]
=\sum_{\bs x_1^{T_0}\in\mathcal{X}^{T_0}}\E\!\left[X^{(k)}_t-x^{(k)}_{T_0}\mid \bs X_1^{T_0}=\bs x_1^{T_0},D=0\right]\nonumber\\
&\quad\quad\qquad\times\left\{\Pr\left(\bs X_1^{T_0}=\bs x_1^{T_0}\mid D=1\right)-\Pr\left(\bs X_1^{T_0}=\bs x_1^{T_0}\mid D=0\right)\right\}\nonumber\\
&\quad\quad\qquad+\E\!\left[X^{(k)}_{T_0}\mid D=1\right]+\E\!\left[X^{(k)}_{t}\mid D=0\right]-\E\!\left[X^{(k)}_{T_0}\mid D=0\right].
\end{align}
Next, $G\!\left(G^{-1}\!\left(\E[X^{(k)}_{T_0}\mid D=0]\right)+\Delta^{(k)}_t\right)=\E[X^{(k)}_t\mid D=0]$ by the definition of $\Delta^{(k)}_t$, so the definition of $\lambda^{(k)}_t$ in (\ref{eq:index-did-weight}) gives%
\begin{align}\label{eq:index-did-bias-step3}
&G\!\left(G^{-1}\!\left(\E\!\left[X^{(k)}_{T_0}\mid D=1\right]\right)+\Delta^{(k)}_t\right)\nonumber\\
&=\E\!\left[X^{(k)}_t\mid D=0\right]\ +\lambda^{(k)}_t\left\{\E\!\left[X^{(k)}_{T_0}\mid D=1\right]-\E\!\left[X^{(k)}_{T_0}\mid D=0\right]\right\}.
\end{align}
Substituting (\ref{eq:index-did-bias-step2}) and (\ref{eq:index-did-bias-step3}) into (\ref{eq:index-did-bias-step1}) and collecting terms, the bias equals%
\begin{align*}
&\sum_{\bs x_1^{T_0}\in\mathcal{X}^{T_0}}\E\!\left[X^{(k)}_t-x^{(k)}_{T_0}\mid \bs X_1^{T_0}=\bs x_1^{T_0},D=0\right]\\
&\quad\qquad\times\left\{\Pr\left(\bs X_1^{T_0}=\bs x_1^{T_0}\mid D=1\right)-\Pr\left(\bs X_1^{T_0}=\bs x_1^{T_0}\mid D=0\right)\right\}\\
&\quad\qquad+\left(1-\lambda^{(k)}_t\right)\left\{\E\!\left[X^{(k)}_{T_0}\mid D=1\right]-\E\!\left[X^{(k)}_{T_0}\mid D=0\right]\right\}.
\end{align*}
Since $\E[X^{(k)}_{T_0}\mid D=d]=\sum_{\bs x_1^{T_0}\in\mathcal{X}^{T_0}}x^{(k)}_{T_0}\Pr(\bs X_1^{T_0}=\bs x_1^{T_0}\mid D=d)$ for $d=0,1$, the last term merges into the sum and replaces the netted-out $x^{(k)}_{T_0}$ with $\lambda^{(k)}_t x^{(k)}_{T_0}$, which yields (\ref{eq:index-did-bias}). $\qedsymbol$

\subsection{Proof of Proposition \crossref{P-1}}

The expectation of untreated potential outcomes for treated units can be expressed as
\begin{align*}
  &  \mathbb{E} \left[ \bs X_t (0) \mid D = 1, Z_i = j \right] 
    =\mathbb{E} \left[ \mathbb{E} \left[\bs X_t (0) \mid  \bs X_{1}^{T_0}(0), D = 1,Z = j  \right] \mid D = 1, Z = j \right]   \\
    &=\mathbb{E} \left[ \mathbb{E} \left[\bs X_t (0) \mid  \bs X_{1}^{T_0}(0), D = 0,Z = j  \right] \mid D = 1, Z = j \right] \\
    &=\mathbb{E} \left[ \mathbb{E} \left[\bs X_t (0) \mid \bs X_{T_0}(0), D = 0,Z = j  \right] \mid D = 1, Z = j \right] \\
    & = \sum_{\bs x_{T_0}\in \mathcal{X}} \Pr(\bs X_{T_0}=\bs x_{T_0}\mid D = 1, Z = j )  \mathbb{E} \left[\bs X_t  \mid \bs X_{T_0}=\bs x_{T_0}, D = 0,Z = j  \right],
\end{align*}
where the first equality follows from the law of iterated expectations, the second from transition independence, the third from the Markov property of $\bs X_t(0)$, and the fourth from the definition of conditional expectations. 

Similarly,
$
\mathbb{E} \left[ \bs X_t (1) \mid D = 1, Z_i = j \right]
= \sum_{\bs x_{T_0}}
\Pr(\bs X_{T_0} = \bs x_{T_0} \mid D = 1, Z = j)$\\ $\times
\mathbb{E} \left[\bs X_t \mid \bs X_{T_0} = \bs x_{T_0}, D = 1, Z = j \right].
$
Therefore,   (\crossref{eq:LTATT-identification}) follows from subtracting the two objects given
$
\bs\mu_t^{ATT,j}
= \mathbb{E} \left[ \bs X_t (1) \mid D = 1, Z_i = j \right]
- \mathbb{E} \left[ \bs X_t (0) \mid D = 1, Z_i = j \right].
$

Moreover, since Proposition~\crossref{P-2} implies that the $j$th type-specific probabilities and expectations appearing in (\crossref{eq:LTATT-identification}) are identified from  $\{p_{\bs W}(\bs w;\bs\psi): \bs w\in\bs{\mathcal W}\}$, it follows directly from (\crossref{eq:LTATT-identification}) that $\bs\mu_t^{ATT,j}$ is identified.
$\qedsymbol$

\subsection{Proof of Proposition \crossref{P-2-k}}
\begin{revblock}
Suppose that $T_0=(k+1)\ell$, $T=2(k+1)\ell$, and $J=|\mathcal{X}|^{k\ell}$; as in the proof of Proposition~\crossref{P-2}, the proof extends to smaller $J$ and to larger $T_0$ and $T$. Let $\bs X_{(m)}:=\bs X_{(m-1)\ell+1}^{m\ell}\in\mathcal{X}^{\ell}$, $m=1,\ldots,2k+2$, denote consecutive blocks of $\ell$ periods, and partition $(\bs X_1,\ldots,\bs X_T)=(\xi_1,\xi_2,\xi_3,\xi_4)$ as
\begin{align}
&\xi_1:=(\bs X_{(1)},\ldots,\bs X_{(k)})=\bs X_1^{k\ell},\quad \xi_2:=\bs X_{(k+1)},\quad \xi_3:=\bs X_{(k+2)},\nonumber\\
&\text{and}\quad \xi_4:=(\bs X_{(k+3)},\ldots,\bs X_{(2k+2)})=\bs X_{(k+2)\ell+1}^{T},
\label{eq:xi-2}
\end{align}
so that $\xi_1,\xi_4\in\mathcal{X}^{k\ell}$ and $\xi_2,\xi_3\in\mathcal{X}^{\ell}$; when $\ell=1$, (\ref{eq:xi-2}) is \crosseqref{eq:xi}, and when $k=1$ it consists of four blocks of $\ell$ periods. For $d\in\{0,1\}$, define $\bs X_{(m)}(d)$ and the potential outcomes $\xi_1(d),\ldots,\xi_4(d)$ by replacing $\bs X_t$ with $\bs X_t(d)$. Under Assumption~\crossref{A-Markov-2}, conditional on $(D,Z)$ the block sequence $(\bs X_{(1)},\ldots,\bs X_{(2k+2)})$ is a first-order Markov chain on $\mathcal{X}^{\ell}$. Its initial joint mass is $p^j_{\{\bs X_s(0)\}_{s=1}^{\ell},D}(\bs x_1^{\ell},d)$, the initial distribution of the component $\bs\varphi^j$, and, writing $(\bs x_{(m)},\bs x_{(m+1)})=(\bs x_{(m-1)\ell+1}^{m\ell},\bs x_{m\ell+1}^{(m+1)\ell})$, its block transition probabilities are
\begin{equation}\label{eq:block-kernel}
p^j_{\bs X_{(m+1)}(d)|\bs X_{(m)}(d)}(\bs x_{(m+1)}|\bs x_{(m)}):=\prod_{t=m\ell+1}^{(m+1)\ell}p^j_{\bs X_t(d)|\{\bs X_s(d)\}_{s=t-\ell}^{t-1}}\bigl(\bs x_t\,\big|\,\bs x_{t-\ell}^{t-1}\bigr),
\end{equation}
for $d=0$ when $m\le k$ and $d\in\{0,1\}$ when $m\ge k+1$: the product of the $\ell$th-order transition probabilities along the block, each conditioning on the $\ell$ outcomes that precede it, which lie in $(\bs x_{(m)},\bs x_{(m+1)})$. Because $T_0=(k+1)\ell$, the blocks $\bs X_{(1)},\ldots,\bs X_{(k+1)}$ are pre-treatment, and by Assumptions~\crossref{A-anticipation}, \crossref{A-transition-Z}, and~\crossref{A-Markov-2} the transitions among them are governed by the untreated block transition probabilities in both arms. The PMF of $\bs W$ therefore has the form of the mixture representation used in the proof of Proposition~\crossref{P-2}, with $(\bs X_{(1)},\ldots,\bs X_{(2k+2)})$ in place of $(\bs X_1,\ldots,\bs X_{2k+2})$ and \eqref{eq:block-kernel} in place of the first-order transition probabilities; the transition probabilities between the $\xi$'s of (\ref{eq:xi-2}) are the corresponding products, for instance $p^j_{\xi_2(0)|\xi_1(0)}(\xi_2|\xi_1)=p^j_{\bs X_{(k+1)}(0)|\bs X_{(k)}(0)}(\xi_2|\bs x_{(k)})$ and $p^j_{\xi_4(d)|\xi_3(d)}(\xi_4|\xi_3)=\prod_{m=k+2}^{2k+1}p^j_{\bs X_{(m+1)}(d)|\bs X_{(m)}(d)}(\bs x_{(m+1)}|\bs x_{(m)})$.

\begin{assumptionIDk}\label{A-P-2-k}
Parts (a), (b), and (d) of Assumption~\crossref{A-P-2} hold with the $\xi$'s defined in (\ref{eq:xi-2}) in place of those in \crosseqref{eq:xi} and with the transition probabilities \eqref{eq:block-kernel} in place of the first-order ones, where $\xi_2$, $\xi_3$, $\check\xi_2$, $\bar\xi_2$, and $\bar\xi_3$ range over $\mathcal{X}^{\ell}$ and $\xi_1$, $\xi_4$, $a_l$, $b_l$, and $e_l$ over $\mathcal{X}^{k\ell}$ (in part (b), $(\xi_2,\xi_3,\xi_4)\in\mathcal{X}^{(k+2)\ell}$); and, in place of part (c): for every $\bs x_1^{\ell}\in\mathcal{X}^{\ell}$, (i) if $k\ge 2$, $p^j_{\{\bs X_s(0)\}_{s=1}^{\ell},D}(\bs x_1^{\ell},0)>0$ for $j=1,\ldots,J$; and (ii) there exist $\{c_l\}_{l=1}^J$ with $c_l=(c_{l,\ell+1},\ldots,c_{l,(k+1)\ell})\in\mathcal{X}^{k\ell}$ such that, writing $(c_{l,1},\ldots,c_{l,\ell}):=\bs x_1^{\ell}$, the $J\times J$ matrix
$$
\begin{pmatrix}
 p^1_{\bs X_{\ell+1}^{(k+1)\ell}(0)|\bs X_1^{\ell}(0)}(c_1|\bs x_1^{\ell})& \cdots& p^J_{\bs X_{\ell+1}^{(k+1)\ell}(0)|\bs X_1^{\ell}(0)}(c_1|\bs x_1^{\ell})\\
 \vdots&\ddots &\vdots\\
 p^1_{\bs X_{\ell+1}^{(k+1)\ell}(0)|\bs X_1^{\ell}(0)}(c_J|\bs x_1^{\ell})&\cdots&p^J_{\bs X_{\ell+1}^{(k+1)\ell}(0)|\bs X_1^{\ell}(0)}(c_J|\bs x_1^{\ell})
\end{pmatrix}
$$
is non-singular, where $p^j_{\bs X_{\ell+1}^{(k+1)\ell}(0)|\bs X_1^{\ell}(0)}(c_l|\bs x_1^{\ell}):=\prod_{t=\ell+1}^{(k+1)\ell}p^j_{\bs X_t(0)|\{\bs X_s(0)\}_{s=t-\ell}^{t-1}}(c_{l,t}|c_{l,t-\ell},\ldots,c_{l,t-1})$ is the type-$j$ probability of the $k\ell$-period path $c_l$ of $\bs X_{\ell+1}^{(k+1)\ell}=(\bs X_{(2)},\ldots,\bs X_{(k+1)})$ given $\bs X_1^{\ell}(0)=\bs x_1^{\ell}$, the product of the block transition probabilities \eqref{eq:block-kernel} along the path. The rows are drawn from the $|\mathcal{X}|^{k\ell}\ge J$ values of that path. When $k=1$, part (c)(i) is vacuous and (ii) requires, for every $\xi_1\in\mathcal{X}^{\ell}$, a non-singular matrix with elements $p^j_{\xi_2(0)|\xi_1(0)}(c_l|\xi_1)$, $c_l\in\mathcal{X}^{\ell}$; when $\ell=1$, this assumption is Assumption~\crossref{A-P-2}.
\end{assumptionIDk}

Under Assumption~\ref{A-P-2-k}, the proof of Proposition~\crossref{P-2} applies with $\mathcal{X}^{\ell}$ as the outcome space: replace $(\bs X_1,\ldots,\bs X_{2k+2})$ by $(\bs X_{(1)},\ldots,\bs X_{(2k+2)})$, the $\xi$'s of \crosseqref{eq:xi} by those of (\ref{eq:xi-2}) together with the corresponding potential outcomes, the first-order transition probabilities by the block transition probabilities \eqref{eq:block-kernel}, and Assumption~\crossref{A-Markov} by Assumption~\crossref{A-Markov-2}. In particular, in the step that identifies $\{\tau^{j}_{\xi_1(0),D}(\cdot,1)\}_{j=1}^J$, the $k$-period block that follows the first period becomes the $k$-block path $\bs X_{\ell+1}^{(k+1)\ell}$ that follows $\bs X_{(1)}=\bs X_1^{\ell}$. For $k\ge 2$ we recover its type-specific probabilities from the control-arm objects identified in the preceding steps by the ratio of that step, whose denominator $\pi^{\delta(j)}p^{\delta(j)}_{\{\bs X_s(0)\}_{s=1}^{\ell},D}(\bs x_1^{\ell},0)$ is positive by Assumption~\ref{A-P-2-k}(c)(i); for $k=1$ the path is $\xi_2$ itself and $p^{\delta(j)}_{\xi_2(0)|\xi_1(0)}$ is identified directly. We then invert the $2\times J$ system for each $\bs x_1^{\ell}\in\mathcal{X}^{\ell}$ at the values $c_1,\ldots,c_J$ of Assumption~\ref{A-P-2-k}(c)(ii) and, for $k\ge 2$, convert the type shares given $\bs X_1^{\ell}$ into those given $\xi_1$ by Bayes' rule. This establishes identification of $\big\{\pi^{j}, p^{j}_{\xi_1(0),D}(\cdot,\cdot), p^{j}_{\xi_2(0)\mid\xi_1(0)}(\cdot\mid\cdot), p^{j}_{\xi_3(0)\mid\xi_2(0)}(\cdot\mid\cdot), p^{j}_{\xi_4(0)\mid\xi_3(0)}(\cdot\mid\cdot), p^{j}_{\xi_3(1)\mid\xi_2(1)}(\cdot\mid\cdot), p^{j}_{\xi_4(1)\mid\xi_3(1)}(\cdot\mid\cdot)\big\}_{j=1}^J$. The components of $\bs\psi$ follow by the chain rule, on partial paths of positive type-$j$ probability: $p^{j}_{\xi_2(0)|\xi_1(0)}$ and $p^{j}_{\xi_3(d)|\xi_2(d)}$ are the block transition probabilities \eqref{eq:block-kernel} for $m=k$ and $m=k+1$; the initial distribution and those for $m\le k-1$ follow from $p^{j}_{\xi_1(0),D}$, and those for $m\ge k+2$ from $p^{j}_{\xi_4(d)|\xi_3(d)}$, in each case by summing out the later blocks and dividing; and then, for $t=m\ell+1,\ldots,(m+1)\ell$, the $\ell$th-order transition probability $p^j_{\bs X_t(d)|\{\bs X_s(d)\}_{s=t-\ell}^{t-1}}(\bs x_t|\bs x_{t-\ell}^{t-1})$ as the ratio of the sums of $p^j_{\bs X_{(m+1)}(d)|\bs X_{(m)}(d)}(\cdot|\bs x_{(m)})$ over the coordinates after $t$ and after $t-1$. This proves part~(a).
\end{revblock}

Part~(b) follows similarly by repeating the proof of 
Proposition~\crossref{P-1} under Assumption~\crossref{A-Markov-2}
in place of Assumption~\crossref{A-Markov}. \qedsymbol

\subsection{Proof of Proposition \crossref{P-estimation}}

The consistency of the MLE $\hat{\bs\psi}$ follows directly from Theorem~2.5 of \citet{Newey1994}.
Specifically, condition~(i) of that theorem is satisfied by Proposition~\crossref{P-2}, while conditions~(ii)–(iv) are implied by Assumption~\crossref{A-MLE}(a).
Asymptotic normality then follows from Theorem~3.3 of \citet{Newey1994}, whose conditions~(i)–(v) are readily verified under Assumption~\crossref{A-MLE}.
Hence, $\hat{\bs\psi} \xrightarrow{p} \bs\psi^*$ and
\begin{align}
&\sqrt n(\hat{\bs\psi}-\bs\psi^*)=\bs I(\bs\psi^*)^{-1}\,\frac{1}{\sqrt n}\sum_{i=1}^n \bs S(\bs W_i;\bs\psi^*)+o_p(1)\ \Rightarrow \ N(0,  \bs I(\bs\psi^*)^{-1}),\label{eq:asy-linear}%
\end{align}
where $\bs S(\bs W_i;\bs\psi^*)$ is the individual score function evaluated at $\bs\psi^*$ for the MLE in (\crossref{mle}).

By continuity of the posterior mapping, the consistency of $\hat{\bs\psi}$ implies that
\begin{equation}\label{eq:uniform-tau}
\sup_i \big| \hat\tau_i^j - \tau^j(\bs W_i; \bs\psi^*) \big| \xrightarrow{p} 0,
\qquad
\tau^{j}(\bs W; \bs\psi^*) := \Pr(Z=j \mid \bs W; \bs\psi^*).
\end{equation}
Furthermore,  Assumption~\crossref{A-MLE}(a) ensures that
$\E[\tau^j(\bs W; \bs\psi)]$ and
$\E[\tau^j(\bs W; \bs\psi)
  \mathbf{1}\{\bs X_{T_0} = \bs x_{T_0}, D=d\}]$
are strictly positive in a neighborhood of~$\bs\psi^*$.

Then, in (\crossref{eq:P-X0}), 
$\widehat{\Pr}(\bs X_{T_0} = \bs x_{T_0} \mid D=d, Z=j)
= \frac{\frac{1}{n}\sum_{i=1}^n
  \mathbf{1}\{\bs X_{iT_0} = \bs x_{T_0}, D_i = d\}\,
  \tau^j(\bs W_i; \bs\psi^*) + o_p(1)}
{\frac{1}{n}\sum_{i=1}^n\mathbf{1}\{D_i=d\} \tau^j(\bs W_i; \bs\psi^*) + o_p(1)} $
$\xrightarrow{p}
\frac{\E\!\big[\mathbf{1}\{\bs X_{T_0} = \bs x_{T_0}, D = d\}\, \tau^j(\bs W; \bs\psi^*)\big]}
{\E\!\big[\mathbf{1}\{D=d\}\tau^j(\bs W; \bs\psi^*)\big]} 
= \Pr(\bs X_{T_0} = \bs x_{T_0} \mid D=d, Z=j),$  
where the first equality follows from~\eqref{eq:uniform-tau}, the second from the Law of Large Numbers and the Continuous Mapping Theorem, and the last from
$\E[\tau^j(\bs W; \bs\psi^*) g(\bs W)]
  = \E[\mathbf{1}\{Z=j\} g(\bs W)]$
for any measurable function~$g$.  

Identical arguments apply to
$\widehat{\E}\!\left[\bs X_t \mid
  \bs X_{T_0} = \bs x_{T_0}, D=d, Z=j\right]$ and $\widehat{\Pr}(Z=j\mid D=1)$ 
in~(\crossref{eq:E-Xt}) and (\crossref{eq:ATT-estimator}).
Hence, $\widehat{\E}\!\left[\bs X_t \mid
  \bs X_{T_0} = \bs x_{T_0}, D=d, Z=j\right]
 \xrightarrow{p}
\E\!\left[\bs X_t \mid
  \bs X_{T_0} = \bs x_{T_0}, D=d, Z=j\right]$ and
 $ \widehat{\Pr}(Z=j\mid D=1)  \xrightarrow{p}  {\Pr}(Z=j\mid D=1).$

Both $\hat{\bs\mu}_{t}^{ATT,j}$ and $\hat{\bs\mu}_{t}^{ATT}$  in (\crossref{eq:LTATT-estimator}) and (\crossref{eq:ATT-estimator})
are continuously differentiable functions of these probabilities and expectations.
Therefore, by the Continuous Mapping Theorem, 
$\hat{\bs\mu}_{t}^{ATT,j} \xrightarrow{p} \bs\mu_{t}^{ATT,j}$ and  $
\hat{\bs\mu}_{t}^{ATT} \xrightarrow{p} \bs\mu_{t}^{ATT},$ 
which establishes part~(a).

Let $\hat{\bs m}_n(\hat{\bs\psi})$ denote the vector collecting the estimators for the probabilities and expectations defined in 
(\crossref{eq:P-X0})-(\crossref{eq:ATT-estimator}).
Specifically, for each $j=1,2,\ldots,J$ and $(\bs x_{T_0}, d)\in \mathcal{X}\times\{0,1\}$,
$
\hat{m}_{1,n}^{j,\bs x_{T_0}}(\hat{\bs\psi})
 =
\frac{
  \frac{1}{n}\sum_{i=1}^n 
  \mathbf{1}\{\bs X_{iT_0}=\bs x_{T_0}, D_i=1\}\,
  \tau^j(\bs W_i;\hat{\bs\psi})
}{
  \frac{1}{n}\sum_{i=1}^n \mathbf{1}\{D_i=1\} \tau^j(\bs W_i;\hat{\bs\psi})
}$,\\ 
$\hat{m}_{2,n}^{j,\bs x_{T_0},d}(\hat{\bs\psi})
 =
\frac{
  \frac{1}{n}\sum_{i=1}^n 
  \bs X_{it}\,
  \mathbf{1}\{\bs X_{iT_0}=\bs x_{T_0}, D_i=d\}\,
  \tau^j(\bs W_i;\hat{\bs\psi})
}{
  \frac{1}{n}\sum_{i=1}^n 
  \mathbf{1}\{\bs X_{iT_0}=\bs x_{T_0}, D_i=d\}\,
  \tau^j(\bs W_i;\hat{\bs\psi})
}$, and
$\hat{m}_{3,n}^j(\hat{\bs\psi})= \frac{ \frac{1}{n}\sum_{i=1}^n \mathbf{1}\{D_i=1\}\hat\tau_i^j}{ \frac{1}{n}\sum_{i=1}^n \mathbf{1}\{D_i=1\}}$, 
and $\hat{\bs m}_n(\hat{\bs\psi})$ stacks all such elements into a vector.
Let $\bs m(\bs\psi)$ denote the population counterpart, with elements
$m_{1}^{j,\bs x_{T_0}}(\bs\psi)
=
\frac{\E\left[\mathbf{1}\{\bs X_{T_0}=\bs x_{T_0},\, D=1\}\,
  \tau^j(\bs W;\bs\psi)\right]}
{\E\left[\mathbf{1}\{D=1\}\,\tau^j(\bs W;\bs\psi)\right]}$
and analogously for $m_{2}^{j,\bs x_{T_0},d,t}$ and $m_{3}^{j}$.
Because $\tau^j(\bs W;\bs\psi^*) = \Pr(Z=j \mid \bs W)$ under correct
specification and the remaining arguments of each ratio are functions
of $\bs W$, the law of iterated expectations implies that
$\bs m(\bs\psi^*)$ equals the vector of population probabilities and
expectations in (\crossref{eq:P-X0})--(\crossref{eq:ATT-estimator}).
   Note that
\begin{align}
\sqrt{n}\big(\hat{\bs m}(\hat{\bs\psi})-\bs m(\bs\psi^*)\big)
=  \underbrace{\sqrt n\!\left(\hat{\bs m}_{n}(\bs\psi^*)-\bs m(\bs\psi^*)\right)}_{(I)}
+\underbrace{\sqrt n\!\left(\hat{\bs m}_{n}(\hat{\bs\psi})-\hat{\bs m}_{n}(\bs\psi^*)\right)}_{(II)}.\label{eq:m}
\end{align}

For term (II), apply a mean-value (componentwise) expansion of the map 
$\bs\psi \mapsto \hat{\bs m}_{n}(\bs\psi)$: there exists
$\tilde{\bs\psi}$ on the line segment between $\hat{\bs\psi}$ and $\bs\psi^*$ such that 
$(II)
= \Big(\nabla_{\bs\psi}\hat{\bs m}_{n}(\tilde{\bs\psi})\Big)^\top \sqrt n\,(\hat{\bs\psi}-\bs\psi^*)$, 
where $\nabla_{\bs\psi}\hat{\bs m}_{n}(\tilde{\bs\psi})=\partial \hat{\bs m}_{n}(\bs\psi)/\partial \bs\psi|_{\bs\psi=\tilde{\bs\psi}}$.
Since $\hat{\bs\psi}\xrightarrow{p}\bs\psi^*$, we have $\tilde{\bs\psi}\xrightarrow{p}\bs\psi^*$. 
Under the Uniform Law of Large Numbers for the class 
$\{\nabla_{\bs\psi}\hat{\bs m}_{n}(\bs\psi): \bs\psi \in \mathcal N(\bs\psi^*)\}$, which holds because $\mathcal{W}$ is finite, combined with continuity of $\tau^j(\cdot;\bs\psi)$ and its derivatives in a neighborhood $\mathcal N(\bs\psi^*)$, it follows that
$\nabla_{\bs\psi}\hat{\bs m}_{n}(\tilde{\bs\psi}) \ \xrightarrow{p}\ \nabla_{\bs\psi}\bs m(\bs\psi^*)$.
Combining this with the asymptotic linear representation of the MLE in  (\ref{eq:asy-linear}), 
we obtain
\begin{align}
(II) =\frac{1}{\sqrt n} \sum_{i=1}^n\bs A^* \bs S(\bs W_i;\bs\psi^*) \;+\; o_p(1),\label{eq:II}
\end{align}
where
$\bs A^*:= \Big(\nabla_{\bs\psi} \bs m(\bs\psi^*)\Big)^\top \bs I(\bs\psi^*)^{-1} $.

For term (I), each component of $\hat{\bs m}_n(\bs\psi^*)$ is a ratio of empirical means, e.g.
$\hat m_{1,n}^{j,\bs x_{T_0}}(\bs\psi^*)
=  {n^{-1}\sum_{i=1}^n  a_{\bs\psi^*}}/{n^{-1}\sum_{i=1}^n  b_{\bs\psi^*}},$ 
 where
$a_{\bs\psi^*}(\bs W):=\tau^j(\bs W;\bs\psi^*)\,\mathbf 1\{\bs X_{T_0}=\bs x_{T_0},D=1\}$ and  
$b_{\bs\psi^*}(\bs W):= \tau^j(\bs W;\bs\psi^*)\mathbf{1}\{D_i=1\}.$ 
Let $A_i:=a_{\bs\psi^*}(\bs W_i)$ and $B_i:=b_{\bs\psi^*}(\bs W_i)$ with 
$\mu_A:=\E[A_i]$ and $\mu_B:=\E[ B_i]>0$.
Write $\hat m_{1,n}^{j,\bs x_{T_0}}(\bs\psi^*)=g(\bar A_n,\bar B_n)$ with $g(a,b)=a/b$ and 
$\bar A_n=n^{-1}\sum_{i=1}^n  A_i$, $\bar B_n=n^{-1}\sum_{i=1}^n  B_i$.
By the multivariate CLT and a first-order delta method, 
$\sqrt n\big(\hat m_{1,n}^{j,\bs x_{T_0}}(\bs\psi^*)-m_1^{j,\bs x_{T_0}}(\bs\psi^*)\big)
=
\nabla g(\mu_A,\mu_B)^\top
\sqrt n
\begin{pmatrix}
\bar A_n-\mu_A\\
\bar B_n-\mu_B
\end{pmatrix}
+o_p(1).$ 
Since $\nabla g(\mu_A,\mu_B)=(1/\mu_B,\,-\mu_A/\mu_B^2)^\top$, this equals
$\frac{1}{\sqrt n}\sum_{i=1}^n
\left[
\frac{A_i-\mu_A}{\mu_B}
-
\frac{\mu_A}{\mu_B^2}(B_i-\mu_B)
\right]
+o_p(1)
=
\frac{1}{\sqrt n}\sum_{i=1}^n$ $ \phi_1^{j,\bs x_{T_0}}(\bs W_i;\bs\psi^*)+o_p(1),$ 
where, using $m_1^{j,\bs x_{T_0}}(\bs\psi^*)=\mu_A/\mu_B$,
\\
$\phi_1^{j,\bs x_{T_0}}(\bs W;\bs\psi^*)
=\frac{\tau^j(\bs W;\bs\psi^*)}{\E[\mathbf{1}\{D=1\} \tau^j(\bs W;\bs\psi^*)]}
\Big[\mathbf 1\{\bs X_{T_0}=\bs x_{T_0},D=1\}-\mathbf{1}\{D=1\}m_1^{j,\bs x_{T_0}}(\bs\psi^*)\Big].$ 
Analogously, $\phi_2^{j,\bs x_{T_0},d}(\bs W;\bs\psi^*)
=
\frac{\tau^j(\bs W;\bs\psi^*)\mathbf 1\{\bs X_{T_0}=\bs x_{T_0},D=d\}}
{\E[\tau^j(\bs W;\bs\psi^*)\mathbf 1\{\bs X_{T_0}=\bs x_{T_0},D=d\}]}
\big[\bs X_t-m_2^{j,\bs x_{T_0},d}(\bs\psi^*)\big]$ and
$\phi_3^j(\bs W;\bs\psi^*)=\frac{\mathbf 1\{D=1\}}{\Pr(D=1)}
\Big(\tau^j(\bs W;\bs\psi^*)-\E[\tau^j(\bs W;\bs\psi^*)\mid D=1]\Big).$
Stacking all such elements into a vector $\bs\phi(\bs W;\bs\psi^*)$ gives
\begin{equation}
(I)=\frac{1}{\sqrt n}\sum_{i=1}^n \bs\phi(\bs W_i;\bs\psi^*)+o_p(1),
\quad
\E[\bs\phi(\bs W;\bs\psi^*)]=\bs 0.\label{eq:I}
\end{equation}

 Therefore, combining (\ref{eq:m})–(\ref{eq:I}), we obtain the asymptotic linear representation:
\[
\sqrt{n}\big(\hat{\bs m}(\hat{\bs\psi})-\bs m(\bs\psi^*)\big)
=\frac{1}{\sqrt{n}}\sum_{i=1}^n\Big(
\bs\phi(\bs W_i;\bs\psi^*)+\bs A^*\,\bs S(\bs W_i;\bs\psi^*)\Big)+o_p(1).
\]

Let $\bs h(\cdot)$ denote the continuously differentiable function that maps the moment vector $\bs m$ to the vector of LTATT and ATT parameters, 
$\boldsymbol\theta^* = \bs h(\bs m(\bs\psi^*))$, 
as defined by (\crossref{eq:LTATT-estimator})–(\crossref{eq:ATT-estimator}).
Denote its Jacobian by 
$\nabla_{\bs m}\bs h(\bs m) := \partial \bs h(\bs m)/\partial \bs m$.
Then, by the multivariate delta method,
\[
\sqrt{n}(\hat{\boldsymbol\theta}-\boldsymbol\theta^*)
=\frac{1}{\sqrt{n}}\sum_{i=1}^n\bs\eta(\bs W_i)+o_p(1)
\;\Rightarrow\;
\mathcal N(0,\bs V),
\]
where the influence function and asymptotic covariance matrix are given by
\[
\bs\eta(\bs W_i)
:=\left(\nabla_{\bs m}\bs h(\bs m(\bs \psi^*))\right)^\top\,
\big[\bs\phi(\bs W_i;\bs\psi^*)+\bs A^*\,\bs S(\bs W_i;\bs\psi^*)\big],
\quad
\bs V=\E[\bs\eta(\bs W_i)\bs\eta(\bs W_i)^\top].
\]
This establishes part (b). \qedsymbol

\subsection{Proof of Proposition~\ref{prop:bic-postsel}}\label{subsec:prop-bic-postsel}
Let \(\hat p_n\) denote the empirical PMF of \(\{\bs W_i\}_{i=1}^n\),
\(\ell_n(\bs\psi^*)\) the sample log-likelihood at the true parameter, and
\(K(q,p) := \sum_{\bs w} q(\bs w)\log\{q(\bs w)/p(\bs w)\}\). Each component
assigns each cell a product of \(T\) entries in \([\epsilon, 1-\epsilon]\),
so cells satisfy \(p_{\bs\psi}(\bs w) \ge \epsilon^{T}\) on every
\(\Theta_J\) and \(p^*(\bs w) \ge \epsilon^{T+1}\).
\emph{Under-selection} (void when \(J_0 = 1\)). For \(J < J_0\), the identity
\(\sum_{\bs w} \hat p_n \log(p_{\bs\psi}/p^*) = -K(p^*, p_{\bs\psi}) +
\sum_{\bs w}(\hat p_n - p^*)\log(p_{\bs\psi}/p^*)\) and the bound
\(|\log(p_{\bs\psi}/p^*)| \le L := (T+1)\log(1/\epsilon)\) give
\(n^{-1}\{\ell_J - \ell_n(\bs\psi^*)\} \le -\delta + L\|\hat p_n - p^*\|_1\),
where \(\delta := \min_{J < J_0}\inf_{\Theta_J} K(p^*, p_{\bs\psi})\) is
positive: each infimum is attained (\(\Theta_J\) compact,
\(\bs\psi \mapsto K(p^*, p_{\bs\psi})\) continuous with cells in
\([\epsilon^T, 1]\)), and a zero would exhibit a structural representation
of \(p^*\) with fewer than \(J_0\) components. Since
\(\{\widehat{J}_{\mathrm{BIC}} = J\}\) requires
\(2(\ell_J - \ell_{J_0}) \ge \{\mathrm{df}(J) - \mathrm{df}(J_0)\}\log n\)
and \(\ell_{J_0} \ge \ell_n(\bs\psi^*)\) (as
\(\bs\psi^* \in \Theta_{J_0}\)), it forces
\(L\|\hat p_n - p^*\|_1 \ge \delta - \mathrm{df}(J_{\max})(\log n)/(2n) \ge
\delta/2\) for all large \(n\), an event of probability at most
\(2^{M} e^{-n\delta^2/(8L^2)}\) by the Bretagnolle--Huber--Carol inequality
\citep[Proposition~A.6.6]{vandervaartwellner96book}.
\emph{Over-selection.} For \(J > J_0\),
\(\{\widehat{J}_{\mathrm{BIC}} = J\}\) requires
\(\Delta \log n \le 2(\ell_J - \ell_{J_0}) \le
2\{\ell_{\mathrm{sat}} - \ell_n(\bs\psi^*)\} = 2nK(\hat p_n, p^*)\), because
\(\mathrm{df}\) is affine in \(J\) with increment \(\Delta\),
\(\ell_J \le \ell_{\mathrm{sat}} := n\sum_{\bs w}\hat p_n(\bs w)\log \hat
p_n(\bs w)\) (the saturated multinomial maximum), and
\(\ell_{J_0} \ge \ell_n(\bs\psi^*)\); the right side is the deviance of the
\emph{simple} null \(p^*\) in the \(M\)-cell multinomial model, under which
no parameter is estimated, so no regularity of the mixture family enters.
Theorem~2.1 of \citet{kallenberg1985aos} at threshold
\(d_n := \Delta(\log n)/(2n)\), whose hypotheses (\(p^*(\bs w) \ge a/M\)
with \(a := M\epsilon^{T+1}\); \(2nd_n = \Delta\log n \ge M\);
\(d_n < 0.15\)) hold for all large \(n\), bounds its tail probability by
\(C(\log n)^{(M-2)/2}\, n^{-\Delta/2}\) for all such \(n\); the union over
the at most \(J_{\max} - J_0\) over-fitting candidates enlarges \(C\) only,
and extending the bound by the trivial bound one below the threshold keeps
it in the class \(n^{-\Delta/2 + o(1)}\), since a finite prefix does not
affect the exponent's limit. The two bounds and
\(\Delta \ge 4\) prove (i) and (ii).
\emph{(iii).} On \(\{\widehat{J}_{\mathrm{BIC}} = J_0\}\), an event of
probability \(1 - o(n^{-1/2})\), the reported fit is the constrained
maximizer \crosseqref{mle} at the true \(J_0\), so
\(\hat{\bs\mu}^{\mathrm{ATT}}(\widehat{J}_{\mathrm{BIC}})\) and any
interval procedure computed from that fit (with common bootstrap weights)
coincide with the estimator and intervals of
Proposition~\crossref{P-estimation}; this gives the total-variation bound
and transfers both the limit law and asymptotic coverage.

\section{Transition Independence Conditional on Covariates}\label{sec:supp-covariates}
If additional covariates are available, we may relax
Assumption~\crossref{A-transition} by conditioning on discrete control
variables $V_{t}\in \mathcal{V}$.
\begin{assumptionTIC}[Transition independence conditional on covariates]\label{A-transition-x} For all $t =T_0+1,..,T$,
$\Pr\!\bigl(\bs X_{t}(0) = \bs x_{t} \bigm| \bs X_1^{T_0} = \bs x_1^{T_0},\allowbreak V_t=v_t,\allowbreak D = 1\bigr)
=\Pr\!\bigl(\bs X_{t}(0) = \bs x_{t} \bigm| \bs X_1^{T_0} = \bs x_1^{T_0},\allowbreak V_t=v_t,\allowbreak D = 0\bigr)$
holds for all $\bs x_t\in\mathcal{X}$, $\bs x_1^{T_0}\in \mathcal{X}^{T_0}$, and $v_t\in\mathcal{V}$.
\end{assumptionTIC}
Here $V_t$ collects observed covariates that are pre-determined or otherwise
unaffected by treatment; if $V_t$ can respond to treatment, the conditioning
involves post-treatment variables and the assumption becomes a cross-world
restriction. With common support extended to conditioning on
$(\bs X_1^{T_0},V_t)$, the ATT is identified by averaging the
covariate-conditional control means over the treated group's joint
distribution of $(\bs X_1^{T_0},V_t)$, in direct analogy with
Proposition~\crossref{prop:att-identification-vanilla}.

\section{Extension to Staggered Treatment Adoption}\label{sec:staggered}
\begin{revblock}
Our identification strategy can be extended to staggered treatment adoption. First, we replace a treatment group $D_{i}$ with a treatment \textit{cohort}: the period in which treatment is given, $G_{i} \in \mathcal{G}$ with $\mathcal{G} \subset \mathcal{T}$, where we set $G_i=0$ for never-treated units and rule out adoption in period 0. Let $\mathcal{G}_{+}:=\mathcal{G}\setminus\{0\}$ denote the set of treated cohorts. Following \citet{CALLAWAY2021} and \citet{Sun2021}, let $Y_{it}(0)$ denote the potential outcome if unit $i$ is never treated and let $Y_{it}(g)$ denote the potential outcome if treatment begins in period $g\in\mathcal{G}_{+}$. For brevity, we focus on binary outcomes with $Y_{it}(g)\in\mathcal{Y}:=\{0,1\}$. Let $H_{i,g-1}(0):=(Y_{i0}(0),\ldots,Y_{i,g-1}(0))$ and $H_{i,g-1}:=(Y_{i0},\ldots,Y_{i,g-1})$ denote the pre-treatment potential and observed histories.

By extending \cite{CALLAWAY2021} to our transition independence setup, one can use never-treated units, not-yet-treated units, or both as controls. The following assumptions state the two restrictions separately.
\begin{samepage}
\begin{assumption}[Transition independence with a never-treated group]\label{A-transition-staggered-never-treated}
For each $g\in\mathcal{G}_{+}$ and $t\geq g$,
  \begin{align}
  &\Pr\bigl(Y_{it}(0)=y_t \mid H_{i,g-1}(0)=h,G_i=g\bigr)\nonumber\\
  &=\Pr\bigl(Y_{it}(0)=y_t \mid H_{i,g-1}(0)=h,G_i=0\bigr)
  \end{align}
for every $y_t\in\mathcal{Y}$ and every $h\in\mathcal{Y}^{g}$ belonging to both the support of $H_{i,g-1}(0)$ given $G_i=g$ and its support given $G_i=0$.
\end{assumption}
\end{samepage}

\begin{samepage}
\begin{assumption}[Transition independence with not-yet-treated groups]\label{A-transition-staggered-not-yet-treated}
For each $g\in\mathcal{G}_{+}$ and $t\geq g$,
  \begin{align}
  &\Pr\bigl(Y_{it}(0)=y_t \mid H_{i,g-1}(0)=h,G_i=g\bigr)\nonumber\\
  &=\Pr\bigl(Y_{it}(0)=y_t \mid H_{i,g-1}(0)=h,G_i>t,G_i\neq0\bigr)
  \end{align}
for every $y_t\in\mathcal{Y}$ and every $h\in\mathcal{Y}^{g}$ belonging to both the support of $H_{i,g-1}(0)$ given $G_i=g$ and its support given $G_i>t$ with $G_i\neq0$.
\end{assumption}
\end{samepage}

The two restrictions permit the researcher to choose controls according to the research design. Let $\mathcal{G}(g,t)$ denote the selected control set for the cohort-period pair $(g,t)$, defined by:
\[
\mathcal{G}(g,t)=
\begin{cases}
\{0\}, & \text{if only never-treated units are used},\\
\{g'\in\mathcal{G}_{+}:g'>t\}, & \text{if only not-yet-treated units are used},\\
\{0\}\cup\{g'\in\mathcal{G}_{+}:g'>t\}, & \text{if both groups are used}.
\end{cases}
\]
Let $\overline{G}$ denote the latest evaluation period supported by the selected controls:
\[
\overline{G}=
\begin{cases}
T, & \text{if never-treated units are included among the controls},\\
\max\mathcal{G}_{+}-1, & \text{if only not-yet-treated units are used}.
\end{cases}
\]
The second case applies even when the sample also contains never-treated units because an analysis using only not-yet-treated controls excludes them.

\begin{samepage}
\begin{assumption}[No anticipatory effects for staggered treatment adoption]\label{A-anticipation-staggered}
For each $g\in\mathcal{G}_{+}$, $Y_{it}(g)=Y_{it}(0)$ for all units with $G_i=g$ and all $t<g$.
\end{assumption}
\end{samepage}

\begin{samepage}
\begin{assumption}[Common support for staggered treatment adoption]\label{A-overlap-staggered}
There exists a constant $\epsilon>0$ such that
\[
\Pr\bigl(G_i\in\mathcal{G}(g,t)\mid H_{i,g-1}=h\bigr)\geq\epsilon
\]
for every $g\in\mathcal{G}_{+}$, every $g\leq t\leq\overline{G}$, and every $h\in\mathcal{Y}^{g}$ satisfying $\Pr(H_{i,g-1}=h\mid G_i=g)>0$.
\end{assumption}
\end{samepage}
Common support requires the selected controls to represent every pre-treatment history observed in cohort $g$, providing a staggered-adoption counterpart of Assumption~\crossref{A-overlap}.

For every $g\in\mathcal{G}_{+}$ and $t\geq g$, define the cohort-specific ATT by
\[
ATT_{g,t}:=\mathbb{E}\bigl[Y_{it}(g)-Y_{it}(0)\mid G_i=g\bigr].
\]

The following proposition establishes identification of the cohort-specific ATT analogous to Proposition~\crossref{prop:att-identification-vanilla} of the main text.

\begin{proposition}\label{prop:att-identification-staggered}
Fix one of the three control-group choices above. Suppose that the corresponding transition independence assumption (Assumption~\ref{A-transition-staggered-never-treated} or~\ref{A-transition-staggered-not-yet-treated}), or both assumptions when both groups are used, holds together with Assumptions~\ref{A-anticipation-staggered} and \ref{A-overlap-staggered}. Then, for every $g\in\mathcal{G}_{+}$ and $g\leq t\leq\overline{G}$, $ATT_{g,t}$ is identified by
\[
ATT_{g,t}=\mathbb{E}\left[Y_{it}-\mathbb{E}\bigl[Y_{it}\mid H_{i,g-1},G_i\in\mathcal{G}(g,t)\bigr]\mid G_i=g\right].
\]
\end{proposition}

\begin{proof}
Fix a history $h$ for units in cohort $g$. 
Because the selected controls remain untreated through period $t$, their observed outcome equals $Y_{it}(0)$. The corresponding transition independence restriction and Assumption~\crossref{A-anticipation-staggered} therefore give
\[
\mathbb{E}\bigl[Y_{it}(0)\mid H_{i,g-1}(0)=h,G_i=g\bigr]
=\mathbb{E}\bigl[Y_{it}\mid H_{i,g-1}=h,G_i\in\mathcal{G}(g,t)\bigr],
\]
whose right-hand side is well-defined by Assumption~\crossref{A-overlap-staggered}. Iterated expectations identify $\mathbb{E}[Y_{it}(0)\mid G_i=g]$ by averaging the right-hand side over the history distribution in cohort $g$. Since $Y_{it}=Y_{it}(g)$ for cohort $g$ in periods $t\geq g$, subtracting $\mathbb{E}[Y_{it}(0)\mid G_i=g]$ from its observed mean yields the stated expression.
\end{proof}
\end{revblock}

There are several ways to aggregate the cohort-specific ATTs to derive the overall ATT.  One approach is to use the average of the cohort-specific ATTs, weighted by the cohort distribution, as follows:
    $ATT_t = \sum_{g \in \mathcal{G} \setminus \{0\}; g \leq t } \Pr(G = g \mid G \neq 0, g \leq t )  ATT_{g,t},$
where the weights are given by the cohort distribution $\Pr(G = g \mid G \neq 0, g \leq t )$.  The choice of aggregation method depends on the research question and the context of the data. For further examples of alternative aggregation of cohort-specific ATTs and discussions on their corresponding weights for estimation, see \cite{CALLAWAY2021}.

\section{Equivalence of Transition Independence and Conditional Parallel Trends}\label{sec:supp-equivalence}
The following proposition states formally the equivalence recorded in
Remark~\crossref{remark:equivalence} of the main text.
\begin{proposition}\label{prop:supp-equivalence}
Suppose that Assumption \crossref{A-anticipation} holds. Then Assumption
\crossref{A-transition} holds if and only if Assumption
\crossref{A-conditional-trend} holds.
\end{proposition}
\begin{proof}
Assumption \crossref{A-anticipation} lets the observed pre-treatment path
stand in for the untreated path in the conditioning sets, on histories with
positive probability. Suppose that Assumption~\crossref{A-transition} holds.
Let $\bar{\bs{x}}^{(k)} \in \mathcal{X}$ denote the value of $\bs X_t(0)$ corresponding to $Y_t(0) = \bar y^{(k)}$, where the $k$-th element equals $1$ and all other elements equal $0$.
Then, $\Pr(\bs X_t(0)= \bar{\bs{x}}^{(k)}  \mid \bs X_{1}^{T_0}(0)=\bs x_{1}^{T_0},D=d)=  \Pr(X^{(k)}_t(0)= 1 \mid \bs X_{1}^{T_0}(0)=\bs x_{1}^{T_0},D=d)=  \E[X^{(k)}_t(0)\mid \bs X_{1}^{T_0}(0)=\bs x_{1}^{T_0},D=d]$. Therefore, Assumption \crossref{A-transition} implies that  $\mathbb{E}[\bs X_{t} (0)\mid \bs X_{1}^{T_0}(0)=\bs x_{1}^{T_0},D=0]=\mathbb{E}[\bs X_{t} (0)\mid  \bs X_{1}^{T_0}(0)=\bs x_{1}^{T_0},D=1]$ for all $t \geq T_0+1$, and Assumption \crossref{A-conditional-trend} is satisfied.

Conversely, suppose that Assumption \crossref{A-conditional-trend} holds. Then, for each $k = 1, \ldots, K$, $\mathbb{E}[ X^{(k)}_{t} (0)\mid \bs X_{1}^{T_0}(0)=\bs x_{1}^{T_0},D=0]=\mathbb{E}[X^{(k)}_{t} (0)\mid  \bs X_{1}^{T_0}(0)=\bs x_{1}^{T_0},D=1]$. Substituting $\mathbb{E}[ X^{(k)}_{t} (0)\mid \bs X_{1}^{T_0}(0)=\bs x_{1}^{T_0},D=d]= \Pr(\bs X_t(0)= \bar{\bs{x}}^{(k)}  \mid \bs X_{1}^{T_0}(0)=\bs x_{1}^{T_0},D=d)$ into both sides yields
$\Pr\left(\bs X_t(0)= \bar{\bs{x}}^{(k)}  \mid \bs X_{1}^{T_0}(0)=\bs x_{1}^{T_0},D=0\right)=\Pr\left(\bs X_t(0)=\bar{\bs{x}}^{(k)}  \mid \bs X_{1}^{T_0}(0)=\bs x_{1}^{T_0},D=1\right)$
for all $k = 1, \ldots, K$.
Since $\mathcal{X} = \{ \bar{\bs{x}}^{(1)}, \ldots,  \bar{\bs{x}}^{(K)}\}$, it follows that Assumption~\crossref{A-transition} holds.
\end{proof}

\section{Estimation Procedure and Weighted Bootstrap}\label{sec:estimation-procedure-details}

\paragraph*{Estimation Procedure.}  We provide below a detailed description of the two-stage estimator presented in the main text.  In the first stage, we estimate the latent transition probabilities using the EM algorithm; in the second stage, we compute the counterfactual expected untreated potential outcome for each latent type to derive the average treatment effect on the treated for post-treatment observations as follows:
    \begin{enumerate} 
        \setcounter{enumi}{-1}
        \item 
        Initiate the process with a given starting value, $\hat{\bs\psi}^{(0)} =(\hat{\bs\pi}^{(0)},\hat{\bs\varphi}^{1(0)},....,\hat{\bs\varphi}^{J(0)}) \in \Theta_{\bs\psi}$.
        \item Estimate $\psi$ by employing the following EM algorithm, starting with $s = 1$:
        \begin{enumerate}
            \item (E-step) Compute the posterior probabilities $\hat\tau_i^{j(s)}$ for each $i$ and $j$:
               $ \hat\tau_i^{j(s)} =  \frac{ \hat\pi^{j(s-1)} p_{\bs W}^j(\bs W_i ;\hat{ \bs \varphi}^{j(s-1)} ) }{ \sum_{k=1}^J \hat\pi^{k(s-1)} p_{\bs W }^k(\bs W_i; \hat{\bs \varphi}^{k(s-1)}) }.$
            \item (M-step) Update the mixture weights $\hat{\bs\pi}^{(s)}=(\pi^{1(s)},\ldots,\pi^{J(s)})^\top$ and type-specific transition probabilities
            \[
           \hat{\bs\varphi}^{j(s)}=\bigl\{ \{\hat p^{j(s)}_{\bs X_1(0),D}(\cdot,\cdot),\, \hat p^{j(s)}_{\bs X_t(0)|\bs X_{t-1}(0)}(\cdot|\cdot)\}_{t=2}^{T},\, \{\hat p^{j(s)}_{\bs X_t(1)|\bs X_{t-1}(1)}(\cdot|\cdot)\}_{t=T_0+1}^{T}\bigr\}
            \]
             for each latent type $j$ as
            \begin{align*}
              & \hat \pi^{j(s)} = \frac{1}{n} \sum_{i=1}^n \hat\tau_i^{j(s)},  \quad
               \hat p^{j(s)}_{\bs X_1(0),D}(\bs x_1,d) =\frac{ \sum_{i=1}^n \mathbf{1}\{\bs X_{i1} = \bs x_1, D_i=d\} \hat\tau_i^{j(s)}}{ \sum_{i=1}^n \hat\tau_i^{j(s)}},   \\ 
          &  \hat    p_{\bs X_t(0) \mid \bs X_{t-1}(0)}^{j(s)}(\bs x_t | \bs x_{t-1}) = \left\{
            \begin{array}{ll}  
  \frac{ \sum_{i=1}^n \mathbf{1}\{ \bs X_{it} = \bs x_t, \bs X_{it-1}= \bs x_{t-1}\} \hat\tau_i^{j(s)}  }{ \sum_{i=1}^n \mathbf{1}\{ \bs X_{it-1}= \bs x_{t-1}\} \hat\tau_i^{j(s)} } &\text{for $t=2,...,T_0$,}\\
          \frac{\sum_{i=1}^n \mathbf{1}\{  \bs X_{it} = \bs x_t, \bs X_{it-1}= \bs x_{t-1},D_i=0 \} \hat\tau_i^{j(s)} }{\sum_{i=1}^n \mathbf{1}\{  \bs X_{it-1}= \bs x_{t-1},D_i=0\} \hat\tau_i^{j(s)} }&\text{for $t=T_0+1,...,T$,}\\ 
            \end{array}\right.\\
           & \hat    p_{\bs X_t(1) \mid \bs X_{t-1}(1)}^{j(s)}(\bs x_t | \bs x_{t-1}) =  \frac{\sum_{i=1}^n \mathbf{1}\{  \bs X_{it} = \bs x_t,\bs X_{it-1}= \bs x_{t-1},D_i=1 \} \hat\tau_i^{j(s)}  }{\sum_{i=1}^n \mathbf{1}\{     \bs X_{it-1}= \bs x_{t-1},D_i=1\} \hat\tau_i^{j(s)} }\\
            &\hspace{12em}\text{for $t=T_0+1,\ldots,T$.}
            \end{align*} 
            \item Set $s = s + 1$ and repeat until convergence.
        \end{enumerate}
        \item  Given the first-step estimate, we estimate LTATTs and ATTs as in (\crossref{eq:LTATT-estimator}) and (\crossref{eq:ATT-estimator}), respectively.
            \end{enumerate}

\paragraph*{Weighted Bootstrap.}
To consistently estimate the asymptotic covariance matrix $\boldsymbol{V}$, we employ a nonparametric weighted bootstrap. 
Specifically, we generate $B$ bootstrap replications $\{\hat{\boldsymbol{\theta}}^{(b)}\}_{b=1}^B$ by reweighting each observation’s contribution in the estimation process as follows:  

\begin{enumerate}
    \item[{1.}] \textit{Generate weights.}  
    For each bootstrap replication $b = 1, \ldots, B$, draw $n$ i.i.d.\ positive weights 
    $\{\zeta_i^{(b)}\}_{i=1}^n$ from the standard exponential distribution, $\mathrm{Exp}(1)$.

    \item[{2.}] \textit{Weighted MLE.}  
    Re-estimate the parameter vector $\boldsymbol{\psi}$ by maximizing the weighted log-likelihood function  
       $ \sum_{i=1}^n \zeta_i^{(b)} \log p_{\boldsymbol{W}}(\boldsymbol{W}_i; \boldsymbol{\psi}),$
    implementing the weights $\zeta_i^{(b)}$ within the Expectation–Maximization (EM) algorithm.

    \item[{3.}] \textit{Labeling correction.}  
    After each weighted MLE, re-order the components of $\hat{\boldsymbol{\psi}}^{(b)}$ according to the identification constraint in Assumption~\crossref{A-sample}(c), ensuring consistent labeling of latent types across bootstrap replications.

    \item[{4.}] \textit{Weighted sample analogue.}  
    Compute the bootstrap estimate $\hat{\boldsymbol{\theta}}^{(b)}$ by replacing all summations in the sample-analogue estimators
    \crosseqref{eq:LTATT-estimator}, \crosseqref{eq:P-X0}, and \crosseqref{eq:E-Xt}
    with their weighted counterparts, incorporating both the random weights $\zeta_i^{(b)}$ and the estimated posterior probabilities $\hat{\tau}_i^{j(b)}$ obtained from $\hat{\boldsymbol{\psi}}^{(b)}$.

    \item[{5.}] \textit{Covariance estimation.}  
    Estimate the asymptotic covariance matrix $\boldsymbol{V}$ as
       ${ \hat{\boldsymbol{V}}}/{n}
        = \frac{1}{B - 1}
        \sum_{b=1}^B
        \left(
            \hat{\boldsymbol{\theta}}^{(b)} - \bar{\boldsymbol{\theta}}
        \right)
        \left(
            \hat{\boldsymbol{\theta}}^{(b)} - \bar{\boldsymbol{\theta}}
        \right)^{\!\top}$ with
       $ \bar{\boldsymbol{\theta}}
        = \frac{1}{B}
        \sum_{b=1}^B \hat{\boldsymbol{\theta}}^{(b)}.$
\end{enumerate}

\paragraph*{Cluster-Level Weighting.}
When the data exhibit cluster-level dependence (e.g., individuals within states, or firms within industries), the bootstrap weights are drawn at the cluster level rather than the individual level. Let $C$ denote the number of clusters. For each bootstrap replication $b$, we draw independent exponential weights $\{w_c^{(b)}\}_{c=1}^C$ with $w_c^{(b)} \sim \mathrm{Exp}(1)$, normalize them to $\tilde{w}_c^{(b)} = w_c^{(b)} / \sum_{c'=1}^C w_{c'}^{(b)}$, and assign each unit $i$ in cluster $c$ the weight $\zeta_i^{(b)} = \tilde{w}_c^{(b)}$. This weighted bootstrap procedure at the cluster level follows \cite{CameronGelbachMiller2008} and \cite{MacKinnonWebb2017}, preserving the within-cluster correlation structure while allowing for inference that is robust to cluster-level heterogeneity.

\phantomsection
\paragraph*{Uniform Confidence Intervals.}\label{para:uniform-ci}
We construct uniform (simultaneous) confidence intervals for sequences of treatment effect estimates $\{\hat{\theta}_t\}_{t=1}^T$ using a studentized bootstrap procedure. Unlike pointwise confidence intervals that control coverage at each individual time period, uniform confidence intervals provide simultaneous coverage across all periods:
$\Pr\left(\theta_t \in \text{CI}_t \text{ for all } t = 1, \ldots, T\right) \geq 1 - \alpha.$
Let $\hat{\theta}_t$ denote the point estimate for period $t$. Given $B$ bootstrap estimates $\{\hat{\theta}^{(b)}_t\}_{b=1}^B$ from the weighted bootstrap procedure:
\begin{enumerate}
    \item Compute the bootstrap standard error for each period:
    $\hat{\sigma}_t = \sqrt{\frac{1}{B-1} \sum_{b=1}^{B} \left(\hat{\theta}^{(b)}_t - \bar{\theta}_t\right)^2},$
    where $\bar{\theta}_t = B^{-1} \sum_{b=1}^{B} \hat{\theta}^{(b)}_t$.

    \item For each bootstrap sample $b$ and period $t$, compute the studentized statistic:
 $   t^{(b)}_t = \frac{\hat{\theta}^{(b)}_t - \hat{\theta}_t}{\hat{\sigma}_t}.$

    \item For each bootstrap sample, compute the supremum of the absolute studentized statistics:
  $  M^{(b)} = \max_{t \in \{1, \ldots, T\}} |t^{(b)}_t|.$

    \item Obtain the critical value $c_{1-\alpha}$ as the $(1-\alpha)$-quantile of $\{M^{(1)}, \ldots, M^{(B)}\}$.

    \item Construct the uniform confidence interval for each $t$:
   $ \left[\hat{\theta}_t - c_{1-\alpha} \cdot \hat{\sigma}_t, \hat{\theta}_t + c_{1-\alpha} \cdot \hat{\sigma}_t\right].$
\end{enumerate}
The key distinction from pointwise intervals is that uniform intervals use a single critical value $c_{1-\alpha}$ derived from the joint distribution of the supremum statistic, rather than period-specific quantiles. This accounts for the multiple testing problem inherent in examining effects across many time periods. When the uniform confidence band lies entirely above (or below) zero for a range of periods, we can conclude that the treatment effect is significantly positive (or negative) for all those periods simultaneously at the $(1-\alpha)$ confidence level.

\paragraph*{Computational details.}
In the empirical applications, the EM algorithm is initialized using a two-stage multistart procedure to reduce the risk of converging to a poor local optimum. In the short-run stage, a large pool of initial values is run for a small number of EM iterations; the most promising candidates are then run to convergence in the long-run stage. The initialization pool combines two sources. First, we draw a set of random starting values as Dirichlet draws over the outcome and transition simplices with a mix of concentration parameters, so that the pool spans both sparse and smooth distributions; in the empirical applications we use a large fixed pool of $20{,}000$ random starts for every latent dimension $J$ in all three applications. Second, we construct a set of structured, data-driven warm starts: $5{,}000$ random Dirichlet partitions of the units into $J$ groups, each converted into a starting value by a single maximization step on the implied group assignment. The random and structured starts are ranked together by their short-run likelihood, and the top $2{,}000$ candidates are run to convergence in the long-run stage. Standard errors are obtained via the weighted bootstrap with $B = 500$ replications; each replication refines the full-sample estimate on the reweighted data rather than repeating the global multistart search.

\begin{revblock}
\section{Consistency of BIC Selection}\label{sec:bic-postsel}

This section shows that the BIC rule of
Section~\ref{sec:estimation-procedure-details} recovers the number of latent
types fast enough that inference treating the selected number as known
remains asymptotically valid. Let \(p^*\) denote the true PMF of \(\bs W\)
on \(\mathcal{W}\), let \(M := |\mathcal{W}| = 2K^{T}\), and let \(J_0\) be
the smallest \(J\) for which \(p^*\) is a \(J\)-component mixture whose
components have the structural form of Assumptions~\crossref{A-sample}(a)
and~\crossref{A-Markov}: a type-specific PMF for \((\bs X_1, D)\) and
type-specific first-order Markov transitions, shared across treatment arms
before \(T_0\) and arm-specific afterward (absent this restriction on the
components, minimality would be vacuous). Each \(\ell_J\) entering
\(\mathrm{BIC}(J)\) is maximized over the space \(\Theta_J\), with component
PMF entries in \([\epsilon, 1-\epsilon]\) for the \(\epsilon\) of
Assumption~\crossref{A-MLE}(a) and mixture weights unrestricted on the
simplex (the two-sided weight bounds of Assumption~\crossref{A-MLE}(a)
would empty the one-component candidate, where \(\pi^1 = 1 > 1-\epsilon\)).
Throughout, \(\widehat{J}_{\mathrm{BIC}}\) denotes the smallest minimizer
of \(\mathrm{BIC}(J)\) over \(\{1,\ldots,J_{\max}\}\).
Write
\(\hat{\bs\mu}^{\mathrm{ATT}}(J) :=
(\hat{\bs\mu}^{ATT}_{T_0+1},\ldots,\hat{\bs\mu}^{ATT}_{T})\) for the ATT
estimator \crosseqref{eq:ATT-estimator} computed from the estimated $J$-component mixture model,
with population counterpart \(\bs\mu^{\mathrm{ATT}}\).

\begin{proposition}[BIC selection of the number of latent types]\label{prop:bic-postsel}
Suppose Assumptions~\crossref{A-sample}(a),(c) and~\crossref{A-Markov} hold,
Assumption~\crossref{A-MLE}(a),(b) holds at \(J = J_0\) (its weight bounds
read as void when \(J_0 = 1\), where \(\pi^1 = 1\)), and
\(J_0 \le J_{\max}\). Then:
(i) \(\Pr(\widehat{J}_{\mathrm{BIC}} = J_0) \to 1\).
(ii) \(\Pr(\widehat{J}_{\mathrm{BIC}} < J_0) \le C e^{-cn}\) for some
constants \(c, C > 0\), and
\(\Pr(\widehat{J}_{\mathrm{BIC}} > J_0) \le n^{-\Delta/2 + o(1)}\), where
\(\Delta := \mathrm{df}(J_0+1) - \mathrm{df}(J_0) = 1 + (2K-1) +
n_{\mathrm{pre}}K(K-1) + 2 n_{\mathrm{post}}K(K-1) \ge 4\); in particular
\(\Pr(\widehat{J}_{\mathrm{BIC}} \ne J_0) = o(n^{-1/2})\). 
(iii) If, in addition, the hypotheses of
Proposition~\crossref{P-estimation} hold at \(J = J_0\) with
Assumption~\crossref{A-sample}(b) replaced by selection of
\(\widehat{J}_{\mathrm{BIC}}\) over \(\{1,\ldots,J_{\max}\}\), then the laws
of \(\hat{\bs\mu}^{\mathrm{ATT}}(\widehat{J}_{\mathrm{BIC}})\) and
\(\hat{\bs\mu}^{\mathrm{ATT}}(J_0)\) differ by at most
\(\Pr(\widehat{J}_{\mathrm{BIC}} \ne J_0) = o(n^{-1/2})\) in total variation;
hence \(\sqrt{n}\{\hat{\bs\mu}^{\mathrm{ATT}}(\widehat{J}_{\mathrm{BIC}}) -
\bs\mu^{\mathrm{ATT}}\} \Rightarrow \mathcal{N}(0, \bs V^{\mathrm{ATT}})\),
where \(\bs V^{\mathrm{ATT}}\) is the ATT block of \(\bs V\) in
Proposition~\crossref{P-estimation}, and any confidence interval for
\(\bs\mu^{\mathrm{ATT}}\) with asymptotic coverage \(1-\alpha\) under known
\(J_0\), including the fixed-model weighted-bootstrap intervals of
Section~\ref{sec:estimation-procedure-details}, retains asymptotic coverage
\(1-\alpha\) at \(\widehat{J}_{\mathrm{BIC}}\).
\end{proposition}

The proof is given in Section~\ref{subsec:prop-bic-postsel}.

Two remarks connect this result to the literature. First, our penalty
\(\tfrac{1}{2}\mathrm{df}(J)\log n\) satisfies the penalty condition (C1) of
\citet{keribin00sankhya} for every \(n \ge 2\) (at \(n = 1\) the
\(\log n\) penalty vanishes), as \citeauthor{keribin00sankhya} notes for BIC; the
identifiability and empirical-process conditions there ((Id), (P1)--(P4)) are tailored to smooth
families such as Gaussian and Poisson mixtures, so rather than verifying them
for our finite-support components we prove consistency directly, with the
deviance bound above playing their role, while Theorem~4 of
\citet{leroux92as} would rule out under-estimation only. Second, the
guarantees are pointwise in the data-generating process rather than uniform:
along sequences with vanishing separation between \(J_0 - 1\) and \(J_0\)
components, mis-selection need not vanish, and in the linear-model
benchmark of \citet{leeb2005et} the maximal scaled risk of
post-model-selection estimators diverges; the same caution applies to any
consistent selection rule, ours included. 
\end{revblock}

\section{Construction of Type-Specific Transition Probabilities}\label{sec:empirical-transition-probability-construction}

Testing for transition independence across pre-treatment periods (Remark~\crossref{remark:testing-for-transition-independence-Z}) in the presence of latent heterogeneity in transition dynamics requires estimating type-specific transition probabilities for both treated and control units during pre-treatment periods.

In this section, we outline the procedure for constructing the type-specific transition probabilities across pre-treatment periods for each treatment status and latent type. First, we estimate the parameters of the latent models using the observed outcomes and treatment status via the EM algorithm described in Section~\ref{sec:estimation-procedure-details}.  Next, we compute the type-specific probabilities for each unit using the estimated parameters. Lastly, we estimate the type-specific transition probabilities for each latent type by weighting pre-treatment observations using the estimated type-specific probabilities. Specifically, the procedure is as follows:

\begin{enumerate}
  \item Estimate the model parameters for each latent type $\hat{\boldsymbol{\psi}}$ via the EM algorithm described in Section~\ref{sec:estimation-procedure-details}.
  \item Obtain the estimated type-specific probabilities $\hat \tau_i^j$ for each unit $i$ and latent type $j$ from \crosseqref{posterior-estimate} from the estimated parameters $\hat{\boldsymbol{\psi}}$.
  \item For each latent type $j$ and treatment status $d \in \{0,1\}$, estimate the type-specific transition probabilities across pre-treatment periods $t = 2, \ldots, T_0$ as follows:\\
    $  \hat{p}_{\boldsymbol{X}_t(0) | \boldsymbol{X}_{t-1}(0), D}^{j}(\boldsymbol{x}_t | \boldsymbol{x}_{t-1}, d) 
      = 
      \frac{
          \sum_{i=1}^{n} 
          \mathbf{1}\{\boldsymbol{X}_{it} = \boldsymbol{x}_t, \boldsymbol{X}_{it-1} = \boldsymbol{x}_{t-1}, D_i = d\}
          \hat{\tau}_{ij}
      }{
          \sum_{i=1}^{n} 
          \mathbf{1}\{\boldsymbol{X}_{it-1} = \bold{x}_{t-1}, D_i = d\}
          \hat{\tau}_{ij} 
      }.$
\end{enumerate}

\section{Type-Specific Flow Decomposition}

\begin{corollary}[Type-specific flow decomposition]\label{cor:flow-decomposition-by-type}
Under Assumptions \crossref{A-transition-Z} and  \crossref{A-Markov}, the LTATT on the $k$th outcome for latent type $j$ decomposes as
\begin{equation}\label{eq:flow-decomposition-by-type}
\begin{aligned}
&\E\!\big[X_{t}^{(k)}(1)-X_{t}^{(k)}(0)\mid D=1, Z=j\big] \\
&=
\sum_{y \neq \bar{y}^{(k)} }
\underbrace{\Big\{\Pr(Y_t=\bar{y}^{(k)}\mid Y_{T_0}\!=\!y, D\!=\!1, Z\!=\!j)
-\Pr(Y_{t}=\bar{y}^{(k)}\mid Y_{T_0}\!=\!y, D\!=\!0, Z\!=\!j)\Big\}}_{\text{inflow CATT from } y \text{ for type } j}\\
&\hspace{12em}\times\;
\Pr(Y_{T_0}=y\mid D=1, Z=j)\\
&\quad
-
\sum_{y \neq \bar{y}^{(k)}}
\underbrace{\Big\{\Pr(Y_{t}=y\mid Y_{T_0}\!=\!\bar{y}^{(k)}, D\!=\!1, Z\!=\!j)
-\Pr(Y_{t}=y\mid Y_{T_0}\!=\!\bar{y}^{(k)}, D\!=\!0, Z\!=\!j)\Big\}}_{\text{outflow CATT to } y \text{ for type } j}\\
&\hspace{12em}\times\;
\Pr(Y_{T_0}=\bar{y}^{(k)}\mid D=1, Z=j).
\end{aligned}
\end{equation}
To connect the type-specific terms to the aggregate inflow and outflow effects in \crosseqref{eq:flow-decomposition}, define the type-specific channel contributions
\begin{align*}
I_{j,t}^{y\to k}
&=
\Big\{\Pr(Y_t=\bar{y}^{(k)}\mid Y_{T_0}=y, D=1, Z=j)
-\Pr(Y_t=\bar{y}^{(k)}\mid Y_{T_0}=y, D=0, Z=j)\Big\}\\
&\hspace{4em}\times \Pr(Y_{T_0}=y\mid D=1,Z=j),\qquad y\neq \bar{y}^{(k)},\\
O_{j,t}^{k\to y}
&=
-\Big\{\Pr(Y_t=y\mid Y_{T_0}=\bar{y}^{(k)}, D=1, Z=j)
-\Pr(Y_t=y\mid Y_{T_0}=\bar{y}^{(k)}, D=0, Z=j)\Big\}\\
&\hspace{4em}\times \Pr(Y_{T_0}=\bar{y}^{(k)}\mid D=1,Z=j),\qquad y\neq \bar{y}^{(k)}.
\end{align*}
The outflow contribution $O_{j,t}^{k\to y}$ includes the minus sign from \eqref{eq:flow-decomposition-by-type}, so positive values increase the focal outcome and negative values reduce it. The aggregate channel contributions are therefore
$I_t^{y\to k}=\sum_{j=1}^{J}\Pr(Z=j\mid D=1) I_{j,t}^{y\to k}$ and 
$O_t^{k\to y}=\sum_{j=1}^{J}\Pr(Z=j\mid D=1) O_{j,t}^{k\to y}.$
These aggregate channels generalize the main-text flow decomposition, reducing to it under marginal transition independence (or when $J=1$). For $J>1$, the control part of each channel is the identified treated-type posterior mixture implied by the latent-type model, so the channels should be read as latent-type counterfactual contributions to $\E[X_t^{(k)}(1)-X_t^{(k)}(0)\mid D=1]$ rather than as the raw marginal control conditional in \crosseqref{eq:flow-decomposition}. They satisfy
$\E\!\big[X_t^{(k)}(1)-X_t^{(k)}(0)\mid D=1\big]
=
\sum_{y\neq \bar{y}^{(k)}} I_t^{y\to k}
+
\sum_{y\neq \bar{y}^{(k)}} O_t^{k\to y}.$
Equivalently, the aggregate ATT on the $k$th outcome is recovered by aggregating over types:
\begin{equation}\label{eq:flow-decomposition-aggregate}
\E\!\big[X_{t}^{(k)}(1)-X_{t}^{(k)}(0)\mid D=1\big]
= \sum_{j=1}^{J} \Pr(Z=j\mid D=1)\,\E\!\big[X_{t}^{(k)}(1)-X_{t}^{(k)}(0)\mid D=1, Z=j\big].
\end{equation}
\end{corollary}

\begin{proof}
Fix type $j$ and outcome category $k$. By Proposition~\crossref{P-1}, the type-specific ATT on the $k$th outcome is
$\E\!\big[X_t^{(k)}(1) - X_t^{(k)}(0) \mid D=1, Z=j\big] = \sum_{\bs x_{T_0}} \Pr(\bs X_{T_0} = \bs x_{T_0} \mid D=1, Z=j) \cdot \Delta_t^{j,k}(\bs x_{T_0}),$
where $\Delta_t^{j,k}(\bs x_{T_0}) := \E[X_t^{(k)} \mid \bs X_{T_0} = \bs x_{T_0}, D=1, Z=j] - \E[X_t^{(k)} \mid \bs X_{T_0} = \bs x_{T_0}, D=0, Z=j]$.
Under the first-order Markov assumption, conditioning on $\bs X_{T_0}$ reduces to $Y_{T_0}$. Partitioning the summation over $Y_{T_0}$ into $\{Y_{T_0} = \bar{y}^{(k)}\}$ (outflows) and $\{Y_{T_0} \neq \bar{y}^{(k)}\}$ (inflows) yields the stated type-specific decomposition. Multiplying each type-specific channel by $\Pr(Z=j\mid D=1)$ and summing over $j$ gives the aggregate channel contributions above. Equation~\eqref{eq:flow-decomposition-aggregate} follows from the law of iterated expectations applied to $\E[\cdot \mid D=1] = \sum_j \Pr(Z=j\mid D=1)\E[\cdot \mid D=1, Z=j]$.
\end{proof}

\vspace*{\fill}
\newpage

\makeatletter
\setlength{\@fptop}{0pt}
\setlength{\@fpsep}{18pt}
\setlength{\@fpbot}{0pt plus 1fil}
\makeatother
\section{Supplementary Figures}
\begin{figure}[!ht]
    \centering
    \caption{Counterfactual Complaint Rates from the Dodd-Frank Act.}\label{fig:complaint-counterfactual-all}
    \includegraphics[width=0.8\textwidth]{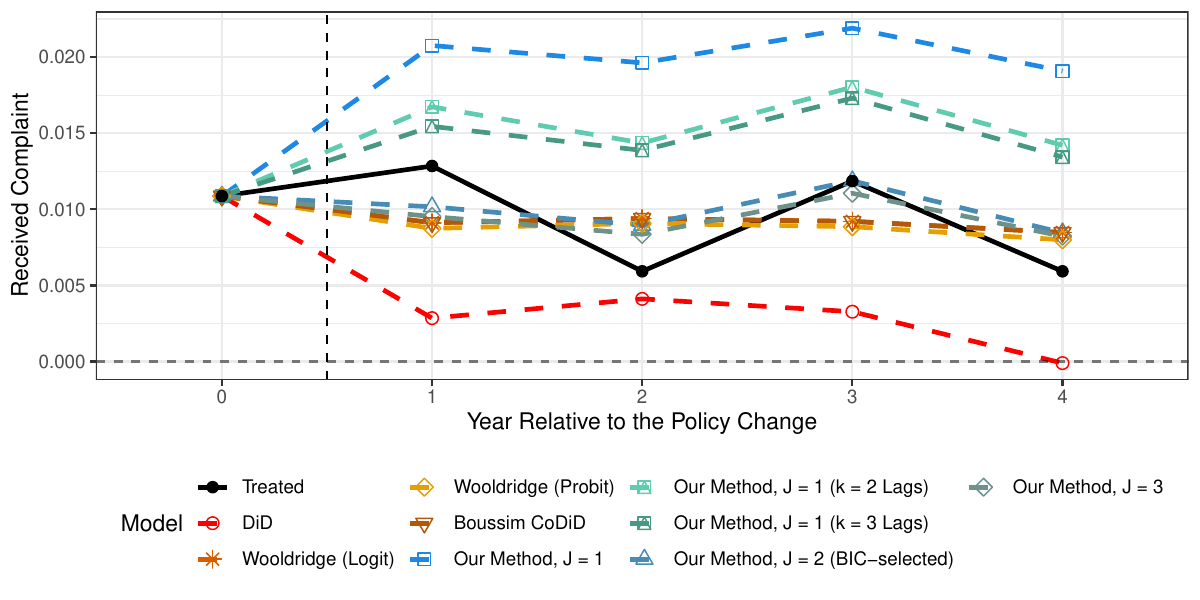}
    \fignote{Average counterfactual complaint rates for treated units in the absence of treatment, using the data of \cite{charoenwong2019does}. The black line is the observed treated path; the counterfactual restrictions and markers are as in Figure~\crossref{fig:ada-att-simple}. Transition-based counterfactuals with $J \in \{1, 2, 3\}$ latent types are constructed from estimated control transition matrices and treated pre-reform distributions; green markers are lag-augmented $J=1$ variants conditioning on $k = 2$ and $k = 3$ outcome lags.}
\end{figure}
\begin{figure}[!ht]
    \centering
    \caption{ATT Estimates of the ADA on Employment.}\label{fig:ada-att-postsel}
    \includegraphics[width=0.8\textwidth]{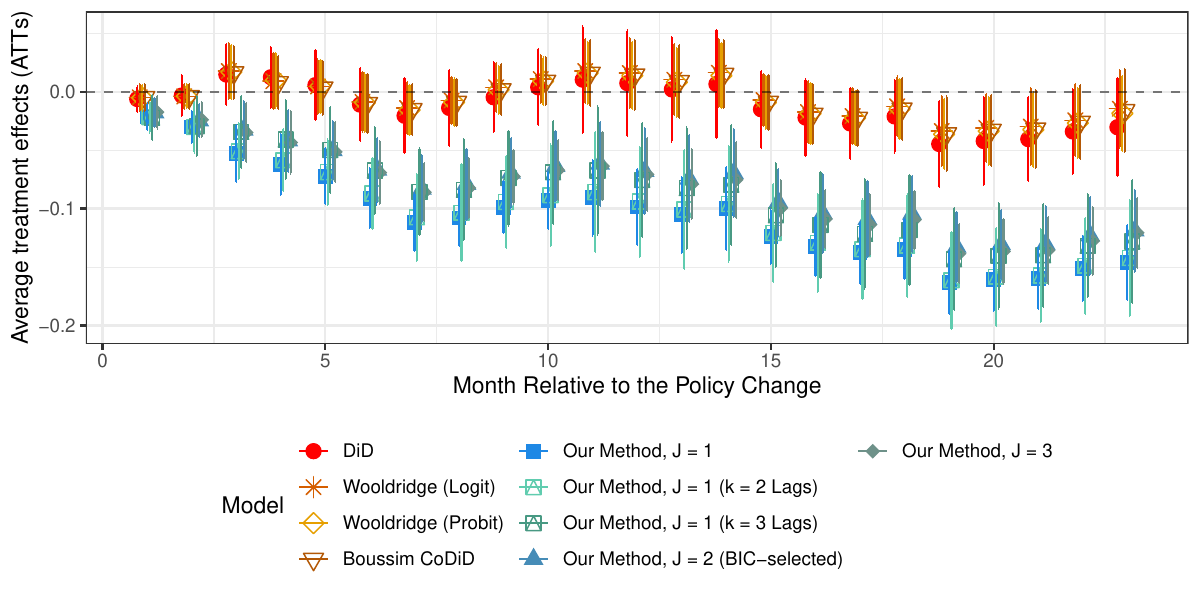}
    \fignote{Figure~\crossref{fig:ada-att-all} with the BIC-selected model ($J_{\mathrm{bic}} = 2$) labeled in the legend. Vertical bars are 95\% pointwise bootstrap confidence intervals for the individual model specifications.}
\end{figure}

\begin{figure}[p]
    \centering
    \caption{ATT Estimates of the Dodd-Frank Act on Complaint Rates.}\label{fig:complaint-att-postsel}
    \includegraphics[width=0.8\textwidth]{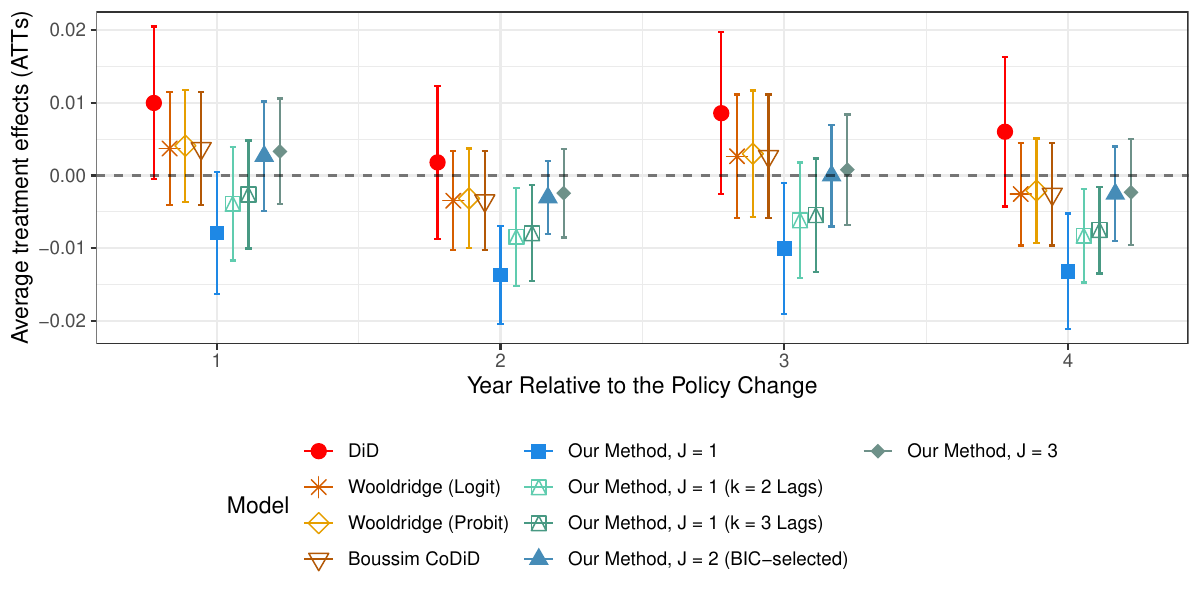}
    \fignote{Figure~\crossref{fig:complaint-att-all} with the BIC-selected model ($J_{\mathrm{bic}} = 2$) labeled in the legend. Vertical bars are 95\% pointwise bootstrap confidence intervals for the individual model specifications.}
\end{figure}

\begin{figure}[p]
    \centering
    \caption{ATT Estimates of the Norwegian Law Reform on Annual Patenting Rates: Full Comparison.}\label{fig:patenting-att-postsel}
    \includegraphics[width=0.8\textwidth]{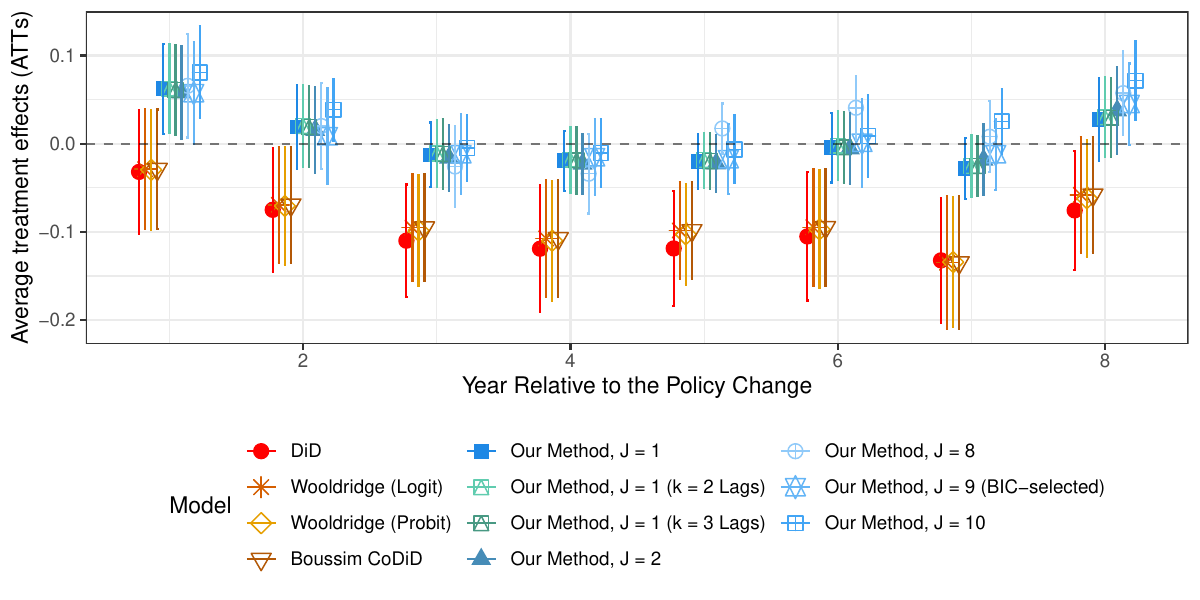}
    \fignote{Figure~\crossref{fig:patenting-att-all} extended to $J \in \{1, 2, 8, 9, 10\}$, with the BIC-selected model ($J_{\mathrm{bic}} = 9$) labeled in the legend. Vertical bars are 95\% pointwise weighted-bootstrap confidence intervals for the individual model specifications; the high-order intervals use refine-only bootstrap estimates.}
\end{figure}

\begin{figure}[p]
    \centering
    \caption{Labor Force Status Trends Before and After the ADA of 1990.}\label{fig:ada-labor-force-status-trends}
    \includegraphics[width=0.8\textwidth]{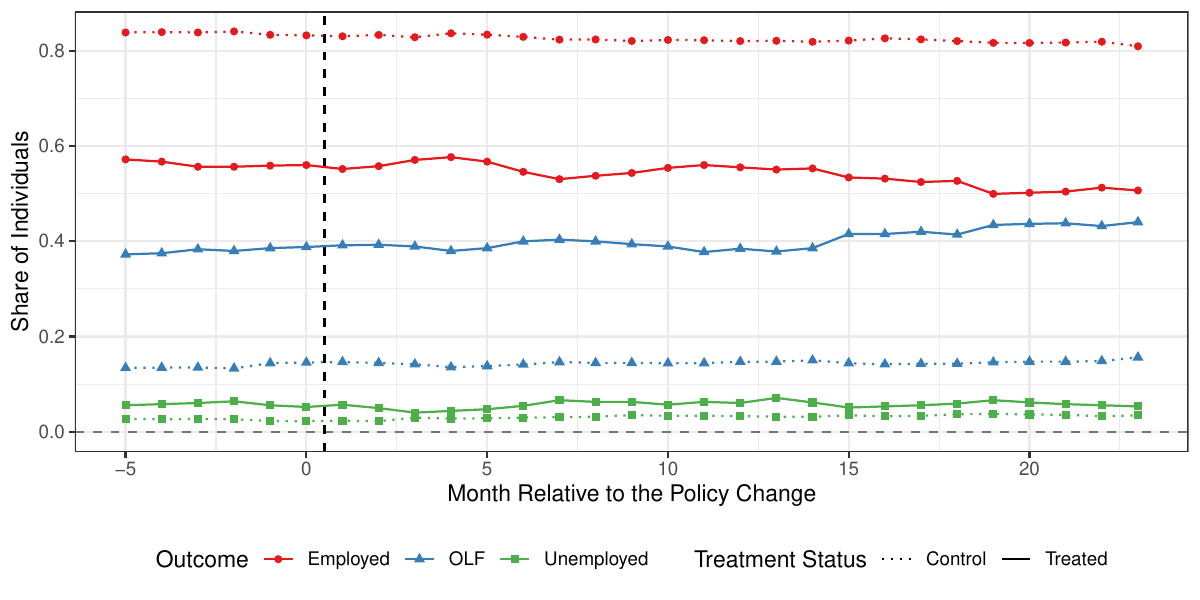}
    \fignote{Share of individuals in each labor force status (employed, unemployed, and out-of-labor-force) for individuals with disabilities (treated) and without disabilities (control) over January 1990 to May 1992. Period zero is June 1990, the last month before the ADA's July 1990 enactment, marked by the dashed line.}
\end{figure}

\begin{figure}[p]
    \centering
    \caption{Differences in Labor Force Transition Probabilities Before the ADA ($J = 1$).}\label{fig:ada-employment-transitions-J2}
    \includegraphics[width=0.8\textwidth]{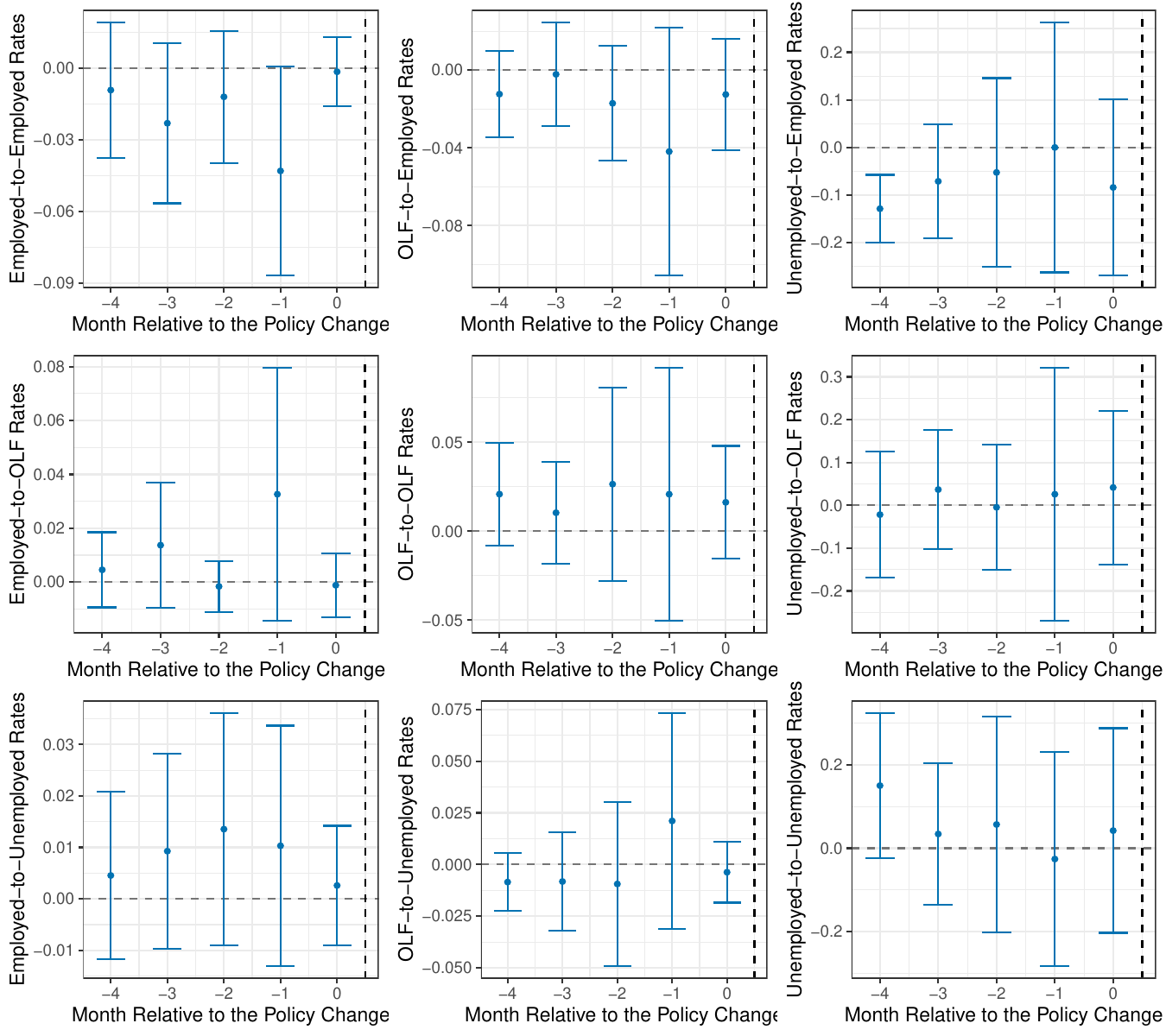}
    \fignote{Differences in monthly labor force transition probabilities between individuals with disabilities (treated) and without disabilities (control) before the ADA, one panel per transition among employed, unemployed, and out-of-labor-force (OLF) states. Period zero is the last pre-treatment period. Vertical bars are 95\% bootstrap uniform confidence intervals across periods within each transition pair.}
\end{figure}

\begin{figure}[p]
    \centering
    \caption{Differences in Labor Force Transition Probabilities Before the ADA ($J = 3$).}\label{fig:ada-employment-transitions-J3}
    \includegraphics[width=0.8\textwidth]{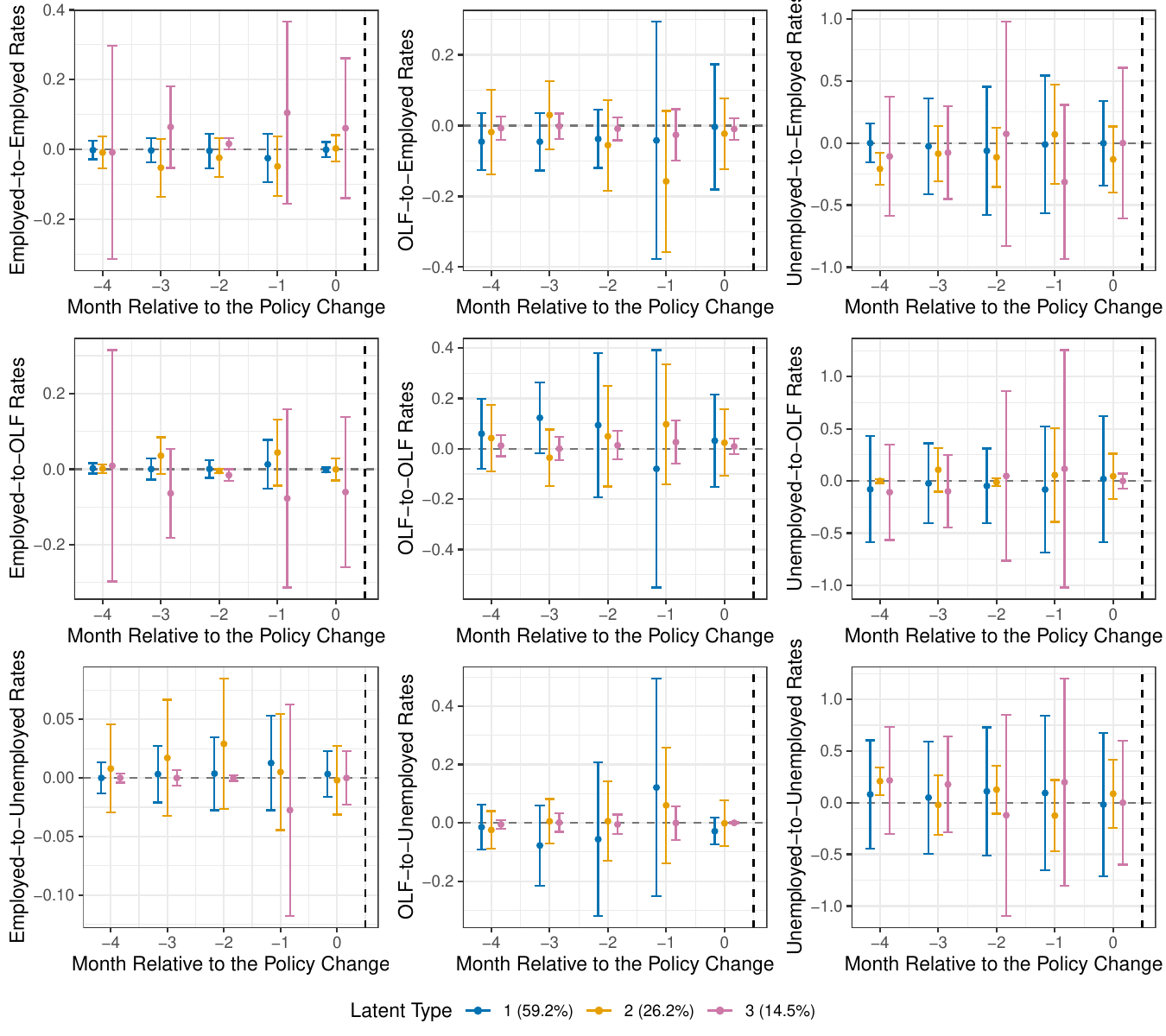}
    \fignote{Differences in monthly labor force transition probabilities between individuals with disabilities (treated) and without disabilities (control) before the ADA, by latent type (colors), one panel per transition among employed, unemployed, and out-of-labor-force (OLF) states. Percentages in the legend are estimated population shares of each type. Period zero is the last pre-treatment period. Vertical bars are 95\% bootstrap uniform confidence intervals across periods within each transition pair and latent type, $j \in \{1, 2, 3\}$.}
\end{figure}

\FloatBarrier
\bibliography{dmlmixture}

\end{document}